\documentclass[11pt,reqno,twoside]{amsart}
\usepackage{graphicx}
\usepackage{amsmath, amssymb, amsthm, bm}
\usepackage{mathrsfs, extarrows}

\usepackage{float}
\usepackage{caption}
\usepackage{subcaption} 
 
\usepackage{tabularx, booktabs, array}
\usepackage{xcolor}
\usepackage{hyperref}

\usepackage{paralist, verbatim, epstopdf}

\usepackage[asymmetric,top=2.5cm,bottom=2.8cm,left=2.38cm,right=2.38cm]{geometry}
\hypersetup{
	colorlinks,
	linkcolor={blue!80!black},
	urlcolor={blue!80!black},
	citecolor={blue!80!black}
}

\allowdisplaybreaks[4]

\newtheorem{thm}{Theorem}[section]

\newtheorem{lem}{Lemma}[section]

\theoremstyle{definition}
\newtheorem{defn}{Definition}[section]
\theoremstyle{remark}
\newtheorem{rem}{Remark}[section]
\theoremstyle{definition}
\newtheorem{exm}{Example}
\numberwithin{equation}{section}

\newcommand{\rmd}{\mathrm{d}}

\newcommand\rmi{{\mathrm{i}}}

\newcommand{\bmf}[1]{{\mathbf{#1}}}

\title[ Stress concentration via quasi-Minnaert resonance in bubble-elastic  ]{ Stress concentration via quasi-Minnaert resonance in bubble-elastic structures and applications }

\author{Ruixiang Tang}
\address{School of Mathematics, Jilin University, Changchun 130012, China.}
\email{tangrx97@gmail.com; tangrx23@mails.jlu.edu.cn}

\author{Huaian Diao}
\address{School of Mathematics, Jilin University and Key Laboratory of Symbolic Computation and Knowledge Engineering of Ministry of Education, Changchun, Jilin, China.}
\email{diao@jlu.edu.cn; hadiao@gmail.com}

\author{Hongyu Liu}
\address{Department of Mathematics, City University of Hong Kong, Kowloon, Hong Kong SAR, China.}
\email{hongyu.liuip@gmail.com; hongyliu@cityu.edu.hk}

\author{Weisheng Zhou}
\address{School of Mathematics, Jilin University, Changchun 130012, China.}
\email{zhouws24@mails.jlu.edu.cn;wszhou1211@163.com}

\date{} 
\begin{document}
	
	\maketitle
  	\begin{abstract}

Stress concentration in bubble-elastic scattering scenarios has significant applications in engineering blasting and medical treatments. This study provides a comprehensive mathematical analysis of stress concentration in bubbly-elastic structures, induced by the quasi-Minnaert resonance. The quasi-Minnaert resonance manifests as two distinct wave patterns near the bubble's boundary: boundary localization and high-oscillation phenomena. We demonstrate how to leverage the quasi-Minnaert resonance to induce stress concentration in the elastic total wave field near the air bubble's boundary by appropriately selecting the incident elastic wave and high-contrast structure. The interaction between the air bubble and the elastic background couples two physical wave fields-acoustic and elastic waves-across the bubble's boundary. The intricate transmission conditions, combined with the scalar nature of acoustic waves and the vectorial nature of elastic waves, present significant analytical challenges. To address these, we employ layer potential theory and asymptotic analysis to rigorously establish the stress concentration and quasi-Minnaert resonance phenomena in a radially geometry bubble-elastic model. Extensive numerical experiments are conducted to demonstrate the stress concentration phenomenon alongside quasi-Minnaert resonance for various bubble geometries, including a unit disk, a corner domain, an apple-shaped domain in $\mathbb{R}^2$, and a ball in $\mathbb{R}^3$. The findings of this study enhance the understanding of stress concentration mechanisms and their applications in engineering blasting and medical therapies.

		
	\end{abstract}
		

	\section{Introduction}

  In this paper, we investigate the stress concentration phenomenon in a bubble embedded within a soft elastic material in the subwavelength regime. This stress concentration arises from the resonance associated with the bubbly-elastic structure. The resonance of bubbly-elastic structures within elastic materials has been extensively studied \cite{CTL12, CTL15, LLZ, PM73}. Leveraging these resonant properties, bubbly-elastic structures have been employed in the development of novel metamaterials, including bubble phononic crystals \cite{LB09}, super-absorption applications \cite{LPLF}, and the mitigation of underwater sound transmission \cite{CTL12}. The bubble immersed in an elastic background undergoes shape deformations and, during its expansion and collapse, induces shear stresses on adjacent structures \cite{G10, OSK, ZWH}. During this scattering process, stress concentration manifests in the external total wave field \cite{LB09, LLH06, ROO, ZWH}. This phenomenon has practical applications in medical diagnostics, such as detecting and characterizing kidney stones \cite{SAB}. Despite numerous physical experiments and medical engineering applications of stress concentration in bubbly-elastic structures, the mathematical analysis of this phenomenon remains largely unexplored. In this paper, we provide a rigorous mathematical investigation of stress concentration in bubbly-elastic structures.

  This study examines the stress concentration phenomenon in the external total field of bubble-elastic structures, utilizing the \emph{quasi-Minnaert resonance} under high density contrast conditions. We demonstrate that the quasi-Minnaert resonance induces stress concentration in the external elastic wave field near the bubble's boundary. Furthermore, we establish that the quasi-Minnaert resonance frequency forms a continuous spectrum, in contrast to the discrete Minnaert resonance frequency analyzed in \cite{AFGL, LLZ, LZ23}.  The Minnaert resonance, a significant low-frequency phenomenon, arises in high-contrast physical configurations, such as a bubble immersed in a liquid, within acoustic contexts \cite{M1933}. The Minnaert resonance in acoustic settings was first rigorously analyzed in \cite{AFGL}, which established its dependence on high-contrast parameters. For bubbles embedded in soft elastic media, \cite{LLZ} derives a relationship between the Minnaert resonance frequency and the density contrast parameters between the bubble and the surrounding medium. A comparable relationship for hard inclusions in soft elastic media is developed in \cite{LZ23}. Recent studies further demonstrate that bubbles, leveraging Minnaert resonance as contrast agents, enable the reconstruction of material properties such as mass density and bulk modulus \cite{ACC, DGS}.

  The stress concentration phenomenon is closely tied to the boundary localization and high-oscillation behavior of the generated wave field near the bubble's boundary. The quasi-Minnaert resonance manifests as two distinct wave patterns near the boundary: boundary localization and high oscillation phenomena. Recent studies \cite{ABD24, ABU25, ACL22, ADH24, AFH18, AHR25, DBH25} have extensively investigated the boundary localization of wave fields under non-Hermitian resonances in high-contrast materials. In this paper, we provide a rigorous analysis of stress concentration induced by the quasi-Minnaert resonance. The quasi-Minnaert resonance in acoustic and elastic scattering was recently introduced in \cite{DTL, ZDL}. Here, we demonstrate how to leverage the quasi-Minnaert resonance to induce stress concentration in the elastic total wave field near the air bubble's boundary by appropriately selecting the incident elastic wave and high-contrast structure. The interaction between the air bubble and the elastic background couples two physical wave fields—acoustic and elastic waves—across the bubble's boundary. The intricate transmission conditions, combined with the scalar nature of acoustic waves and the vectorial nature of elastic waves, pose significant analytical challenges. To address these, we employ layer potential theory and asymptotic analysis to rigorously establish the stress concentration and quasi-Minnaert resonance phenomena in a radially symmetric bubble-elastic model.  Extensive numerical experiments are conducted to validate the theoretical findings of this paper. Stress concentration and quasi-Minnaert resonance are demonstrated through several numerical examples involving various bubble shapes in $\mathbb{R}^N$ ($N=2, 3$), achieved by carefully selecting incident waves.

\medskip
	
      The main contributions of this paper can be summarized in several key points:

	 \begin{enumerate}
	 	\item[(i)]

       We establish a sharp quantitative lower bound for the stress in a neighborhood of the bubble's boundary in Theorem~\ref{thm:Eu definition in thm}, as presented in \eqref{eq:Eus th4.2}. This lower bound reveals an intricate relationship between the stress of the exterior total field and several key parameters: the Lam\'e parameters, the index $n$ associated with a tuned incident wave, the comparison ratio $\tau$ of two different wave speeds in the acoustic and elastic media, and the incident wave number. From \eqref{eq:Eus th4.2}, we deduce that the stress concentration of the exterior total field, induced by the choice of the incident wave's index and the high-oscillation parameter, increases with a larger index $n$, corresponding to greater stress. Numerical examples, presented in Table~\ref{tab:22} of Section~\ref{sec:5}, demonstrate that the lower bound in \eqref{eq:Eus th4.2} is sharp for estimating the stress. Meanwhile, through the mathematical analysis of the radially symmetric case, we have established in Section~\ref{sec:5} that the stress concentration of the exterior total field holds for general shapes.

      For further explanation of the stress concentration of the exterior total field in Theorem~\ref{thm:Eu definition in thm}, we provide the following remarks for discussion. Remark~\ref{rem:41rem} clarifies that, as the stress of the incident wave is bounded relative to the stress of the exterior scattering field, the stress concentration in the exterior total field is mainly contributed by  the corresponding stress concentration in the exterior scattering field. In Remark~\ref{rem:43rem}, we discuss how the interior total field, the exterior scattering field, and the exterior total field exhibit different phenomena depending on the index $n$ of the incident wave or the high-oscillation parameter satisfying different conditions.

       \item[(ii)]

       To induce stress concentration in a neighborhood of the bubble's boundary, we first localize the wave and generate high-oscillation behavior near the boundary, as established in Theorems~\ref{thm:th3.1} and~\ref{thm:nabla u in thm}. These wave behaviors, referred to as boundary localization and surface resonance, respectively, characterize the dynamics near the bubble's boundary. When both boundary localization and surface resonance occur simultaneously, the quasi-Minnaert resonance is triggered. In this context, we demonstrate that the stress concentration mechanism is driven by the quasi-Minnaert resonance.

       Meanwhile, Remark~\ref{rem:32rem} introduces that we can flexibly choose the incident wave to separately realize the quasi-Minnaert resonance for the interior total field and the exterior scattering field, respectively. In Remark~\ref{rem:33rem}, we provide another option, incorporating the index of the incident wave and the high-oscillation parameter, to independently achieve the quasi-Minnaert resonance of the interior total field and the exterior scattering field, respectively. Both stress concentration and quasi-Minnaert resonance depend on a tuned incident wave and the high-contrast structure of the physical configuration. Specifically, a larger index $n$ associated with the incident wave induces stronger boundary localization and surface resonance.

	 \end{enumerate}

	The paper is structured as follows: Section~\ref{sec:mathematical formula} establishes the mathematical framework and layer potential theory, formally defining the concepts of  {stress concentration} and  {quasi-Minnaert resonance}. Section~\ref{sec:concentration of u and u^s} analyzes the boundary localization and surface resonance behavior of the internal acoustic total field and the exterior elastic scattered field across their respective interior and exterior boundary layers. Section~\ref{sec:summary of major findings} rigorously establishes the stress concentration phenomenon in the exterior elastic total field. Section~\ref{sec:5} validates the theoretical results through numerical experiments and demonstrates the stress concentration induced by appropriately selected incident waves and high-contrast bubble-elastic structures.

	\section{Mathematical formulations}\label{sec:mathematical formula}
	In this section, we present the mathematical formulation for our subsequent analysis. Let $ D \subset \mathbb{R}^3 $ denote an air bubble defined by the parameters $ (\rho_b, \kappa) $. Here, $ \rho_b \in \mathbb{R}_+ $ denotes the air density, and $ \kappa \in \mathbb{R}_+ $ represents the bulk modulus of the air inside the bubble. The domain $ \mathbb{R}^3 \setminus \overline{D} $ is occupied by a homogeneous elastic medium parameterized by $ (\tilde{\lambda}, \tilde{\mu}, \rho_e) $. The constant $ \rho_e \in \mathbb{R}_+ $ corresponds to the density of the surrounding elastic medium. The Lam\'e constants adhere to the strong convexity condition as specified in \cite{Kupradze}:
	\begin{align}\notag
		\tilde \mu>0, \quad  3 \tilde\lambda+2\tilde \mu>0.
	\end{align}
	Let $\mathbf{u}^i$ denote a time harmonic incident elastic wave constituting an entire solution to  
	\begin{align}\label{eq:incident elastic wave}  
		\mathcal{L}_{\tilde{\lambda}, \tilde{\mu}} \mathbf{u}^i(\mathbf{x}) + \omega^2 \rho_e \mathbf{u}^i(\mathbf{x}) = \mathbf{0} \quad \text{in } \mathbb{R}^3,  
	\end{align}  
	where $\omega > 0$ denotes the angular frequency, and the Lam\'e operator $\mathcal{L}_{\tilde{\lambda}, \tilde{\mu}}$ corresponding to the parameters $(\tilde{\lambda}, \tilde{\mu})$ is defined as  
	$$  
	\mathcal{L}_{\tilde{\lambda}, \tilde{\mu}} := \tilde{\mu} \triangle + (\tilde{\lambda} + \tilde{\mu}) \nabla \nabla \cdot.  
	$$  
	
	We study the interaction between a single air bubble $(D; \kappa, \rho_b)$ and the surrounding elastic medium $(\mathbb{R}^3 \setminus \overline{D}; \tilde{\lambda}, \tilde{\mu}, \rho_e)$. The  bubble-elastic scattering  is governed by the coupled PDE system as established in \cite{DLLT,LLZ,OS}:  
	\begin{equation}\label{eq:system}  
		\begin{cases}  
			\displaystyle  
			\frac{1}{\rho_b} \nabla \cdot (\nabla u(\mathbf{x})) + \frac{\omega^2}{\kappa} u(\mathbf{x}) = 0 & \text{in } D, \\  
			\mathcal{L}_{\tilde{\lambda}, \tilde{\mu}} \mathbf{u}(\mathbf{x}) + \omega^2 \rho_e \mathbf{u}(\mathbf{x}) = \mathbf{0} & \text{in } \mathbb{R}^3 \setminus \overline{D}, \\  
			\displaystyle  
			\mathbf{u}(\mathbf{x}) \cdot \boldsymbol{\nu} = \frac{1}{\rho_b \omega^2} \nabla u(\mathbf{x}) \cdot \boldsymbol{\nu} & \text{on } \partial D, \\  
			\partial_{\boldsymbol{\nu}, \tilde{\lambda}, \tilde{\mu}} \mathbf{u}(\mathbf{x}) = -u(\mathbf{x}) \boldsymbol{\nu} & \text{on } \partial D, \\  
			\mathbf{u}(\mathbf{x}) - \mathbf{u}^i(\mathbf{x}) \quad \text{satisfies the radiation condition},  
		\end{cases}  
	\end{equation}  
	where $\mathbf{u}(\mathbf{x})$ represents the total elastic wave field in $\mathbb{R}^3 \setminus \overline{D}$, and $u(\mathbf{x})$ denotes the acoustic pressure within $D$. In \eqref{eq:system}, the third equation enforces continuity of the normal displacement across $\partial D$, while the fourth ensures stress continuity at the interface. The co-normal derivative operator $\partial_{\boldsymbol{\nu}, \tilde{\lambda}, \tilde{\mu}}$, associated with the Lamé parameters $(\tilde{\lambda}, \tilde{\mu})$, is defined as  
	\begin{align}\notag 
		\partial_{\boldsymbol{\nu}, \tilde{\lambda}, \tilde{\mu}} \mathbf{u} := \tilde{\lambda} (\nabla \cdot \mathbf{u}) \boldsymbol{\nu} + 2\tilde{\mu} (\nabla^s \mathbf{u}) \boldsymbol{\nu},  
	\end{align}  
	where $\boldsymbol{\nu}$ denotes the outward unit normal to $\partial D$, and the symmetric gradient operator $\nabla^s$ is expressed as  
	\begin{align}\notag 
		\nabla^s \mathbf{u} := \frac{1}{2} \left( \nabla \mathbf{u} + \nabla \mathbf{u}^\top \right).  
	\end{align}  
      Here $\nabla \mathbf{u}=\left(\partial_j u_i\right)_{i, j=1}^3$ denotes the gradient of $\mathbf u$ at $\mathbf x$ and the superscript $\top$ denotes the matrix transpose. Using Helmholtz decomposition, the scattered wave field \(\mathbf{u}^s\) can be expressed as a superposition of  compressional and shear waves:
		\[
		\mathbf{u}^s = \mathbf{u}_p^s + \mathbf{u}_s^s,
		\]
		where the compressional wave component \(\mathbf{u}_p^s\) is given by:
		\[
		\mathbf{u}_{p}^{s} = -\frac{1}{\tilde{k}_{p}^{2}} \nabla (\nabla \cdot \mathbf{u}^{s}),
		\]
		and the shear wave component \(\mathbf{u}_s^s\) is defined as:
		\[
		\mathbf{u}_{s}^{s} = \frac{1}{\tilde{k}_{s}^{2}} \nabla \times \nabla \times \mathbf{u}^{s}.
		\]
	In this formulation, \(\tilde{k}_p\) and \(\tilde{k}_s\) are the wave numbers for the compressional and shear waves, respectively, defined as:
	\begin{align}\notag
			\tilde{k}_p = \frac{\omega}{\tilde{c}_p}, \quad \tilde{k}_s = \frac{\omega}{\tilde{c}_s},
	\end{align}
		with 
		\begin{align}\notag
			\tilde{c}_s=\sqrt{\tilde{\mu} / \rho_e} \quad \mbox{and} \quad \tilde{c}_p=\sqrt{(\tilde{\lambda}+2 \tilde{\mu}) / \rho_e},
		\end{align} 
		representing the compressional and shear wave speeds in the medium. The radiation condition in \eqref{eq:system} designates the following condition as $|\mathbf{x}| \rightarrow+\infty$ \cite{Kupradze}:
		$$
		\begin{array}{r}
			\left(\nabla \times \nabla \times \mathbf{u}^s\right)(\mathbf{x}) \times \frac{\mathbf{x}}{|\mathbf{x}|}- \mathrm{i} \tilde{k}_s \nabla \times \mathbf{u}^s(\mathbf{x})=\mathcal{O}\left(|\mathbf{x}|^{-2}\right), \\
			\frac{\mathbf{x}}{|\mathbf{x}|} \cdot\left[\nabla\left(\nabla \cdot \mathbf{u}^s\right)\right](\mathbf{x})-\mathrm{i} \tilde{k}_p \nabla \mathbf{u}^s(\mathbf{x})=\mathcal{O}\left(|\mathbf{x}|^{-2}\right),
		\end{array}
		$$
		where $\rmi$ is the imaginary unit. To facilitate analysis of the coupled system \eqref{eq:system}, we introduce the following non-dimensional parameters:  
	\begin{align}\label{eq:delta tau}
		\delta = \frac{\rho_b}{\rho_e}, \quad
		\tau = \frac{c_b}{\tilde{c}_p} = \frac{\sqrt{\kappa / \rho_b}}{\sqrt{(\tilde{\lambda} + 2\tilde{\mu}) / \rho_e}}, \quad
		\tilde{k} = \omega / c_b, \quad
		c_b = \sqrt{\kappa / \rho_b},
	\end{align}  
	where $c_b$ denotes the acoustic wave speed within $D$. The parameter $\delta \ll 1$ serves to quantify the density contrast between the bubble $D$ and the surrounding elastic medium $ \mathbb{R}^3 \setminus \overline{D} $, whereas $\tau$ serves to characterize the wave speed contrast between the air bubble $D$ and the elastic medium $ \mathbb{R}^3 \setminus \overline{D} $. Let $l_D$ denote the characteristic diameter of the bubble $D$. We introduce the following non-dimensional parameters:  
	\begin{align}\label{eq:non dim}
		\mathbf{x}' = \frac{\mathbf{x}}{l_D}, \quad
		k = \tilde{k} l_D, \quad
		\mathbf{u}' = \frac{\mathbf{u}}{l_D}, \quad
		\mu = \frac{\tilde{\mu}}{\tilde{\lambda} + 2\tilde{\mu}}, \quad
		\lambda = \frac{\tilde{\lambda}}{\tilde{\lambda} + 2\tilde{\mu}}, \quad
		v' = \frac{v}{\rho_b c_b^2}.
	\end{align}  
    Therefore, under the aforementioned assumptions, it follows that   
	\begin{align}\label{eq:assume}
		k = o(1), \quad
		\delta = o(1), \quad
		\tau < 1, \quad
		\mu = \mathcal{O}(1), \quad
		\lambda = \mathcal{O}(1).
	\end{align} 

    \begin{rem}\label{rem:tau}
         The assumption of $\tau$ in \eqref{eq:assume} is consistent with the physical configuration \cite{Born}, in which the wave speed within the background material exceeds that in the bubble region $D$. The elastic medium in the background may include metallic materials, such as titanium, nickel, or gold, as well as non-metallic materials, such as glass or rubber.
    \end{rem}

   By substituting these parameters into \eqref{eq:system} and omitting the prime notation, we derive the following non-dimensional coupled PDE system (cf.\cite{DLLT,LLZ}):  
	\begin{equation}\label{eq:system2}
		\begin{cases}
		\Delta u(\mathbf{x}) + k^2 u(\mathbf{x}) = 0 & \text{in } D, \\
		\mathcal{L}_{\lambda,\mu} \mathbf{u}(\mathbf{x}) + k^2 \tau^2 \mathbf{u}(\mathbf{x}) = \mathbf{0}	 & \text{in } \mathbb{R}^3 \setminus \overline{D}, \\
		\mathbf{u}(\mathbf{x}) \cdot \boldsymbol{\nu} - \frac{1}{k^2} \nabla u(\mathbf{x}) \cdot \boldsymbol{\nu} = 0 & \text{on } \partial D, \\
		\partial_{\boldsymbol{\nu}, \lambda, \mu} \mathbf{u}(\mathbf{x}) + \delta \tau^2 u(\mathbf{x}) \boldsymbol{\nu} = \mathbf{0} & \text{on } \partial D.
		\end{cases}
	\end{equation}
In the following study, we replace $\partial_{\boldsymbol{\nu},\lambda,\mu}$ with $\partial_{\boldsymbol{\nu}}$. It is remarked that the system \eqref{eq:system2} is equivalent to the original system \eqref{eq:system}. Consequently, in what follows we focus on studying the system \eqref{eq:system2} instead of the system \eqref{eq:system}. Here, $\tau$ is as defined in \eqref{eq:delta tau}. We emphasize that in \eqref{eq:system2}, the compressional and shear waves satisfy 
\begin{align}\label{eq:ks kp defination}
	k_p= \frac{k\tau}{\sqrt{\lambda + 2\mu}}=o(1), \quad k_s=\frac{k\tau}{\sqrt{ 2\mu}}=o(1).
\end{align}

We employ layer potential theory to recast the coupled PDE system \eqref{eq:system2} as a scattering problem expressed through a system of integral equations. To that end, we first introduce the layer potential operators for our subsequent analysis. Let $G^k(\mathbf{x})$ be the fundamental solution of the operator  $\Delta + k^2$, namely 
\begin{align}\notag 
	G^k(\mathbf{x})=-\frac{e^{\rmi k|\mathbf{x}|}}{4\pi |\mathbf{x}|}.
\end{align}
The single layer potential associated with the Helmholtz system is defined by
\begin{align}\label{eq:single layer potential with acoustic}
	\mathcal{S}_{\partial D}^{k}[\varphi](\mathbf{x})=\int_{\partial D} G^{k}(\mathbf{x}-\mathbf{y})\varphi(\mathbf{y})ds(\mathbf{y}) \quad \mathbf{x} \in \mathbb{R}^{3},
\end{align}
with $\varphi(\mathbf{x}) \in L^2(\partial D)^3$. Then the co-normal derivative of the single layer potential enjoys the jump formula 
\begin{align}\label{eq:jump in acoustic}
	\nabla \mathcal{S}_{\partial D}^{k}[\varphi] \cdot \boldsymbol{\nu}_{\pm}(\mathbf{x}) = \left( \pm \frac{1}{2} \mathcal{I} + \mathcal{K}_{\partial D}^{k,*}\right) [\varphi](\mathbf{x}), \quad \mathbf{x} \in \partial D,
\end{align}
where $\mathcal{I}$ is an identity operator and 
\begin{align}\label{eq:np acoustic}
	\mathcal{K}_{\partial D}^{k,*}[\varphi](\mathbf{x}) = \int_{\partial D} \nabla_{\mathbf{x}} G^{k}(\mathbf{x} - \mathbf{y}) \cdot \boldsymbol{\nu}_{\mathbf{x}} \varphi(\mathbf{y}) ds(\mathbf{y}), \quad \mathbf{x} \in \partial D,
\end{align}
which is also known as the Neumann-Poincar\'e (N-P) operator associated with Helmholtz system. Here and also in what follows, the subscript $\pm$ indicates the limits from outside and inside $D$, respectively.

For elasticity, we begin by introducing the potential theory relevant to the Lam\'e system. The fundamental solution $\mathbf{\Gamma}^\omega = \left(\Gamma_{i, j}^\omega\right)_{i, j=1}^3$ for the operator $\mathcal{L}_{\lambda, \mu} + \omega^2\rho_e$ in three dimensions is given in  as \cite{ABG}:
\begin{align}
	\left(\Gamma_{i, j}^\omega\right)_{i, j=1}^3(\mathbf{x}) = -\frac{\delta_{ij}}{4 \pi \mu |\mathbf{x}|} e^{\mathrm{i} k_s |\mathbf{x}|} + \frac{1}{4 \pi \omega^2 \rho_e} \partial_i \partial_j \frac{e^{\mathrm{i} k_p |\mathbf{x}|} - e^{\mathrm{i} k_s |\mathbf{x}|}}{|\mathbf{x}|},\notag
\end{align}
where $\delta_{ij}$ is the Kronecker delta function. Here, $k_\alpha$, for $\alpha = p, s$, is defined by \eqref{eq:ks kp defination}. It can be shown that $\Gamma_{i j}^\omega$ has the following series representation (cf. \cite{ABG}):
\begin{align}
	\Gamma_{i j}^\omega(\mathbf{x}) = & -\frac{1}{4 \pi \rho_e} \sum_{n=0}^{+\infty} \frac{\mathrm{i}^n}{(n+2) n!}\left(\frac{n+1}{c_s^{n+2}} + \frac{1}{c_p^{n+2}}\right) \omega^n \delta_{ij} |\mathbf{x}|^{n-1} \notag\\
	& + \frac{1}{4 \pi \rho_e} \sum_{n=0}^{+\infty} \frac{\mathrm{i}^n (n-1)}{(n+2) n!} \left(\frac{1}{c_s^{n+2}} - \frac{1}{c_p^{n+2}}\right) \omega^n |\mathbf{x}|^{n-3} x_i x_j.\notag
\end{align}
In particular, when $\omega = 0$, we denote $\mathbf{\Gamma}^0$ as $\mathbf{\Gamma}$ for simplicity, and its expression (cf. \cite{ABG}) is given by:
\begin{align}\notag
	\Gamma_{i, j}^0(\mathbf{x}) = -\frac{1}{8 \pi} \left(\frac{1}{\mu} + \frac{1}{\lambda + 2\mu}\right) \frac{\delta_{ij}}{|\mathbf{x}|} - \frac{1}{8 \pi} \left(\frac{1}{\mu} - \frac{1}{\lambda + 2\mu}\right) \frac{x_i x_j}{|\mathbf{x}|^3}.
\end{align}
The single-layer potential associated with the fundamental solution $\mathbf{\Gamma}^\omega$ is defined as:
\begin{align}\label{eq:single layer lame}
	\mathbf{S}_{\partial D}^\omega[\boldsymbol{\varphi}](\mathbf{x}) = \int_{\partial D} \mathbf{\Gamma}^\omega(\mathbf{x} - \mathbf{y}) \boldsymbol{\varphi}(\mathbf{y}) \, \mathrm{d}s(\mathbf{y}), \quad \mathbf{x} \in \mathbb{R}^3,
\end{align}
for $\boldsymbol{\varphi} (\mathbf{x}) \in L^2(\partial D)^3$. On $\partial D$, the co-normal derivative of the single-layer potential satisfies the following jump relation:
\begin{align}\label{eq:jump relation lame}
	\left.\partial_{\boldsymbol{\nu}} \mathbf{S}_{\partial D}^\omega[\boldsymbol{\varphi}]\right|_{\pm}(\mathbf{x}) = \left( \pm \frac{1}{2} \mathbf{I} + \mathbf{K}_{\partial D}^{\omega, *} \right)[\boldsymbol{\varphi}](\mathbf{x}), \quad \mathbf{x} \in \partial D,
\end{align}
where
\begin{align}\label{eq:np operator lame}
	\mathbf{K}_{\partial D}^{\omega, *}[\boldsymbol{\varphi}](\mathbf{x}) = \text{p.v.} \int_{\partial D} \partial_{\boldsymbol{\nu}_{\mathbf{x}}} \mathbf{\Gamma}^\omega(\mathbf{x} - \mathbf{y}) \boldsymbol{\varphi}(\mathbf{y}) \, \mathrm{d}s(\mathbf{y}).
\end{align}
Here, p.v. refers to the Cauchy principal value. Notably, the operator $\mathbf
{K}_{\partial D}^{\omega, *}$, defined in \eqref{eq:np operator lame}, is known as the Neumann-Poincar\'e (N-P) operator with Lam\'e system. With the help of the potential operators introduced above, the solution to the system \eqref{eq:system2} can be expressed as:
\begin{align}\label{eq:u express}
	\mathbf{u}=
	\begin{cases}
		\mathcal{S}_{\partial D}^{k}[\varphi_b](\mathbf{x}), & \mathbf{x} \in D, \\
		\mathbf{S}_{\partial D}^{k\tau}[\boldsymbol{\varphi}_e](\mathbf{x})+\mathbf{u}^i, & \mathbf{x} \in \mathbb{R}^3 \setminus \overline{D},
	\end{cases}
\end{align}
where the density functions $\varphi_b \in L^2(\partial D)^3$ and $\boldsymbol{\varphi}_e \in L^2(\partial D)^3$ are determined by the transmission conditions across $\partial D$ in \eqref{eq:system2}. Here, the operator $\mathcal{S}_{\partial D}^{k}$ represents the single-layer potential for the Helmholtz system defined in \eqref{eq:single layer potential with acoustic}, corresponding to the parameters $\kappa$ and $\rho_b$. Meanwhile, the operator $\mathbf{S}_{\partial D}^{k\tau}$ corresponds to the single-layer potential for the Lam\'e system defined in \eqref{eq:single layer lame}, associated with the parameters $\lambda$, $\mu$, and $\rho_e$. Using the transmission conditions in \eqref{eq:system2} and the jump relations \eqref{eq:jump in acoustic} and \eqref{eq:jump relation lame}, the equivalent integral equations for \eqref{eq:system2} can be derived as:  
\begin{align}\label{eq:2.20}
	\mathcal{A}(k, \delta)[\bmf{\Phi}](\mathbf{x}) = \bmf{F}(\mathbf{x}), \quad \mathbf{x} \in \partial D,
\end{align}
where
\begin{align}
	\mathcal{A}(k, \delta) 
	=\begin{pmatrix}
		-\frac{\mathcal{I}}{2}+\mathcal{K}_{\partial D}^{k,*}& 
		-k^2 \boldsymbol{\nu}\cdot\mathbf{S}_{\partial D}^{k\tau} \notag\\
		\delta\tau^2\boldsymbol{\nu}\mathcal{S}_{\partial D}^{k}& \frac{\mathbf{I}}{2}+\mathbf{K}^{k \tau,*}_{\partial D}
	\end{pmatrix},\quad 
	\bmf{\Phi} = \begin{pmatrix}
		\varphi_b \\
		\boldsymbol{\varphi_e}
	\end{pmatrix}, \quad 
	\bmf{F} = \begin{pmatrix}
		k^2 \boldsymbol{\nu}\cdot\mathbf{u}^i \\
		-\partial_{\boldsymbol{\nu}} \mathbf{u}^i
	\end{pmatrix}.
\end{align}
Here the Neumann-Poincar\'e operator $\mathcal{K}_{\partial D}^{k,*}$ is defined in \eqref{eq:np acoustic}, corresponding to the parameters $\kappa$ and $\rho_b$. Similarly, the operator $\mathbf{K}_{\partial D}^{k\tau,*}$ represents the Neumann-Poincar\'e operator for the Lam\'e system defined in \eqref{eq:np operator lame}, associated with the parameters $(\lambda, \mu, \rho_e)$. We demonstrate that the scattering problem \eqref{eq:system2} is mathematically equivalent to the boundary integral equations \eqref{eq:2.20}. For $k \ll 1$, the single-layer potential operator $\mathcal{S}^{k}_{\partial D}: L^{2}(\partial D) \rightarrow L^{2}(\partial D)$ is invertible. Combining \eqref{eq:2.20} with the invertibility of $\mathcal{S}^{k}_{\partial D}$, we derive  
\begin{align}\label{eq:varphi b}
   \varphi_b=\frac{1}{\delta\tau^2}\left(\mathcal{S}^{k}_{\partial D}\right)^{-1}\left(-\boldsymbol{\nu}\cdot\left(\frac{\mathbf{I}}{2}+\mathbf{K}_{\partial D}^{k \tau,*}\right)\boldsymbol{\varphi_e}\left(\mathbf{x}\right)-\boldsymbol{\nu}\cdot\partial_{\boldsymbol{\nu}}\mathbf{u}^i\right).
\end{align} 
Inserting this into \eqref{eq:2.20} yields that
\begin{align}\label{eq:tilde A k delta}
	\tilde{\mathcal{A}}(k,\delta)[\boldsymbol{\varphi_e}](\mathbf{x})=\tilde{\bmf{F}}(\mathbf{x}),
\end{align}
where 
 \begin{align}
 	\tilde{\mathcal{A}}(k,\delta)[\boldsymbol{\varphi_e}](\mathbf{x})
 	&=\left(\left(-\frac{\mathcal{I}}{2}+\mathcal{K}_{\partial D}^{k,*}\right)\left(\mathcal{S}^{k}_{\partial D}\right)^{-1}
 	\boldsymbol{\nu}\cdot\left(\frac{\mathbf{I}}{2}+\mathbf{K}_{\partial D}^{k\tau,*}\right)
 	+\delta \tau^2 k^2 \boldsymbol{\nu} \cdot \mathbf{S}_{\partial D}^{k \tau}\right)
 	[\boldsymbol{\varphi_e}](\mathbf{x}), \label{eq:tilde Ak delta}\\
 	\tilde{\bmf{F}}(\mathbf{x})
 	&=-\delta\tau^2 k^2 \boldsymbol{\nu} \cdot \mathbf{u}^i
 	-\left(-\frac{\mathcal{I}}{2}+\mathcal{K}_{\partial D}^{k,*}\right)\left(\mathcal{S}^{k}_{\partial D}\right)^{-1}
 	\boldsymbol{\nu}\cdot \partial_{\boldsymbol{\nu}} \mathbf{u}^i (\mathbf{x}). \label{eq:tilde F}
 \end{align}
This is an equivalent integral equation of the system \eqref{eq:2.20}.

 In this paper, we investigate the stress concentration of the external total field $\mathbf{u}|_{\mathbb{R}^3 \setminus \overline{D}}$ in the system \eqref{eq:system2} through quasi-Minnaert resonance. The quasi-Minnaert resonance is characterized by boundary localization and surface resonance of the associated wave field. Specifically, we demonstrate that when the external total field exhibits boundary localization and surface resonance, stress concentration occurs near the boundary of the air bubble. To facilitate this analysis, we first introduce the relevant definitions and notations.

  \begin{defn}\label{def:stress concentration}
 	Let $D$ be a bounded Lipschitz domain with a connected complement. We consider the exterior total field $\mathbf{u}|_{\mathbb{R}^3 \setminus \overline{D}}$, which satisfies the scattering problem \eqref{eq:system2} for a given incident wave $\mathbf{u}^i$.
 	Let $B_R$ be an origin centered ball of radius $R$ in $\mathbb{R}^3$ such that $D \subset B_R$. For sufficiently small parameters $\varsigma_1, \varsigma_2 \in \mathbb{R}_+$, we define the interior and exterior boundary layers with respect to $\partial D$ by
 	\begin{align}\label{eq:Sdef}
 		\mathscr{S}_{-}^{\varsigma_1}(\partial D) := \left\{ \mathbf{x} \in D \mid \mathrm{dist}_{\partial D}(\mathbf{x}) < \varsigma_1 \right\}, \quad \mathscr{S}_{+}^{\varsigma_2}(\partial D) := \left\{ \mathbf{x} \in B_R \setminus \overline{D} \mid \mathrm{dist}_{\partial D}(\mathbf{x}) < \varsigma_2 \right\},
 	\end{align}
 	where $\mathrm{dist}_{\partial D}(\mathbf{x}) := \inf\limits_{\mathbf{y} \in \partial D}  \|\mathbf{x} - \mathbf{y}\|$ represents the Euclidean distance between $\mathbf{x}$ and $\partial D$. Define the stress \(E(\mathbf{u})\) as
 	\begin{align}\label{eq:Eu}
 		E(\mathbf{u}) = \int_{\mathscr{S}_{+}^{\varsigma_2 }(\partial D)} \sigma(\mathbf{u}) : \nabla \overline{\mathbf{u}} \, \mathrm{d} \mathbf{x}, \quad \text{where} \quad \sigma(\mathbf{u}) = \lambda (\nabla \cdot \mathbf{u}) \mathbf{I} + \mu \left( \nabla \mathbf{u} + (\nabla \mathbf{u})^{\top} \right).
 	\end{align}
 	Here, \(\lambda\) and \(\mu\) are the Lam\'e constants, as defined in \eqref{eq:non dim}, \(\mathbf{I}\) is the identity matrix, and the operator ":" represents the Frobenius inner product for matrices. If the following condition holds:
 	\begin{equation}
 		\frac{E(\mathbf{u})}{\| \mathbf{u}^i\|_{L^2(D)^3}^2} \gg 1,
 	\end{equation}
 	then the exterior total field \(\mathbf{u}|_{\mathbb{R}^3 \setminus \overline{D}}\) demonstrates stress concentration.
 \end{defn}

 \begin{rem}
     
   In linear elasticity, the vector field $\mathbf{u}$ in \eqref{eq:system} represents the displacement field, which characterizes the deformation of the external elastic medium $\mathbb{R}^3 \setminus \overline{D}$ (cf.~\cite{Kupradze}).  The total strain energy $E(\mathbf{u})$ in linear elasticity quantifies the energy stored in the displacement field $\mathbf{u}$. Typically, we assume $\zeta_2 = 1 + \delta_{\zeta_2}$, where $\delta_{\zeta_2} \ll 1$, indicating that stress concentration occurs near the boundary of the domain $D$. The integrand 
   \begin{equation}\label{eq:Eu223}
       {\mathcal E}(\mathbf{u})=\sigma(\mathbf{u})(\mathbf x) : \nabla \overline{\mathbf{u}}(\mathbf x)
   \end{equation}
   represents the strain energy density (cf.~\cite{A12, Kupradze}), with its integral over the surrounding elastic medium $\mathscr{S}_{+}^{\zeta_2}(\partial D)$ quantifying the stored energy. In the sub-wavelength regime, using a tuned incident wave $\mathbf{u}^i$ and a bubbly-elastic structure, we demonstrate in Theorem \ref{thm:Eu definition in thm} that stress concentration occurs in the external total wave field, where $E(\mathbf{u})$ is primarily driven by the stress concentration of the external scattered field $\mathbf{u}^s$. Furthermore, we prove that the stress associated with the incident wave $\mathbf{u}^i$ is bounded relative to that of the scattered field $\mathbf{u}^s$. A detailed explanation is provided in Remark \ref{rem:41rem}.
 \end{rem}

To induce stress concentration near the bubble's boundary, we first localize the wave and generate high-oscillation behavior near the bubble's boundary. These phenomena, termed boundary localization and surface resonance, respectively, are used to characterize  quasi-Minnaert resonance. We then provide definitions for boundary localization and surface resonance.

\begin{defn}\label{def:boundary localization}
	Let $D$ be a bounded Lipschitz domain with a connected complement. We consider the interior total field $\mathbf{u}|_{D}$ and the exterior scattered field $\mathbf{u}^s|_{\mathbb{R}^3 \setminus \overline{D}}$, which satisfy the scattering problem defined in \eqref{eq:system2} for a given incident wave $\mathbf{u}^i$. The interior and exterior boundary layers are defined in \eqref{eq:Sdef}. The field $\mathbf{u}|_{D}$ is termed  {interior boundary localization}, and $\mathbf{u}^s|_{\mathbb{R}^3 \setminus \overline{D}}$ is termed {exterior boundary localization} if there exist sufficiently small $\varsigma_1, \varsigma_2, \eta \in \mathbb{R}_+$ such that
	\begin{align}\label{eq:226}
		\frac{\|{u}\|_{L^2(D \setminus \mathscr{S}_{-}^{\varsigma_1}(\partial D))^3}}{\|{u}\|_{L^2(D)^3}} \leqslant \eta, \quad 
		\frac{\|\mathbf{u}^s\|_{L^2((B_R \setminus \overline{D}) \setminus \mathscr{S}_{+}^{\varsigma_2}(\partial D))^3}}{\|\mathbf{u}^s\|_{L^2(B_R \setminus \overline{D})^3}} \leqslant \eta.
	\end{align}  
	The parameter $\eta$ quantifies the level of boundary localization.   
\end{defn}

\begin{defn}\label{def:surface resonant}
	Let ${u}|_{D}$ denote the interior total field and $\mathbf{u}^s|_{\mathbb{R}^3 \setminus \overline{D}}$ denote the exterior scattered field for the scattering problem \eqref{eq:system2} with a given incident wave $\mathbf{u}^i$, where $D$ is a bounded Lipschitz domain. The boundary layers $\mathscr{S}_{-}^{\varsigma_1}(\partial D)$ and $\mathscr{S}_{+}^{\varsigma_2}(\partial D)$ are defined as in \eqref{eq:Sdef}. We say the fields ${u}|_{D}$ and $\mathbf{u}^s|_{\mathbb{R}^3 \setminus \overline{D}}$ exhibit  {surface resonance} if the following conditions hold:
	\begin{align}\label{eq:228}
		\frac{\|\nabla {u}\|_{L^2(\mathscr{S}_{-}^{\varsigma_1}(\partial D))^3}}{\|\mathbf{u}^i\|_{L^2(D)^3}} \gg 1, 
		\quad \text{and} \quad
		\frac{\|\nabla \mathbf{u}^s\|_{L^2(\mathscr{S}_{+}^{\varsigma_2}(\partial D))^3}}{\|\mathbf{u}^i\|_{L^2(D)^3}} \gg 1.
	\end{align}
\end{defn}

\begin{rem}\label{rem:2.2}

   In the sub-wavelength regime, where the bubble $D$ is significantly smaller than the incident wavelength, the $L^2$-norm of the incident field $\mathbf{u}^i$ over $D$ is typically small. Consequently, it is reasonable to normalize $\mathbf{u}^i$ with respect to its $L^2$-norm over $D$. This normalization introduces the ratio between $\|\nabla \mathbf{u}\|_{L^2(\mathscr{S}_{-}^{\zeta_1}(\partial D))^3}$, $\|\nabla \mathbf{u}^s\|_{L^2(\mathscr{S}_{+}^{\zeta_2}(\partial D))^3}$, and $\|\mathbf{u}^i\|_{L^2(D)^3}$ in \eqref{eq:228}. Notably, the two ratios defined in \eqref{eq:226} remain invariant when both their numerators and denominators are divided by $\|\mathbf{u}^i\|_{L^2(D)^3}$. Thus, we implicitly normalize $\mathbf{u}^i$ in Definition \ref{def:boundary localization}.

\end{rem}

Building on the preceding Definitions~\ref{def:boundary localization} and~\ref{def:surface resonant}, we formulate the following definition of quasi-Minnaert resonance, which will used to generate the stress concentration in Section \ref{sec:concentration of u and u^s} by appropriately choosing the incident wave and bubble-elastic structure.


\begin{defn}\label{def:quasi minnaert}
	
	Let $D \subset \mathbb{R}^3$ be a bounded Lipschitz domain in the scattering system \eqref{eq:system2}, excited by an incident wave $\mathbf{u}^i$ at frequency $\omega$. We say $\omega$ constitutes a \textit{quasi-Minnaert resonance frequency} for $\mathbf{u}^i$ if the induced total field $\mathbf{u}|_{D}$ and scattered field $\mathbf{u}^s|_{\mathbb{R}^3 \setminus \overline{D}}$ simultaneously satisfy the conditions for boundary localization in Definition~\ref{def:boundary localization} and surface resonance in Definition~\ref{def:surface resonant}. The domain $D$ that exhibits these phenomena is then designated as the \textit{quasi-Minnaert resonator} associated with $\mathbf{u}^i$.  
	
\end{defn}

\section{Boundary localization and surface resonance }\label{sec:concentration of u and u^s}

This section presents the spectral properties of the single-layer acoustic potential $\mathcal{S}_{\partial D}^{k}$  and elastic potential $\mathbf{S}_{\partial D}^\omega$, alongside the Neumann-Poincaré operators $\mathcal{K}_{\partial D}^{k,*}$ and $\mathbf{K}_{\partial D}^{\omega,*}$, which are crucial to proving Theorems \ref{thm:th3.1} and \ref{thm:nabla u in thm}. Utilizing these spectral properties, we select an incident wave $\mathbf{u}^i$, as defined in \eqref{eq:ui def}, that is an entire solution to \eqref{eq:incident elastic wave}. This ensures that the interior total field $\mathbf{u}|_D$ and the exterior scattered field $\mathbf{u}^s|_{\mathbb{R}^3 \setminus \overline{D}}$, governed by \eqref{eq:system2} and associated with $\mathbf{u}^i$, exhibit boundary localization and surface resonance within the interior and exterior of $D$, respectively, as established in Theorems \ref{thm:th3.1} and \ref{thm:nabla u in thm}. After we show the exterior scattered field $\mathbf{u}^s|_{\mathbb{R}^3 \setminus \overline{D}}$ exhibits localized, highly oscillatory  behavior near the boundary of the bubble $D$, we shall rigorously prove the stress concentration of the exterior elastic total field $\mathbf u$ near $\partial D$ in Section \ref{sec:summary of major findings}. This analysis is elaborated in detail therein.

Throughout this section, the bubble $D$ in \eqref{eq:system2} is assumed to be a unit ball centered at the origin in $\mathbb{R}^3$. In the subsequent analysis, we formalize the requisite notation and analytical framework. Let $\mathbb{N}$ denote the set of positive integers, and let $\mathbb{N}_0 = \mathbb{N} \cup \{0\}$. The spherical harmonic functions $Y_n^m(\theta, \varphi)$ (cf.\cite{CK}) are defined as
\begin{align}\label{eq:ynm and cnm def}
	Y_n^m(\theta, \varphi) := C_n^m P_n^m(\cos\theta) e^{\mathrm{i} m \varphi}, \quad  C_n^m = \sqrt{\frac{2n+1}{4\pi} \frac{(n-|m|)!}{(n+|m|)!}},
\end{align}
where $n \in \mathbb{N}_0$ and $-n \leqslant m \leqslant n$. It is straightforward to verify that
\begin{align}\label{eq:ynm3.2}
	\int_{\mathbb{S}} \nabla_\mathbb{S} Y_{n}^{m} \cdot \nabla_\mathbb{S}  \overline{Y_{n^\prime}^{m^\prime}} \mathrm{d}s = n(n + 1)\int_{\mathbb{S}} Y_{n}^{m} \overline{Y_{n^\prime}^{m^\prime}} \mathrm{d}s = n(n+1) \delta_{n n^\prime} \delta_{m m^\prime},
\end{align}
where $\delta_{n n^\prime}$ denotes the Kronecker delta and $\mathbb{S}$ represents the surface of the unit sphere. Let $j_n(z)$ and $h_n(z)$ denote the spherical Bessel and Hankel functions of the first kind, respectively, of order $n$. For any fixed $n \in \mathbb{N}_0$ and $0 < |z| \ll 1$, the asymptotic expansions hold (cf. \cite[(2.32)]{CK}):
\begin{subequations}
	\begin{align}
		j_n(z) &= \frac{z^n}{(2n+1)!!} \left(1 - \frac{z^2}{2(2n+3)} + \frac{z^4}{8(2n+3)(2n+5)} + \mathcal{O}(z^5)\right), \label{eq:jn expansion}\\
		h_0(z) &= \frac{1}{\mathrm{i} z}\left(1 + \mathrm{i} z - \frac{z^2}{2} - \frac{\mathrm{i} z^3}{6} + \frac{z^4}{24} + \mathcal{O}(z^5)\right), \notag\\ 
		h_1(z) &= \frac{1}{\mathrm{i} z^2}\left(1 + \frac{z^2}{2} + \frac{\mathrm{i} z^3}{3} - \frac{z^4}{8} + \mathcal{O}(z^5)\right), \notag\\
		h_n(z) &= \frac{(2n-1)!!}{\mathrm{i} z^{n+1}} \left(1 - \frac{z^2}{2(2n-1)} + \frac{z^4}{8(2n-1)(2n-3)} + \mathcal{O}(z^5)\right), \quad n \geqslant 2. \label{eq:hn expansion}
	\end{align}
\end{subequations}
From the series representations of $h_n(z)$, the following recurrence relations hold (cf. \cite{NIST}):
\begin{align}\label{eq:hn qiudao}
	h_n^{\prime}(z) &= h_{n-1}(z) - \frac{(n+1)}{z} h_n(z), \quad n=1, 2, \dots, \notag \\
	h_n^{\prime}(z) &= -h_{n+1}(z) + \frac{n}{z} h_n(z), \quad n=0, 1, \dots,\\
	h_{n+1}(z)&=\frac{2n+1}{z}h_n(z)-h_{n-1}(z), \quad n=1, 2, \dots .\notag
\end{align}
These differentiation formulae apply equally to $j_n(z)$. Next, let us define the vectorial spherical harmonics of order n. The set $\left(\mathcal{I}_n^m, \mathcal{T}_n^m, \mathcal{N}_n^m\right)$ constitutes an orthogonal basis of $\left(L^2(\partial D)\right)^3$, where 
\begin{align}\label{eq:Inm Tnm Nnm definition}
	\mathcal{I}_n^m &=\nabla_{\partial D} Y_{n+1}^m+(n+1) Y_{n+1}^m \boldsymbol{\nu}, \quad n \geqslant 0, n+1 \geqslant m \geqslant-(n+1), \notag\\
	\mathcal{T}_n^m &=\nabla_{\partial D} Y_n^m \wedge \boldsymbol{\nu}, \quad n \geqslant 1, n \geqslant m \geqslant-n, \\
	\mathcal{N}_n^m &=-\nabla_{\partial D} Y_{n-1}^m+n Y_{n-1}^m \boldsymbol{\nu}, \quad n \geqslant 1, n+1 \geqslant m \geqslant-(n+1).\notag 
\end{align}
 Thanks to \eqref{eq:Inm Tnm Nnm definition}, one has that
\begin{align}\label{eq:nu TIN}
	\boldsymbol{\nu}  \cdot \mathcal{T}_{n}^{m}
	&=\boldsymbol{\nu}  \cdot (\nabla_{\mathbb{S}^2} Y_{n + 1}^{m} \wedge \boldsymbol{\nu}) = 0, \notag\\
	\boldsymbol{\nu} \cdot \mathcal{I}_{n}^{m}
	&=\boldsymbol{\nu} \cdot (\nabla_{\mathbb{S}^2} Y_{n+1}^{m}+(n+1)Y_{n+1}^{m}\boldsymbol{\nu}) = (n+1)Y_{n + 1}^{m}, \\
	\boldsymbol{\nu}  \cdot \mathcal{N}_{n}^{m} 
	&=\boldsymbol{\nu} \cdot (-\nabla_{\mathbb{S}^2} Y_{n-1}^{m}+nY_{n-1}^{m}\boldsymbol{\nu})=nY_{n-1}^{m}.\notag
\end{align}

We proceed to define the eigenfunctions and their corresponding eigenvalues, which are central to the subsequent analytical framework. We begin by examining the eigenvalues and eigenfunctions of the acoustic single layer potential operator $\mathcal{S}_{\partial D}^k$, defined in \eqref{eq:single layer potential with acoustic}, for clarity.

\begin{lem}\cite[Lemma 3.1]{DLL2020}
	The eigensystem of the single layer potential operator $\mathcal{S}_{\partial D}^k$ defined in \eqref{eq:single layer potential with acoustic} is given as follows:
	\begin{align}\notag
		\mathcal{S}_{\partial D}^k [Y_n^m](\mathbf{x}) = -\rmi k j_n(k)h_n(k)Y_n^m, \quad \mathbf{x} \in \partial D.
	\end{align}
	Moreover, the following two identities hold:
	\begin{equation}\label{eq:S in D}
		\mathcal{S}_{\partial D}^k [Y_n^m](\mathbf{x}) = -\rmi k j_n(k|\mathbf{x}|)h_n(k)Y_n^m, \quad \mathbf{x} \in D,
	\end{equation}
	and
	\begin{align}\notag
		\mathcal{S}_{\partial D}^k [Y_n^m](\mathbf{x}) = -\rmi k j_n(k)h_n(k|\mathbf{x}|)Y_n^m, \quad \mathbf{x} \in \mathbb{R}^3\setminus\overline{D}.
	\end{align}
\end{lem}

The spectral structure of the elastic single layer potential operator $\mathbf{S}_{\partial D}^{k}$ defined in \eqref{eq:single layer lame} can be explicitly characterized as follows.

 \begin{lem}\label{lem:single on TIN}\cite[Proposition 3.3]{DLL2020}
	The single layer potential operator $\mathbf{S}_{\partial D}^k$ on $\partial D$ satisfies that
	\begin{align}
		\mathbf{S}_{\partial D}^k[\mathcal{T}_{n}^m] &=b_{n}(k) \mathcal{T}_{n}^m, \notag\\
		\mathbf{S}_{\partial D}^k[\mathcal{I}_{n-1}^m] &=c_{1n}(k) \mathcal{I}_{n-1}^m+ d_{1n}(k)\mathcal{N}_{n+1}^m, \notag\\
		\mathbf{S}_{\partial D}^k[\mathcal{N}_{n + 1}^m] &= c_{2n}(k) \mathcal{I}_{n-1}^m+ d_{2n}(k)\mathcal{N}_{n+1}^m.\notag
	\end{align}
	Here
	\begin{align}\label{eq:52}
		b_{n}(k) &=-\frac{\rmi k_{s}j_{n}(k_{s})h_{n}(k_{s})}{\mu}, \notag \\ 
		c_{1n}(k) &=-\rmi \left(\frac{(n+1)j_{n-1}(k_{s})h_{n-1}(k_{s})k_{s}}{\mu(2n+1)} +\frac{nj_{n-1}(k_{p})h_{n-1}(k_{p})k_{p}}{(\lambda+2\mu)(2n+1)} \right), \notag \\ 
		d_{1n}(k) &=-\rmi \left(\frac{nj_{n-1}(k_{s})h_{n+1}(k_{s})k_{s}}{\mu(2n+ 1)}-\frac{nj_{n-1}(k_{p})h_{n+1}(k_{p})k_{p}}{(\lambda+2\mu)(2n+1)} \right), \\
		c_{2n}(k) &=-\rmi \left(\frac{(n+1)j_{n+1}(k_{s})h_{n-1}(k_{s})k_{s}}{\mu(2n+1)}- \frac{(n+1)j_{n+1}(k_{p})h_{n-1}(k_{p})k_{p}}{(\lambda+2\mu)(2n+1)} \right), \notag \\ 
		d_{2n}(k) &=-\rmi \left(\frac{nj_{n+1}(k_{s})h_{n+1}(k_{s})k_{s}}{\mu(2n+1)}+\frac{(n+1)j_{n+1}(k_{p})h_{n+1}(k_{p})k_{p}}{(\lambda+2\mu)(2n+1)} \right), \notag
	\end{align}
	where $k_{s}$ and $k_{p}$ are defined in \eqref{eq:ks kp defination}.
\end{lem}

\begin{lem}\label{lem:trace single in TIN}\cite[Proposition 3.3]{DLL2020}
	The tractions of the single layer potentials $\mathbf{S}_{\partial D}^k[\mathcal{T}_n^m]$, $\mathbf{S}_{\partial D}^k[\mathcal{I}_{n-1}^m]$ and $\mathbf{S}_{\partial D}^k[\mathcal{N}_{n+1}^m]$ on $\partial D$ satisfy that
	\begin{align}
		\left.\partial_{\boldsymbol{\nu}}\mathbf{S}_{\partial D}^k[\mathcal{T}_n^m]\right|_+ &= \mathfrak{b}_{n}(k)\mathcal{T}_n^m,\notag\\
		\left. \partial_{\boldsymbol{\nu}}\mathbf{S}_{\partial D}^k[\mathcal{I}_{n-1}^m]\right|_+ &= \mathfrak{c}_{1n}(k)\mathcal{I}_{n-1}^m+\mathfrak{dd}_{1n}(k)\mathcal{N}_{n+1}^m,\label{eq:partial nu in I}\\
		\left. \partial_{\boldsymbol{\nu}}\mathbf{S}_{\partial D}^k[\mathcal{N}_{n+1}^m]\right|_+ &= \mathfrak{c}_{2n}(k)\mathcal{I}_{n-1}^m+ \mathfrak{d}_{2n}(k)\mathcal{N}_{n+1}^m. \label{eq:partial nu in N}
	\end{align}
	Here
	\begin{subequations}
		\begin{align}
			\mathfrak{b}_{n}(k) &= -\rmi k_s j_n(k_s)(k_s h_n^{\prime}(k_s)-h_n(k_s)), \notag\\
			\mathfrak{c}_{1n}(k) &= -2(n-1)\rmi \left(\frac{j_{n-1}(k_s)h_{n-1}(k_s)k_s(n+1)}{2n+1}+\frac{j_{n-1}(k_p)h_{n- 1}(k_p)k_p \mu n}{(\lambda+2\mu)(2n+1)} \right) \label{eq:c1n d1n def}\\
			& \quad +\rmi \left(\frac{j_{n-1}(k_s)h_n(k_s)k_s^2(n+1)+j_{n-1}(k_p)h_n(k_p)k_p^2 n}{2n+1} \right), \notag\\
			\mathfrak{d}_{1n}(k) &=2n(n+2)\rmi \left(\frac{j_{n-1}(k_s)h_{n+1}(k_s)k_s}{2n+1}-\frac{j_{n-1}(k_p)h_{n+1}(k_p)k_p \mu}{(\lambda+2\mu)(2n+1)} \right)\notag \\
			& \quad + n\rmi \left(\frac{-j_{n-1}(k_s)h_n(k_s)k_s^2+ j_{n-1}(k_p)h_n(k_p)k_p^2}{2n+1} \right), \notag\\
			\mathfrak{c}_{2n}(k) &= -2(n^2- 1)\rmi \left(\frac{j_{n+1}(k_s)h_{n-1}(k_s)k_s }{2n + 1}-\frac{j_{n+1}(k_p)h_{n-1}(k_p)k_p \mu}{(\lambda+2\mu)(2n+1)} \right) \notag\\
			& \quad -(n+1)\rmi \left(\frac{-j_{n+1}(k_s)h_n(k_s)k_s^2+ j_{n+1}(k_p)h_n(k_p)k_p^2}{2n+1} \right), \notag\\
			\mathfrak{d}_{2n}(k) &= 2(n+2)\rmi \left(\frac{j_{n+ 1}(k_s)h_{n+1}(k_s)k_s n}{2n+1}+\frac{j_{n+1}(k_p)h_{n+1}(k_p)k_p \mu (n+1)}{(\lambda+2\mu)(2n+1)} \right) \notag\\
			& \quad-\rmi \left(\frac{j_{n+1}(k_s)h_n(k_s)k_s^2 n+j_{n+1}(k_p)h_n(k_p)k_p^2 (n+1)}{2n+1} \right),\label{eq:d2n def}
		\end{align}
	\end{subequations}
	where $k_s$ and $k_p$ are given in \eqref{eq:ks kp defination}.
\end{lem}

\begin{lem}\cite{DLL,DLL2020}
	The eigenvalues of ~~$\mathbf{S}_{\partial D}^{k}$ corresponding to $\boldsymbol{\nu}$ on $\mathbb{R}^3\setminus\overline{D}$ are respectively given by
	\begin{align}\label{eq:S nu,K nu}
		\mathbf{S}_{\partial D}^{k}[\boldsymbol{\nu}](\mathbf{x}) = \frac{-\rmi k_p^2}{\lambda + 2\mu}j_n(k_p)h_n(k_p|\mathbf{x}|)\boldsymbol{\nu},
		\quad \mathbf{x} \in \mathbb{R}^3\setminus\overline{D},
	\end{align}
	where $k_p$ is as defined in \eqref{eq:ks kp defination}.
\end{lem}

In the following two lemmas, we establish the properties of the operators $\boldsymbol{\nu} \cdot \mathbf{S}_{\partial D}^{k}[Y_{n}^{m} \boldsymbol{\nu}]$ and $\boldsymbol{\nu} \cdot \left(\frac{\mathbf{I}}{2} + \mathbf{K}^{k,*}_{\partial D}\right)[Y_{n}^{m} \boldsymbol{\nu}]$ on $\partial D$, analyzing each operator sequentially with respect to  $k$.

\begin{lem}\cite[Lemma 3.5]{chengao}
	The operator $\boldsymbol{\nu} \cdot \mathbf{S}_{\partial D}^{k}[Y_{n}^{m} \boldsymbol{\nu}]$ on $\partial D$ satisfies the relation
	\begin{align}\label{eq:nu S acoustic}
		\boldsymbol{\nu} \cdot \mathbf{S}_{\partial D}^{k}[Y_{n}^{m} \boldsymbol{\nu}] = \alpha_{n}(k) Y_{n}^{m},
	\end{align}
	where  
	\begin{align}\label{eq:alpha expansion}
		\alpha_{n}(k) = \frac{n\left(c_{1n}(k) + c_{2n}(k)\right) + (n+1)\left(d_{1n}(k) + d_{2n}(k)\right)}{2n+1},
	\end{align}
	with $c_{1n}(k)$, $c_{2n}(k)$, $d_{1n}(k)$, and $d_{2n}(k)$ defined \eqref{eq:52} of Lemma \ref{lem:single on TIN}. Furthermore, for fixed $n \in \mathbb{N}$ and $  k \ll 1$, $\alpha_{n}(k)$ admits the asymptotic expansion:
	\begin{align}\notag
		\alpha_{n}(k) = -\frac{2(\lambda+\mu)n(n + 1) + \mu\left(4n^{4} + 4n - 1\right)}{\mu(\lambda + 2\mu)(2n + 3)(2n + 1)(2n - 1)} + \mathcal{O}(k^2\tau^2).
	\end{align}
\end{lem}

\begin{lem}\label{lem:nu cdot K}\cite[Lemma 3.6]{chengao}
	The operator $\boldsymbol{\nu} \cdot \left(\frac{\mathbf{I}}{2} + \mathbf{K}^{k,*}_{\partial D}\right)$ on $\partial D$ satisfies the relation
	\begin{align}\label{eq:nu K}
		\boldsymbol{\nu} \cdot \left(\frac{\mathbf{I}}{2} + \mathbf{K}^{k,*}_{\partial D}\right)[Y_{n}^{m} \boldsymbol{\nu}] = \beta_{n}(k) Y_{n}^{m},
	\end{align}
	where
	\begin{align}\label{eq:beta def}
		\beta_{n}(k) = \frac{n\left(\mathfrak{c}_{1n}(k) + \mathfrak{c}_{2n}(k)\right) + (n + 1)\left(\mathfrak{d}_{1n}(k) + \mathfrak{d}_{2n}(k)\right)}{2n + 1},
	\end{align}
	with $\mathfrak{c}_{1n}(k)$, $\mathfrak{c}_{2n}(k)$, $\mathfrak{d}_{1n}(k)$, and $\mathfrak{d}_{2n}(k)$ as defined in \eqref{eq:c1n d1n def}-\eqref{eq:d2n def} of Lemma \ref{lem:trace single in TIN}. Additionally, for fixed $n \in \mathbb{N}$ and $ k \ll 1$, $\beta_{n}(k)$ expands asymptotically as:
	\begin{align}\label{eq:betak expansion}
		\beta_{n}(k) = \beta_{n0} + \beta_{n2s}k_s^2 + \beta_{n2p}k_p^2 + \mathcal{O}\left[\frac{k^4\tau^4}{n^4}\right],
	\end{align}
		where  
	\begin{align}
		\beta_{n0}&=\frac{(\lambda + 2\mu)(2n+3)(2n^3+2n^2\mu-2n^3\mu+n)+2\left[n\lambda+\mu(3n+1)\right](n+2)(n+1)(2n-1)}{(\lambda+2\mu)(2n+3)(2n+1)^2(2n-1)}, \label{eq:betano def}\\
		\beta_{n2s}&=\frac{12n^3+18n^2+6n}{(2n+5)(2n+3)(2n+1)^2(2n-1)(-2n+3)},\notag\\
		\beta_{n2p}&=\frac{(\lambda+2\mu)\mu(4n^{4}+18n^{3}+8n^{2}-30n)
			+(\lambda+2\mu)(-8n^{3}-12n^{2}+26n+15)+\beta_{n2p1}}
		{(\lambda+2\mu)(2n+5)(2n+3)(2n+1)^{2}(2n-1)(-2n+3)},\label{eq:betan2p def}\\
		\beta_{n2p1}&=\mu(-4n^{4}+2n^{3}+22n^{2}-8n-24).\notag
	\end{align}
	and $k_s$ and $k_p$ are given in \eqref{eq:ks kp defination}.
	 
\end{lem}

Definition \ref{def:boundary localization} states that the interior total field ${u}|_{D}$ and the exterior scattered field $\mathbf{u}^s|_{\mathbb{R}^3 \setminus \overline{D}}$ of the scattering problem \eqref{eq:system2} exhibit boundary localization provided that \eqref{eq:226} holds. We now introduce the incident wave  for our subsequent analysis. It is emphasized that an appropriately choice of the incident wave plays important role in our main results.   For any $n \in \mathbb{N}_0$, let
\begin{equation}\label{eq:ui def}
	\mathbf{u}^{i} = \sum_{m=-n}^n f_{n,m} \, j_n\left(k_p|\mathbf{x}|\right) \mathcal{I}_{n-1}^m,
\end{equation}
where $k_p$ and $\mathcal{I}_{n-1}^m$ are defined in \eqref{eq:ks kp defination} and \eqref{eq:Inm Tnm Nnm definition}, respectively, and $(f_{n,-n},\ldots,f_{n,n}) \in \mathbb{C}^{2n+1}$ is a nonzero vector. The index $n$  of the incident wave defined in \eqref{eq:ui def} plays a crucial role in subsequent analyses encompassing the quasi-Minnaert resonance and the stress concentration phenomenon. Building upon results from \cite{BS81,Dassions,DLL2020}, the incident field $\mathbf{u}^{i}$ in \eqref{eq:ui def} constitutes an entire solution to \eqref{eq:incident elastic wave}, specifically designed to generate the boundary localization behavior of the total field $u$ within $D$ and the scattered field $\mathbf{u}^s$ in $\mathbb{R}^3\setminus\overline{D}$, as rigorously analyzed for the coupled system \eqref{eq:system2} in Theorem \ref{thm:th3.1}.

The following lemma provides asymptotic expressions for $u|_{D}$ and $\mathbf{u}^s|_{\mathbb{R}^3\setminus\overline{D}}$ with respect to the sub-wavelength frequency $k$ and high-contrast density ratio $\delta$, where $k \ll 1$ and $\delta \ll 1$. Here, $\delta$ denotes the density ratio between the bubble $D$ and the surrounding elastic medium, as defined in \eqref{eq:delta tau}.

\begin{lem}\label{lem:3.6}
	Consider the bubble-elastic scattering problem \eqref{eq:system2} with the incident wave defined in \eqref{eq:ui def}. Under the assumptions \eqref{eq:assume} and \eqref{eq:ks kp defination}, for $k \ll 1$ and sufficiently large $n$ in \eqref{eq:ui def}, the interior total field $u$ in $D$ and the scattered field $\mathbf{u}^s$ in $\mathbb{R}^3 \setminus \overline{D}$ admit the following asymptotic expansions with respect to $k$ and $\delta$:
	\begin{align}
		u|_D &= \sum_{m=-n}^{n} \frac{-f_{n,m}n^{3}k^{n}|\bmf{x}|^n}{(2n + 1)!!(2n+1)(\lambda+2\mu)^{n/2}\delta\tau^{2}}
		\left(\mathcal{O}\left(1\right) + \mathcal{O}\left(\delta\tau^2 k^2\right)\right)Y_{n}^{m}(\theta,\varphi), \label{eq:total field u} \\
		\mathbf{u}^s|_{\mathbb{R}^3 \setminus \overline{D}} &= \sum_{m=-n}^{n} \frac{-f_{n,m}n^{2}k^{n+1}\tau}{(2n + 1)!!(2n+1)(\lambda+2\mu)^{(n+3)/2}|\bmf{x}|^{n+1}}
		\left(\mathcal{O}\left(1\right) + \mathcal{O}\left(\delta\tau^2 k^2\right)\right)Y_{n}^{m}(\theta,\varphi)\boldsymbol{\nu}. \label{eq:scattered field}
	\end{align}
\end{lem}

\begin{proof}
	Given that the integral equation \eqref{eq:2.20} is equivalent to \eqref{eq:system2} and the incident wave in \eqref{eq:ui def}, utilizing the orthogonality of $\mathcal{T}_n^m$, $\mathcal{I}_n^m$, and $\mathcal{N}_n^m$,  the density function $\boldsymbol{\varphi_e}$ can be written as:
	 \begin{align}\label{eq:varphie}
	 	\boldsymbol{\varphi_e}=\varphi_e \boldsymbol{\nu}.
	 \end{align}
   with $\varphi_e\in L^2 (\partial D)$ is unknown. 
   Substituting \eqref{eq:varphie} into \eqref{eq:tilde A k delta}, there holds that
	\begin{align}\notag
		\tilde{\mathcal{A}}(k,\delta)[\varphi_e \boldsymbol{\nu}](\mathbf{x})=\tilde{\bmf{F}}(\mathbf{x}),
	\end{align}
	where $	\tilde{\mathcal{A}}(k,\delta)$ and $\tilde{\bmf{F}}(\mathbf{x})$ are as seen in \eqref{eq:tilde Ak delta} and \eqref{eq:tilde F}, respectively. Combining with \eqref{eq:tilde Ak delta}, \eqref{eq:tilde F}, \eqref{eq:S in D}, \eqref{eq:nu S acoustic}, \eqref{eq:nu K}, \eqref{eq:ui def}, \eqref{eq:nu TIN} and \eqref{eq:partial nu in I}, we can derive that
\begin{align}
	&\tilde{\mathcal{A}}(k,\delta)[\varphi_e \boldsymbol{\nu}](\mathbf{x})
	= \left(\frac{k j_n^{\prime}(k)}{j_n(k)} \beta_{n}(k)+\delta \tau^2 k^2 \alpha_{n} (k)\right) 
    \varphi_e, \notag\\
	&\tilde{\bmf{F}}(\mathbf{x})
	=\sum_{m=-n}^n f_{n,m}
	 \left[-\delta\tau^2 k^2 n j_n(k_p)
	 -\left(\left(\lambda+2\mu\right)n k_p j_n^{\prime}(k_p) 
	 +\lambda n(1-n)j_n(k_p)\right)\frac{k j_n^{\prime}(k)}{j_n(k)}\right] Y_n^m. \notag
\end{align}
By employing the aforementioned formulas, namely \eqref{eq:varphi b}, \eqref{eq:nu K}, \eqref{eq:nu TIN}, and \eqref{eq:partial nu in I}, we obtain
\begin{align}\label{eq:varphi_b}
	\boldsymbol{\varphi_e}
	=\varphi_e \boldsymbol{\nu}, \qquad
	\varphi_b=\frac{\beta_{n}(k)\varphi_e}{\rmi\delta\tau^2 k j_n(k)h_n(k)}-\sum_{m=-n}^n f_{n,m} n j_n(k_p) Y_n^m,
\end{align}
where 
\begin{align}\label{eq:varphi e def}
	\varphi_e=\sum_{m=-n}^n
	\frac{f_{n,m}\left[-\delta\tau^2 k^2 n j_n(k_p) j_n(k)
		-\left(\left(\lambda+2\mu\right)n k_p j_n^{\prime}(k_p) 
		+\lambda n(1-n)j_n(k_p)\right)k j_n^{\prime}(k)\right]}{k j_n^{\prime}(k)\beta_{n}(k)+\delta\tau^2k^2\alpha_n(k)j_n(k)} Y_n^m. 
\end{align}
By substituting \eqref{eq:ks kp defination}, \eqref{eq:jn expansion}-\eqref{eq:hn qiudao}, \eqref{eq:alpha expansion}, and \eqref{eq:betak expansion} into \eqref{eq:varphi e def}, we adopt the approach of asymptotic expansion with respect to  $k$. After rearranging the terms, we obtain the following result:
\begin{align}\label{eq:varphie=}
	 \boldsymbol{\varphi_e}
	=\sum_{m=-n}^n \frac{f_{n,m} A_{n,\lambda,\mu} k^{2n}
	\left(2 n^2 \mu +\frac{-n^2\mu(\lambda+2 \mu +1)+\left[\left(1-2\mu -2\delta\tau^2\right)n-3\delta^2\tau\right]\left(\lambda+2\mu\right)}{(2n+3)(\lambda+2\mu)}k^2+\mathcal{O}\left(k^4\right)\right)}
	{a_{n,\lambda,\mu} k^n+b_{n,\lambda,\mu}k^{n+2}+c_{n,\lambda,\mu}\delta\tau^2 k^{n+2}+\mathcal{O}\left(\frac{k^{n+4}}{(2n+3)!!}\right)} Y_n^m \boldsymbol{\nu},
\end{align}
where 
\begin{subequations}
	\begin{align}
		&a_{n,\lambda,\mu}=\frac{(\lambda+2\mu)\mu(-4n^5-2n^4+6n^3)+(\lambda+2\mu)(8n^5+16n^4+12n^3-5n^2)+2\mu a_{n,\lambda,\mu,1}}
		{(\lambda+2\mu)(2n+1)^2(2n-1)(2n+3)!!},\label{eq:a_n lambda mu}\\
		&b_{n,\lambda,\mu}=-\frac{(\lambda+2\mu)\mu(-4n^5-10n^4+2n^3+12n^2)+(\lambda+2\mu)b_{n,\lambda,\mu,1}+2\mu b_{n,\lambda,\mu,2}}
		{2(\lambda+2\mu)(2n+3)(2n+1)^2(2n-1)(2n+3)!!},\notag\\
		&c_{n,\lambda,\mu}=-\frac{2(\lambda+\mu)n(n+1)+\mu(4n^4+4n-1)}{\mu(\lambda+2\mu)(2n+1)(2n-1)(2n+3)!!},\quad 	A_{n,\lambda,\mu}=\frac{n}{[(2n+1)!!]^2(\lambda+2\mu)^{\frac{n}{2}}}, \label{eq:Anlambda}
	\end{align}
\end{subequations}

with 
\begin{align}
	a_{n,\lambda,\mu,1}&=2n^5+7n^4+6n^3-n^2-2n,\quad b_{n,\lambda,\mu,1}=8n^5+32n^4+44n^3+19n^2-10n,\notag\\
	b_{n,\lambda,\mu,2}&=2n^5+11n^4+20n^3+11n^2-4n-4.\notag
\end{align}	
Based on the assumption \eqref{eq:assume}, and combining the definitions of \eqref{eq:a_n lambda mu}-\eqref{eq:Anlambda}, it is readily evident that:
 \begin{align}\label{eq:a,b,c O}
 	a_{n,\lambda,\mu}=\mathcal{O}\left(\frac{1}{(2n-1)!!}\right), \quad b_{n,\lambda,\mu}=\mathcal{O}\left(\frac{1}{(2n+1)!!}\right), \quad
 	c_{n,\lambda,\mu}=\mathcal{O}\left(\frac{1}{(2n-1)!!}\right).
 \end{align}
 Combining \eqref{eq:assume}, \eqref{eq:varphie=}, \eqref{eq:Anlambda}, and \eqref{eq:a,b,c O}, when $n$ is sufficiently large, we can derive the asymptotic expansion of $\boldsymbol{\varphi_e}$ with respect to $k$ as follows:
\begin{align}
     \boldsymbol{\varphi_e}
   &=\sum_{m=-n}^n \frac{f_{n,m} A_{n,\lambda,\mu} k^{2n}
   	\left(2 n^2 \mu +\frac{-n^2\mu(\lambda+2 \mu +1)+\left[\left(1-2\mu -2\delta\tau^2\right)n-3\delta^2\tau\right]\left(\lambda+2mu\right)}{(2n+3)(\lambda+2\mu)}k^2+\mathcal{O}\left(k^4\right)\right)}
   {a_{n,\lambda,\mu} k^n+b_{n,\lambda,\mu}k^{n+2}+c_{n,\lambda,\mu}\delta\tau^2 k^{n+2}+\mathcal{O}\left(\frac{k^{n+4}}{(2n+3)!!}\right)} Y_n^m \boldsymbol{\nu} \notag\\
   &=\sum_{m=-n}^n \frac{ f_{n,m} A_{n,\lambda,\mu} 2 n^2 \mu k^{2n} \left(1+\frac{-n^2\mu(\lambda+2 \mu +1)+\left[\left(1-2\mu -2\delta\tau^2\right)n-3\delta^2\tau\right]\left(\lambda+2\mu\right)}{2n^2(2n+3)(\lambda+2\mu)\mu}k^2+\mathcal{O}\left(\frac{k^4}{n^2}\right)\right)}
   {a_{n,\lambda,\mu} k^n \left(1+\frac{b_{n,\lambda,\mu}}{a_{n,\lambda,\mu}}k^{2}+\frac{c_{n,\lambda,\mu}\delta\tau^2}{a_{n,\lambda,\mu}} k^{2}
   +\mathcal{O}\left(\frac{k^{4}}{(2n+1)(2n+3)}\right)\right)} Y_n^m \boldsymbol{\nu}  \notag\\
   &=\sum_{m=-n}^n f_{n,m}\frac{2 n^2 \mu A_{n,\lambda,\mu}}{a_{n,\lambda,\mu}} k^n \left(1+\frac{-n^2\mu(\lambda+2 \mu +1)+\left[\left(1-2\mu -2\delta\tau^2\right)n-3\delta^2\tau\right]\left(\lambda+2\mu\right)}{2n^2(2n+3)(\lambda+2\mu)\mu}k^2\right.\notag\\
   &\quad\left.+\mathcal{O}\left(\frac{k^4}{n^2}\right)\right)
   \left(1-\frac{c_{n,\lambda,\mu}\delta\tau^2}{a_{n,\lambda,\mu}} k^{2}+\mathcal{O}\left(\delta^2\tau^4 k^4\right)\right)   
   Y_n^m \boldsymbol{\nu}  \notag\\
   &=\sum_{m=-n}^n \frac{2 f_{n,m} n^2 \mu A_{n,\lambda,\mu}}{a_{n,\lambda,\mu}} k^n \left(1 +\mathcal{O}\left(\delta\tau^2 k^2\right)\right)Y_n^m\boldsymbol{\nu}. \notag\\
   &=\sum_{m=-n}^n \frac{f_{n,m} n^2  k^n}{(2n+1)!!(\lambda+2\mu)^{\frac{n}{2}}}(\mathcal{O}\left(1\right)+\mathcal{O}\left(\delta\tau^2 k^2\right))Y_n^m\boldsymbol{\nu} .\label{eq:varphi_e}
\end{align}
By applying a similar argument to $\boldsymbol{\varphi_e}$, it follows from \eqref{eq:jn expansion}, \eqref{eq:hn expansion}, \eqref{eq:varphi_b}, and \eqref{eq:betak expansion} that the asymptotic expansion of $\varphi_b$ with respect to $k$ is as follows:
	\begin{align}
		\varphi_b
		&= \sum_{m = -n}^{n} f_{n,m}  \frac{2(2n + 1)n^2 A_{n,\lambda,\mu}(\lambda + 2\mu)^{\frac{n}{2}}(2n + 1)!! - n a_{n,\lambda,\mu}\delta\tau^{n + 2}}{a_{n,\lambda,\mu}\delta\tau^2(\lambda + 2\mu)^{\frac{n}{2}}(2n + 1)!!} k^n Y_n^m  \notag \\
		&\quad +\mathcal{O}\left(\frac{n^3 k^{n + 2}}{(2n + 1)!!(\lambda + 2\mu)^{\frac{n}{2}}}\right) \notag\\
		&=\sum_{m=-n}^{n} \frac{ f_{n,m} n^{3} k^{n}}{(2n + 1)!!(\lambda+2\mu)^{\frac{n}{2}}\delta\tau^{2}}\left(\mathcal{O}\left(1\right)+\mathcal{O}\left(\delta\tau^2 k^2\right)\right)Y_{n}^{m}, \label{eq:varphib}
	\end{align}
	where $a_{n,\lambda,\mu}$ and $A_{n,\lambda,\mu}$ are given by \eqref{eq:a_n lambda mu} and \eqref{eq:Anlambda}, respectively. By virtue of \eqref{eq:u express}, \eqref{eq:S in D} and \eqref{eq:varphib} we obtain \eqref{eq:total field u} . Similarly, using formulas \eqref{eq:u express}, \eqref{eq:S nu,K nu} and \eqref{eq:varphi_e}, one can derive \eqref{eq:scattered field}.
	
		The proof is complete.
	\end{proof}

\medskip

\subsection{Boundary localization of $u|_D$ and $\mathbf{u}^s|_{\mathbb{R}^3 \setminus \overline{D}}$ }

In Theorem \ref{thm:th3.1}, under the assumptions \eqref{eq:assume} and \eqref{eq:ks kp defination}, we shall rigorously prove that both the interior total field $u|_D$ and the exterior scattered field $\mathbf{u}^s|_{\mathbb{R}^3 \setminus \overline{D}}$ to \eqref{eq:system2} are boundary localized in the sense of Definition \ref{def:boundary localization}, provided that the incident wave $\mathbf{u}^i$ is chosen as defined by \eqref{eq:ui def} with a sufficiently large $n$.

\begin{thm}\label{thm:th3.1}
	
	Consider the bubble-elastic scattering problem \eqref{eq:system2} characterizing the bubble $(D; \kappa, \rho_b)$ embedded in a homogeneous elastic medium $(\mathbb{R}^3 \setminus \overline{D}; \lambda, \mu, \rho_e)$, where $D$ is a unit ball centered at the origin in $\mathbb{R}^3$. Let $B_R \subset \mathbb{R}^3$ denote a ball of radius $R$ containing $D$. Let the corresponding internal total wave field $u$ in $D$ and external scattered field $\mathbf{u}^s$ in $\mathbb{R}^3 \setminus D$ denote the solutions to \eqref{eq:system2} governed by the incident wave $\mathbf{u}^i$ as defined in \eqref{eq:ui def}. We recall the definitions of $\mathscr{S}_{-}^{1 - \zeta_1}(\partial D)$ and $\mathscr{S}_{+}^{\zeta_2 - 1}(\partial D)$ in \eqref{eq:Sdef}, where $\zeta_1 \in (0, 1)$ and $\zeta_2 \in (1, R)$ are positive constants.  
	
	Under the assumptions \eqref{eq:assume} and \eqref{eq:ks kp defination}, for any $\zeta_1 $, $\zeta_2  $, and sufficiently small $\eta \in \mathbb{R}_+$, if the index $n$ corresponding to $\mathbf{u}^i$ as defined in \eqref{eq:ui def} satisfies
	\begin{align}\label{eq:n1 n2 def}
		n \geqslant N_1= \max\{n_1, n_2\}, \quad \text{with} \quad 
		n_1 = \frac{1}{2}\left(\frac{\ln \eta}{\ln \zeta_1} - 3\right), \quad 
		n_2 = \frac{1}{2}\left(1 - \frac{\ln \eta}{\ln \zeta_2}\right).
	\end{align}
	then 
	\begin{align}\label{eq:328}
		\frac{\|{u}\|_{L^2 (D \setminus \mathscr{S}_{-}^{1-\zeta_1}(\partial D) )^3}^2}{\|	{u}\|_{L^2(D)^3}^2}
		\leqslant\mathcal{O}\left(\eta\right)+\mathcal{O}\left(\eta\delta\tau^2 k^2\right),
		\quad  
		\frac{\|\mathbf{u}^s\|_{L^2 ( (B_R \setminus \overline  D)\setminus \mathscr{S}_{+}^{\zeta_2-1}(\partial D)   )^3}^2}{\| \mathbf{u}^s\|_{L^2(B_R \backslash \overline D)^3}^2}
		\leqslant\mathcal{O}\left(\eta\right) +\mathcal{O}\left(\eta\delta\tau^2 k^2\right).
	\end{align}
	The parameter $\eta$ quantifies the level of boundary localization.  
\end{thm}

\begin{proof}

  Firstly, we derive the asymptotic analysis for the $L^2$-norm of the internal total field ${u}|_D$ within $D\setminus \mathscr{S}_{-}^{1-\zeta_1}(\partial D)$ with respect to the parameter $k$. Since $\mathbf u^i$ is defined by \eqref{eq:ui def}, Lemma \ref{lem:3.6} provides the explicit expression of ${u}|_{D}$ as given in \eqref{eq:total field u}. By the orthogonality of ${Y}_n^m$ for distinct $n$ and $m$, and in view of \eqref{eq:S in D} and \eqref{eq:ynm3.2}, for any $\zeta_1 \in (0,1)$, under the assumptions \eqref{eq:assume} and \eqref{eq:ks kp defination}, it follows that 
	\begin{align}\label{eq:u L2norm}
		&\| { u}\|_{L^2(D\setminus \mathscr{S}_{-}^{1-\zeta_1}(\partial D))^3}^2\notag\\
		=&\int_{D\setminus \mathscr{S}_{-}^{1-\zeta_1}(\partial D)}
		\left|\sum_{m=-n}^{n}  \frac{-f_{n,m}n^{3}k^{n}|\bmf{x}|^n}{(2n + 1)!!(2n+1)(\lambda+2\mu)^{\frac{n}{2}}\delta\tau^{2}}Y_{n}^{m}\left(\mathcal{O}\left(1\right)+\mathcal{O}\left(\delta\tau^2 k^2\right)\right)
		\right|^2 \rmd \mathrm{x} \notag\\
		=&\int_0^{2 \pi} \int_0^\pi \int_0^{\zeta_1}  \sum_{m=-n}^{n}\left[\frac{f_{n,m} n^3 r^{n}k^n}{(2n+1)!!(2n+1)(\lambda+2\mu)^{\frac{n}{2}}\delta\tau^2}
		\left(\mathcal{O}\left(1\right)+\mathcal{O}\left(\delta\tau^2 k^2\right)\right)\right]^2 \cdot r^2 \sin\theta \rmd r \rmd \theta \rmd \varphi \notag\\
		=&\sum_{m=-n}^n  \frac{4\pi |f_{n,m}|^2  n^6 k^{2 n} }{[(2 n+1) ! !]^2 (2n+1)^2 (\lambda+2\mu)^{n} \delta^2 \tau^4} \int_0^{\zeta_1} r^{2 n+2} \rmd r 	
		\left(\mathcal{O}\left(1\right)+\mathcal{O}\left(\delta\tau^2 k^2\right)\right) \notag\\
		=&K_{n,\lambda,\mu,\delta,\tau} \zeta_{1}^{2n+3} k^{2 n} \left(\mathcal{O}\left(1\right)+\mathcal{O}\left(\delta\tau^2 k^2\right)\right),
	\end{align}
	where  
	\begin{align}\notag
		K_{n,\lambda,\mu,\delta,\tau}&=\frac{4\pi  n^6  }{[(2 n+1) ! !]^2 (2n+3) (2n+1)^2 (\lambda+2\mu)^{n} \delta^2 \tau^4} \sum_{m=-n}^n  |f_{n,m}|^2. 
	\end{align}
	Here, $\delta$ denotes the ratio of the density of the bubble to the density of the elastic medium. $\tau$ represents the wave speed contrast between the air bubble $D$ and the elastic medium $\mathbb{R}^3 \setminus D$, as defined in \eqref{eq:delta tau}, and satisfies \eqref{eq:assume}. Using a similar argument to that for \eqref{eq:u L2norm}, one can show that 
	\begin{align}\label{eq:u3.35}
	   \|{u}\|_{L^2\left(D\right)^3}^2=K_{n,\lambda,\mu,\delta,\tau}  k^{2 n} \left(\mathcal{O}\left(1\right)+\mathcal{O}\left(\delta\tau^2 k^2\right)\right).
	\end{align}
	Since $n > n_1$, where $n_1$ is defined in \eqref{eq:n1 n2 def}, it follows that $\zeta_1^{2 n+3} < \eta$. Therefore, combining \eqref{eq:u L2norm} and \eqref{eq:u3.35}, it yields that
	\begin{align}\label{eq:u concentrate}
		\frac{\|{u}\|_{L^2\left(D\setminus \mathscr{S}_{-}^{1-\zeta_1}(\partial D)\right)^3}^2}{\|{u}\|_{L^2\left(D\right)^3}^2}
		&=\frac{K_{n,\lambda,\mu,\delta,\tau} \zeta_{1}^{2n+3} k^{2 n} \left(\mathcal{O}\left(1\right)+\mathcal{O}\left(\delta\tau^2 k^2\right)\right)}{K_{n,\lambda,\mu,\delta,\tau}  k^{2 n} \left(\mathcal{O}\left(1\right)+\mathcal{O}\left(\delta\tau^2 k^2\right)\right)}\\
		&=\zeta_1^{2 n+3}
		\left(\mathcal{O}\left(1\right)+\mathcal{O}\left(\delta\tau^2 k^2\right)\right)
		\left(\mathcal{O}\left(1\right)-\mathcal{O}\left(\delta\tau^2 k^2\right)\right) \leqslant \mathcal{O}\left(\eta\right)+\mathcal{O}\left(\eta\delta\tau^2 k^2\right),\notag
	\end{align}
	Here, the parameter $\eta$ quantifies the level of boundary localization. We demonstrate that the energy of the internal total field ${u}|_D$ in \eqref{eq:system2} is predominantly concentrated in the domain $\mathscr{S}_{-}^{1-\zeta_1}(\partial D)$ of $D$, as indicated by Definition \ref{def:boundary localization}.

	 Furthermore, we show that the $L^2$-norm of the external scattered field $\mathbf{u}^s|_{\mathbb{R}^3 \setminus D}$ is concentrated in the domain $\mathscr{S}_{+}^{\zeta_2-1}(\partial D)$ outside $D$. Using a similar argument to that for $\| {u}\|_{L^2(D\setminus \mathscr{S}_{-}^{1-\zeta_1}(\partial D))^3}^2$, one can show that 
	\begin{align}
		&\|\mathbf{u}^s\|_{L^2((B_R \setminus \overline  D)\setminus \mathscr{S}_{+}^{\zeta_2-1}(\partial D) )^3}^2 \notag\\
		=&\int_{(B_R \setminus \overline  D)\setminus \mathscr{S}_{+}^{\zeta_2-1}(\partial D) } 
		\left|\sum_{m=-n}^{n}  \frac{-f_{n,m} n^{2} k^{n+1} \tau}{(2n + 1)!!(2n+1)(\lambda+2\mu)^{\frac{n+3}{2}}|\bmf{x}|^{n+1}}\left(\mathcal{O}\left(1\right)+\mathcal{O}\left(\delta\tau^2 k^2\right)\right)Y_{n}^{m}\boldsymbol{\nu}\right|^2\rmd \mathrm{x} \notag\\
	    =&\int_{0}^{2\pi}\int_{0}^{\pi}\int_{\zeta_2}^{R}\sum_{m = -n}^{n}\frac{\vert f_{n,m}\vert^2 n^4 k^{2n + 2}\tau^2\sin\theta}{[(2n+1)!!]^2(2n + 1)^2(\lambda + 2\mu)^{n+3}r^{2n}}(O(1)+O(\delta\tau^2k^2))\mathrm{d}r\mathrm{d}\theta\mathrm{d}\varphi \notag\\
		=&\sum_{m = -n}^{n}\frac{4\pi\vert f_{n,m}\vert^2 n^4 k^{2n + 2}\tau^2}{[(2n+1)!!]^2(2n + 1)^2(\lambda + 2\mu)^{n+3}}\int_{\zeta_2}^{R} r^{2n}\mathrm{d}r\left(\mathcal{O}\left(1\right)+\mathcal{O}\left(\delta\tau^2 k^2\right)\right)\notag\\
		=& \frac{G_{n,\lambda,\mu,\tau}\left(R^{2n-1}-{\zeta_2}^{2 n-1}\right)}{{\zeta_2}^{2 n-1} R^{2n-1}} k^{2n+2} \left(\mathcal{O}\left(1\right)+\mathcal{O}\left(\delta\tau^2 k^2\right)\right), \label{eq:us3.37}
	\end{align}
	where
	\begin{align}
		G_{n,\lambda,\mu,\tau}&=\frac{4\pi n^4 \tau^2}{[(2n+1)!!]^2(2n+1)^2 (2n-1)
			(\lambda + 2\mu)^{n+3}}\sum_{m=-n}^n |f_{n,m}|^2. \notag
	\end{align}
   	Using a similar argument to that for \eqref{eq:us3.37}, it is evident that
   	\begin{align}\label{eq:us33}
   		 \| \mathbf{u}^s\|^2_{L^2\left(B_R \backslash \overline D\right)^3}= \frac{	G_{n,\lambda,\mu,\tau}\left(R^{2n-1}-1\right)}{ R^{2n-1}} k^{2n+2}
   		\left(\mathcal{O}\left(1\right)+\mathcal{O}\left(\delta\tau^2 k^2\right)\right).
   	\end{align}
   Since $n \geqslant n_2$, where $n_2$ is defined in \eqref{eq:n1 n2 def}, it follows that $\frac{1}{\zeta_2^{2 n-1}} < \eta$. Therefore, combining \eqref{eq:us3.37} and \eqref{eq:us33}, we obtain 
   \begin{align}\label{eq:us concentrate}
   	\frac{\| \mathbf{u}^s\|_{L^2((B_R \setminus \overline  D)\setminus \mathscr{S}_{+}^{\zeta_2-1}(\partial D) )^3}^2}{\| \mathbf{u}^s\|^2_{L^2\left(B_R \backslash \overline D\right)^3}}
   	=&\frac{\left(R^{2n-1}-{\zeta_2}^{2 n-1}\right)}{\zeta_2^{2n-1}\left(R^{2n-1}-1\right)}
   	 \left(\mathcal{O}\left(1\right)+\mathcal{O}\left(\delta\tau^2 k^2\right)\right) \left(\mathcal{O}\left(1\right)-\mathcal{O}\left(\delta\tau^2 k^2\right)\right) \\
   	=&\frac{1-\left(\frac{\zeta_{2}}{R}\right)^{2n-1}}{\zeta_2^{2n-1}\left(1-\frac{1}{R^{2n-1}}\right)}
   	\left(\mathcal{O}\left(1\right)+\mathcal{O}\left(\delta\tau^2 k^2\right)\right) \left(\mathcal{O}\left(1\right)-\mathcal{O}\left(\delta\tau^2 k^2\right)\right) \notag\\
   	=&\frac{1}{\zeta_2^{2n-1}}\left(\mathcal{O}\left(1\right)+\mathcal{O}\left(\delta\tau^2 k^2\right)\right) 
   	\leqslant \mathcal{O}\left(\eta\right) +\mathcal{O}\left(\eta\delta\tau^2 k^2\right). \notag
   \end{align}

	The proof is complete.
\end{proof}

\begin{rem}\label{rem:31}

   To induce stress concentration in a neighborhood of the bubble’s boundary, we first localize the wave  near the bubble’s boundary. In Theorem \ref{thm:th3.1}, using \eqref{eq:u concentrate} and \eqref{eq:us concentrate}, it can be readily shown that both the interior total field $u|_D$ and the exterior scattered field $\mathbf{u}^s|_{\mathbb{R}^3 \setminus \overline{D}}$ to \eqref{eq:system2} are boundary localized in the sense of Definition \ref{def:boundary localization}, provided that the incident wave $\mathbf{u}^i$ is chosen as defined by \eqref{eq:ui def} with $n$ satisfying \eqref{eq:n1 n2 def}, and $D$ is a unit ball. This theorem represents the first rigorous attempt to characterize the physical wave patterns associated with high-contrast parameter of the coupled acoustic-elastic  PDE system \eqref{eq:system2} in radial geometry using rigorously  analysis in sub-wavelength region. Specifically, this result holds for  radial geometry, a structure frequently encountered in metamaterial science \cite{HL23,M1933,SD23}. From Theorem~\ref{thm:th3.1}, it is clear that the interior total field \( u|_D \) and the exterior scattered field \(\mathbf{u}^s|_{\mathbb{R}^3 \setminus \overline{D}}\) are confined to the domains \(\mathscr{S}_{-}^{1 - \zeta_1}(\partial D)\) and \(\mathscr{S}_{+}^{\zeta_2 - 1}(\partial D)\), respectively, as defined in \eqref{eq:Sdef} near \(\partial D\). Consequently, these waves remain localized near \(\partial D\) and do not propagate significantly away from the boundary. This type of wave propagation is consistent with phenomena observed in the literature \cite{HL23,SD23}. We use the extensive numerical examples to show the findings of this theorem can extend to more general shapes. A detailed discussion of this aspect can be found in Section \ref{sec:5}.   In future work, we will provide a rigorous analysis of bubble-elastic scattering problems for bubbles with more general shapes.

\end{rem}

\medskip

\subsection{Surface resonance of $u|_D$ and $\mathbf{u}^s|_{\mathbb{R}^3 \setminus \overline{D}}$ }

Furthermore, to characterize the high oscillation of the internal total field \( u|_D \) and the external scattered field \( \mathbf{u}^s|_{\mathbb{R}^3 \setminus \overline{D}} \) of the  bubble-elastic scattering system~\eqref{eq:system2}, we rigorously analyze and characterize the surface resonance phenomena in \( u|_D \) and \( \mathbf{u}^s|_{\mathbb{R}^3 \setminus \overline{D}} \) in Theorem~\ref{thm:nabla u in thm}. The occurrence of surface resonance in the internal total field \( u|_D \) and the external scattered field \( \mathbf{u}^s|_{\mathbb{R}^3 \setminus \overline{D}} \) depends on the index parameter \( n \) of the incident wave in~\eqref{eq:ui def} and the density ratio  \( \delta \).

In the subsequent analysis, we denote \(W_0\) by the principal branch of the Lambert \(W\)-function, which satisfies $ W_0(z) e^{W_0(z)} = z $ with $z\in \mathbb C$. It is noted that when     $   z \in \mathbb R  $ and $  z \geq -1/e $, then $   W_0(z)   $ is a single real valued function in $[-1/e, \infty)$ \cite{CG96}.

\begin{thm}\label{thm:nabla u in thm}

	Consider the bubble-elastic scattering problem~\eqref{eq:system2}, where the bubble \( (D;\kappa,\rho_b) \) is embedded in a homogeneous elastic medium \( (\mathbb{R}^3 \setminus \overline{D}; \lambda,\mu,\rho_e) \). We recall that \( \mathscr{S}_{-}^{1 - \zeta_1}(\partial D) \) and \( \mathscr{S}_{+}^{\zeta_2 - 1}(\partial D) \), which are defined in~\eqref{eq:Sdef} and \( \zeta_1 \in (0, 1) \) and \( \zeta_2 \in (1, R) \). Recall that \( \tau \) represents the wave speed ratio, satisfying~$0<\tau<1$. Under the assumptions~\eqref{eq:assume} and~\eqref{eq:ks kp defination}, for any given large positive real number $M\gg 1$, let 
    \begin{equation}\label{eq:N2 def}
      N_2= \max\left\{n_3, n_4 \right\}, 
    \end{equation}
    with
   \begin{align}\label{eq:n3n4 def}
   n_3 = \frac{2}{-\ln \tau}  W_0\left( \frac{-\tau  \ln \tau  \sqrt{3 M \delta}}{2   (1 - \zeta_1)^{1/4}} \right), 
   \quad
   n_4 = \frac{1}{-\ln \tau} \, W_0\left(    \frac{-3 \sqrt{3 \zeta_2}  \ln \tau (\lambda + 2\mu)^{3/2} M}{k \tau \sqrt{10(\zeta_2 - 1)} } \right),
   \end{align}
   where \( \delta \), defined in~\eqref{eq:delta tau}, denotes the density ratio between the bubble \( D \) and the background elastic medium \( \mathbb{R}^3 \setminus \overline{D} \), satisfying \( \delta \ll 1 \).  Here, \( W_0 \) denotes the principal branch of the Lambert \( W \)-function. If     the parameter \( n \) of the incident wave \( \mathbf{u}^i \), as defined in~\eqref{eq:ui def}, satisfies $ n \geqslant N_2$, then surface resonance of  $u|_D$ and $\mathbf u^s |_{\mathbb{R}^3 \setminus \overline D  }$  occurs. Furthermore, it holds that
   \begin{align}\label{eq:nablau nabla us gg 1} 
		\frac{\|\nabla u\|_{L^2(\mathscr{S}_{-}^{1 - \zeta_1}(\partial D))^3}}{\| \bmf u^i\|_{L^2(D)^3}}
	 	\geqslant \frac{n^2\sqrt{1-\zeta_1}}{3 \tau^{n+2}\delta}  > M\gg 1,
	 \end{align}
  and
  \begin{align}\label{eq:43}
		\frac{\|\nabla \bmf u^s\|_{L^2(\mathscr{S}_{+}^{\zeta_2 - 1}(\partial D))^3}}{\| \bmf u^i\|_{L^2(D)^3}}
		\geqslant   \frac{n k\sqrt{10\left(\zeta_2-1\right)}}{3\sqrt{3}\sqrt{\zeta_2}(\lambda + 2\mu)^{3/2}\tau^{n - 1}} > M\gg 1.
  \end{align}
 Here, $M$ is defined as the parameter characterizing the high oscillation level of the wave.

\end{thm}

\begin{proof}
	In what follows, the proof is structured in two steps. Initially, we establish \eqref{eq:nablau nabla us gg 1}. In order to investigate $\|\nabla   u\|_{L^2(\mathscr{S}_{-}^{1 - \zeta_1}(\partial D))^3}$, we first need to derive the  asymptotic analysis for $\nabla   u|_{\mathscr{S}_{-}^{1 - \zeta_1}(\partial D)}$ with respect to $k$. It turns out that we need calculate $\nabla\left(j_n(k r) Y_n^m(\theta, \varphi)\right)$   , where $Y_n^m(\theta, \varphi)$ is defined by \eqref{eq:ynm and cnm def}. In the following, let us derive the explicit expression for $\nabla\left(j_n(k r) Y_n^m(\theta, \varphi)\right)$.
	\begin{align}\notag
		\nabla \left(j_n(kr)Y_n^m(\theta, \varphi) \right)
		= j_n^{\prime}(kr) k Y_n^m(\theta, \varphi)\hat{r}+ \frac{j_n(kr)}{r}\nabla_{\mathbb{S}}Y_n^m(\theta,\varphi).
	\end{align}
	where
	\begin{align}\label{eq:r hat}
		\hat{r}=\left(\begin{array}{c}
			\sin \theta \cos \varphi \\
			\sin \theta \sin \varphi \\
			\cos \theta
		\end{array}\right), 
		\hat{\theta}=\left(\begin{array}{c}
			\cos \theta \cos \varphi \\
			\cos \theta \sin \varphi \\
			-\sin \theta
		\end{array}\right), 
		\hat{\varphi}=\left(\begin{array}{c}
			-\sin \varphi \\
			\cos \varphi \\
			0
		\end{array}\right).
	\end{align}

	Following a similar argument for proving Lemma \ref{lem:3.6}, combing \eqref{eq:S in D} and \eqref{eq:varphib}, one can obtain the  asymptotic analysis for $u |_D$ as follows 
	\begin{align}\label{eq:u4.5}
		u |_D=\sum_{m=-n}^{n} \frac{- f_{n,m} n^{3}  }{(2n + 1)(\lambda+2\mu)^{\frac{n}{2}}\delta\tau^{2}}\left(\mathcal{O}\left(1\right)+\mathcal{O}\left(\delta\tau^2 k^2\right)\right) j_n(k |{\bf x}|) Y_{n}^{m}.
	\end{align}
	Here, \( \delta \), defined in \eqref{eq:delta tau}, denotes the density ratio between bubble \( D \) and the background elastic medium \( \mathbb{R}^3 \setminus \overline{D} \). Additionally, \(\tau\) represents the wave speed contrast between the air bubble \((D;\kappa,\rho_b)\) and the elastic medium \((\mathbb{R}^3 \setminus \overline{D}; \lambda,\mu,\rho_e)\). Consequently, by combining \eqref{eq:ynm3.2}, \eqref{eq:jn expansion}, \eqref{eq:hn qiudao}, and \eqref{eq:u4.5}, and utilizing the orthogonality of $\hat{r}$ and $\nabla_{\mathbb{S}}Y_n^m(\theta,\varphi)$, we can deduce that 
	\begin{align}\label{eq:nabla u4.6}
		&\|\nabla   u\|_{L^2(\mathscr{S}_{-}^{1 - \zeta_1}(\partial D))^3}^2 \notag\\
		=&\int_{\mathscr{S}_{-}^{1 - \zeta_1}(\partial D)} \left|\sum_{m=-n}^{n} \frac{- f_{n,m} n^{3}  }{(2n + 1)(\lambda+2\mu)^{\frac{n}{2}}\delta\tau^{2}} \nabla\left(j_n(k |{\bf x}|) Y_{n}^{m} \right)\left(\mathcal{O}\left(1\right)+\mathcal{O}\left(\delta\tau^2 k^2\right)\right)\right|^2    \rmd \mathrm{x}  \notag\\
		=&\int_{\mathscr{S}_{-}^{1 - \zeta_1}(\partial D)} \left|\sum_{m=-n}^{n} \frac{- f_{n,m} n^{3}  }{(2n + 1)(\lambda+2\mu)^{\frac{n}{2}}\delta\tau^{2}} 
		\left(j_n^{\prime}(kr) k Y_n^m(\theta, \varphi)\hat{r}+ \frac{j_n(kr)}{r}\nabla_{\mathbb{S}}Y_n^m(\theta,\varphi)\right)\right. \notag\\
		&\quad\times\left.\left(\mathcal{O}\left(1\right)+\mathcal{O}\left(\delta\tau^2 k^2\right)\right)\right|^2 r^2 \sin \theta    \mathrm{d}r\mathrm{d}\theta\mathrm{d}\varphi \notag\\
		=&\int_0^{2 \pi} \int_0^\pi \int_{\zeta_1}^1 \sum_{m=-n}^{n} \frac{ |f_{n,m}|^2 n^{6} \sin \theta  }{(2n + 1)^2(\lambda+2\mu)^{n}\delta^2\tau^{4}} 
		\left( \left(j_n^{\prime}(kr)\right)^2 k^2 r^2 +n(n+1) j_n^2(kr)\right)
		\notag\\
		&\quad\times\left(\mathcal{O}\left(1\right)+\mathcal{O}\left(\delta\tau^2 k^2\right)\right) \mathrm{d}r\mathrm{d}\theta\mathrm{d}\varphi\notag\\
		=& \sum_{m=-n}^{n} \frac{ 4\pi |f_{n,m}|^2 n^{6}   }{(2n + 1)^2(\lambda+2\mu)^{n}\delta^2\tau^{4}} 
		\left[\frac{n k^{2n}}{(2n+1)!!(2n-1)!!} \left(1+\mathcal{O}(k^{2})\right)	\right]
		\int_{\zeta_1}^1  r^{2n} \mathrm{d}r \notag\\
		&\quad\times\left(\mathcal{O}\left(1\right)+\mathcal{O}\left(\delta\tau^2 k^2\right)\right)\notag\\
		=& \sum_{m=-n}^{n} \frac{4\pi |f_{n,m}|^2 n^{7}(1 - \zeta_{1}^{2n + 1}) k^{2n}}{(2n+1)^{2}[(2n+1)!!]^{2}(\lambda + 2\mu)^{n}\delta^{2}\tau^{4}}(\mathcal{O}(1)+\mathcal{O}(k^{2})).
	\end{align}

	Given the expression for $\mathbf{u}^i$ in \eqref{eq:ui def}, we utilize the orthogonality property of the function $\mathcal{I}_{n - 1}^{m}(\theta,\varphi)$ to analyze. Recalling the definition of $k_p$ from \eqref{eq:ks kp defination} and employing the series expansion of $j_n(k_p r)$ presented in \eqref{eq:jn expansion}, we derive the asymptotic behavior of $k$ through direct computation. The result follows immediately from the application of these foundational relationships.
	\begin{align}\label{eq:ui L2 norm}
		\left\lVert  \mathbf{u}^{i} \right\rVert_{L^{2}(D)^3}^{2}&=\int_{D} \sum_{m = - n}^{n} \left|f_{n,m}\right|^{2} j^2_{n}(k_{p}r)\left|\mathcal{I}_{n - 1}^{m}(\theta,\varphi)\right|^{2} \mathrm{d}\mathrm{x} \notag\\
		&=\int_{0}^{2\pi}\int_{0}^{\pi}\int_{0}^{1}\sum_{m = - n}^{n}\frac{\left|f_{n,m}\right|^{2}n(n + 1)(k\tau)^{2n}r^{2n+2}\sin\theta}{[(2n+1)!!]^{2}(\lambda+2\mu)^{n}}\left(1 - \mathcal{O}\left(k^{2}\tau^{2}\right)\right)\mathrm{d}r\mathrm{d}\theta\mathrm{d}\varphi \notag \\
		&=\sum_{m = - n}^{n}\frac{4 \pi\left|f_{n,m}\right|^{2}n(n + 1)(k\tau)^{2n}}{[(2n+1)!!]^{2}(\lambda+2\mu)^{n}} \left(1 - \mathcal{O}\left(k^{2}\tau^{2}\right)\right) \int_{0}^{1}r^{2n+2}\mathrm{d}r \notag \\
		&=\sum_{m = - n}^{n}\frac{4 \pi \left|f_{n,m}\right|^{2}n(n + 1)(k\tau)^{2n}}{(2n + 3)[(2n+1)!!]^{2}(\lambda+2\mu)^{n}}\left(1 - \mathcal{O}\left(k^{2}\tau^{2}\right)\right).
	\end{align}
	Combining the fact that $\zeta_1 \in (0, 1)$ with the inequality $1 - \zeta_1^{2n+1} \geqslant 1-\zeta_1$, one directly observes that the inequality  
	\begin{align}\label{eq:4.9}
		\frac{\left(2 + \frac{3}{n}\right)\left(1 - \zeta_1^{2n+1}\right)}{\left(2 + \frac{1}{n}\right)^2\left(1 + \frac{1}{n}\right)} \geqslant \frac{1-\zeta_1}{9}.
	\end{align}
	By applying \eqref{eq:nabla u4.6}, \eqref{eq:ui L2 norm}, and \eqref{eq:4.9} in sequence, along with the condition $\delta \ll 1$, we derive the following estimate:  
	\begin{align} \label{eq:48}
		\frac{\|\nabla   u\|_{L^2(\mathscr{S}_{-}^{1 - \zeta_1}(\partial D))^3}^2 }{\left\lVert  \mathbf{u}^{i} \right\rVert_{L^{2}(D)^3}^{2}}
		&=\frac{\sum_{m=-n}^{n} \frac{4\pi |f_{n,m}|^2 n^{7}(1 - \zeta_{1}^{2n + 1}) k^{2n}}{(2n+1)^{2}[(2n+1)!!]^{2}(\lambda + 2\mu)^{n}\delta^{2}\tau^{4}}(\mathcal{O}(1)+\mathcal{O}(k^{2}))}{\sum_{m = - n}^{n}\frac{4 \pi \left|f_{n,m}\right|^{2}n(n + 1)(k\tau)^{2n}}{(2n + 3)[(2n+1)!!]^{2}(\lambda+2\mu)^{n}}\left(1 - \mathcal{O}\left(k^{2}\tau^{2}\right)\right)} \notag\\
		&=\frac{\frac{(2n + 3)n^6(1 - \zeta_1^{2n+1})}{(2n+1)^2(n + 1)\tau^{2n + 4}\delta^2}\left(\mathcal{O}(1)+\mathcal{O}(k^{2})\right)}{\left(1 - \mathcal{O}\left(k^{2}\tau^{2}\right)\right)}     \notag\\
		&=\frac{n^4}{\tau^{2n+4}\delta^2}\cdot\frac{\left(2+\frac{3}{n}\right)(1 - \zeta_1^{2n+1})}{\left(2+\frac{1}{n}\right)^2\left(1+\frac{1}{n}\right)} \left(\mathcal{O}(1)+\mathcal{O}(k^{2})\right) 
		\left(1 + \mathcal{O}\left(k^{2}\tau^{2}\right)\right) \notag\\
		&\geqslant \frac{n^4\left(1-\zeta_1\right)}{9 \tau^{2n+4}\delta^2} \left(\mathcal{O}(1)+\mathcal{O}(k^{2})\right) > M.
	\end{align}

In view of \eqref{eq:48}, we deduce that the internal total field $u$ exhibits  {surface resonance}. Subsequently, we demonstrate that the inequality \eqref{eq:43} holds. By combining the scattering problem \eqref{eq:system2} with the solution expression in \eqref{eq:u express}, and by selecting the appropriate incident wave defined in \eqref{eq:ui def}, the scattering problem can be reformulated as solving a linear system represented by \eqref{eq:2.20}. Using \eqref{eq:jn expansion}, \eqref{eq:hn expansion}, \eqref{eq:S nu,K nu}, and \eqref{eq:varphi_e}, we derive the scattered field.
	\begin{align}\label{eq:ushn}
		\mathbf{u}^s = \frac{-\rmi n^2 k^{2n + 2} \tau^{n+2}}{[(2n + 1)!!]^2(\lambda + 2\mu)^{n+2}} \left(\mathcal{O}(1)+ \mathcal{O}(\delta\tau^{2}k^{2})\right)\sum_{m = -n}^{n} f_{n,m} h_n(k_p|\mathbf{x}|)Y_n^m(\theta,\varphi) \boldsymbol{\nu}.
	\end{align}
	After direct calculation, it is readily known that
	\begin{align}\label{eq:4.12}
		\nabla \bmf u^s= \frac{-\rmi n^2 k^{2n + 2} \tau^{n+2}}{[(2n + 1)!!]^2(\lambda + 2\mu)^{n+2}} \left(\mathcal{O}(1)+ \mathcal{O}(\delta\tau^{2}k^{2})\right)\sum_{m = -n}^{n} f_{n,m} \nabla \left(h_n(k_p|\mathbf{x}|)Y_n^m(\theta,\varphi) \boldsymbol{\nu}\right)
	\end{align}
	In order to derive the  asymptotic analysis for $\nabla \mathbf{u}^s|_{\mathscr{S}_{+}^{\zeta_2 - 1}(\partial D)}$ with respect to $k$. It turns out that we need calculate $\nabla\left(h_n(k
	r) Y_n^m(\theta, \varphi) \hat{r}\right)$, where $Y_n^m(\theta, \varphi)$ is defined by \eqref{eq:ynm and cnm def}. In the following, let us derive the explicit expression for $\nabla\left(h_n(k r) Y_n^m(\theta, \varphi) \hat{r}\right)$.
	\begin{align}\label{eq:nablahn}
		&\nabla\left(h_n(k_p r) Y_n^m(\theta,\varphi)\hat{r}\right) 
		= h_n^{\prime}(k_p r)k_p Y_n^m(\theta,\varphi)\hat{r}\otimes\hat{r} + \frac{h_n(k_p r)}{r} \frac{\partial Y_n^m(\theta,\varphi)}{\partial \theta}\hat{r}\otimes\hat{\theta} \notag\\
		&\quad+\frac{h_n(k_p r)}{r\sin\theta} \frac{\partial Y_n^m(\theta,\varphi)}{\partial \varphi}\hat{r}\otimes\hat{\varphi}+\frac{h_n(k_p r)Y_n^m(\theta,\varphi)}{r}\hat{\theta}\otimes\hat{\theta}
		+\frac{h_n(k_p r)Y_n^m(\theta,\varphi)}{r}\hat{\varphi}\otimes\hat{\varphi},
	\end{align}
	where $\hat{r}$, $\hat{\theta}$, and $\hat{\varphi}$ are defined correspondingly in  \eqref{eq:r hat}. Combining with \eqref{eq:ks kp defination}, \eqref{eq:hn expansion}, \eqref{eq:4.12}, and \eqref{eq:nablahn}, when $n$ is sufficiently large, we can derive that
	\begin{align}\label{eq:4.14}
		&\|\nabla   \mathbf{u}^s \|_{L^2(\mathscr{S}_{+}^{\zeta_2 - 1}(\partial D))^3}^2 \notag\\
		=&\int_{\mathscr{S}_{+}^{\zeta_2 - 1}(\partial D)} \left|\frac{-\rmi n^2 k^{2n + 2} \tau^{n+2}}{[(2n + 1)!!]^2(\lambda + 2\mu)^{n+2}} \left(\mathcal{O}(1)+ \mathcal{O}(\delta\tau^{2}k^{2})\right)\sum_{m = -n}^{n} f_{n,m} \nabla \left(h_n(k_s|\mathbf{x}|)Y_n^m(\theta,\varphi) \boldsymbol{\nu}\right)\right|^2    \rmd \mathrm{x}  \notag\\
		=& \int_{0}^{2\pi}\int_{0}^{\pi}\int_{1}^{\zeta_{2}} \sum_{m = -n}^{n} 
		\frac{|f_{n,m}|^2 n^4 k^{4n + 4} \tau^{2n+4} \sin\theta}{[(2n + 1)!!]^4 (\lambda + 2\mu)^{2n+4}} 
		\left\{[k_p r h_n^{\prime}(k_p r)]^2 + (n^2 + n + 2) h_n^2(k_p r)\right\}    \notag\\
		&\quad\times\left(\mathcal{O}\left(1\right)+\mathcal{O}\left(\delta\tau^2 k^2\right)\right) \mathrm{d}r\mathrm{d}\theta\mathrm{d}\varphi\notag\\
		=&  \sum_{m = -n}^{n} 
		\frac{4 \pi|f_{n,m}|^2 n^4(2n^2+3n+3) k^{4n + 4} \tau^{2n+4} }{[(2n + 1)!!]^4 (\lambda + 2\mu)^{2n+4}} \left\{\frac{[(2n - 1)!!]^2(\lambda + 2\mu)^{n + 1}}{k^{2n+2}\tau^{2n+2}}(1 - \mathcal{O}(k^{2}\tau^{2}))\right\}  \notag\\
		&\quad\times \int_{1}^{\zeta_{2}}\frac{1}{r^{2n+2}} \rmd r \left(\mathcal{O}\left(1\right)+\mathcal{O}\left(\delta\tau^2 k^2\right)\right)\notag\\
		=&\sum_{m = -n}^{n}  \frac{4\pi\left|f_{n,m}\right|^2  n^4(2n^2+3n+3) k^{2n+2} \tau^2 (\zeta_2^{2n + 1}-1)}{[(2n + 1)!!]^2(\lambda + 2\mu)^{n+3}(2n + 1)^3\zeta_2^{2n+1}} \left(\mathcal{O}(1)-\mathcal{O}(k^2\tau^2)\right).
	\end{align}
	Combining the fact that $\zeta_2 \in (1, R)$ with the inequality $1 - \frac{1}{\zeta_2^{2n+1}} \geqslant 1-\frac{1}{\zeta_2}$, we directly deduce that the inequality  
	\begin{align}\label{eq:416}
		\frac{\left(10+\frac{7}{n}+\frac{3}{n^2}\right) \left(2+\frac{3}{n}\right)\left(1 - \frac{1}{\zeta_2^{2n+1}}\right)}{\left(1+\frac{1}{n}\right)\left(2+\frac{1}{n}\right)^3} \geqslant \frac{10\left(\zeta_2-1\right)}{27\zeta_2}.
	\end{align}
	From \eqref{eq:ui L2 norm}, \eqref{eq:4.14}, and \eqref{eq:416}, given a fixed incident wave $\mathbf{u}^i$ and $\tau$ satisfying \eqref{eq:assume}, it follows that for sufficiently large $n$ in the definition of $\mathbf{u}^i$ (as defined in \eqref{eq:ui def}),  
	\begin{align}\label{eq:414}
		\frac{\|\nabla \bmf  u^s\|_{L^2(\mathscr{S}_{+}^{\zeta_2 - 1}(\partial D))^3}^2}{\| \bmf  u^i\|_{L^2(D)^3}^2}
		=&\frac{\sum_{m = -n}^{n}\frac{4\pi\left|f_{n,m}\right|^2 n^4(2n^2+3n+3) k^{2n+2} \tau^2 (\zeta_2^{2n + 1}-1)}{[(2n + 1)!!]^2(\lambda + 2\mu)^{n+3}(2n + 1)^3\zeta_2^{2n+1}} \left(\mathcal{O}(1)-\mathcal{O}(k^2\tau^2)\right)}{\sum_{m = - n}^{n}\frac{4 \pi \left|f_{n,m}\right|^{2}n(n + 1)(k\tau)^{2n}}{(2n + 3)[(2n+1)!!]^{2}(\lambda+2\mu)^{n}}\left(1 - \mathcal{O}\left(k^{2}\tau^{2}\right)\right)} \notag\\
		=&\frac{\frac{n^3(2n^2+3n+3)(2n + 3)(\zeta_2^{2n+1}-1)k^2}{(n + 1)(2n+1)^3(\lambda + 2\mu)^3\zeta_2^{2n+1}\tau^{2n - 2}} \left(\mathcal{O}(1)-\mathcal{O}(k^2\tau^2)\right)}{\left(1 - \mathcal{O}\left(k^{2}\tau^{2}\right)\right)} \notag\\
		=&\frac{n^2 k^2}{(\lambda + 2\mu)^3\tau^{2n - 2}} \frac{\left(10+\frac{7}{n}+\frac{3}{n^2}\right) \left(2+\frac{3}{n}\right)\left(1 - \frac{1}{\zeta_2^{2n+1}}\right)}{\left(1+\frac{1}{n}\right)\left(2+\frac{1}{n}\right)^3}  \left(\mathcal{O}(1)-\mathcal{O}(k^2\tau^2)\right)\notag\\
		& \times \left(1 + \mathcal{O}\left(k^{2}\tau^{2}\right)\right)\notag\\
		\geqslant&\frac{10\left(\zeta_2-1\right)n^2 k^2}{27\zeta_2(\lambda + 2\mu)^3\tau^{2n - 2}} \left(\mathcal{O}(1)-\mathcal{O}(k^2\tau^2)\right) > M. 
	\end{align}

 Therefore, we deduce that the exterior scattered field $\mathbf{u}^s$ exhibits surface resonance in $\mathscr{S}_{+}^{\zeta_2 - 1}(\partial D)$, as defined in~\eqref{eq:Sdef}, via~\eqref{eq:414} when the index $n$ of the incident wave satisfies $n > n_4$, where $n_4$ is defined in \eqref{eq:n3n4 def}. Similarly, the interior total field $u$ exhibits surface resonance in $\mathscr{S}_{-}^{1 - \zeta_1}(\partial D)$, as defined in~\eqref{eq:Sdef}, associated with Definition~\ref{def:surface resonant}, when $n > n_3$, where $n_3$ is defined in \eqref{eq:n3n4 def}.

	The proof is complete.
\end{proof}

\begin{rem}\label{rem:rem32}

Indeed, when the index $ n $ of the incident wave $ \mathbf{u}^i $, defined in \eqref{eq:ui def}, satisfies \eqref{eq:n1 n2 def}, the interior total field $ u $ and the exterior scattering field $ \mathbf{u}^s $ exhibit boundary localization in $ \mathscr{S}_{-}^{1 - \zeta_1}(\partial D) $ and $ \mathscr{S}_{+}^{\zeta_2 - 1}(\partial D) $, respectively, as defined in \eqref{eq:Sdef}, where $ \zeta_1 \in (0, 1) $ and $ \zeta_2 \in (1, R) $, according to Theorem \ref{thm:th3.1}. Similarly, when the index $ n $ of the incident wave $ \mathbf{u}^i $, satisfies \eqref{eq:N2 def}, the interior total field $ u $ and the exterior scattering field $ \mathbf{u}^s $ exhibit surface resonance in $ \mathscr{S}_{-}^{1 - \zeta_1}(\partial D) $ and $ \mathscr{S}_{+}^{\zeta_2 - 1}(\partial D) $, respectively. It is noted that in Theorem \ref{thm:th3.1}, when the index $n \geqslant \max\{n_1, n_2\}$, where $n_1$ and $n_2$ are defined in \eqref{eq:n1 n2 def}, the boundary localization of $u|_D$ and $\mathbf{u}^s|_{\mathbb{R}^3 \setminus \overline{D}}$ is achieved. However, since $n_3$ and $n_4$ in \eqref{eq:n3n4 def} are independent of $n_1$ and $n_2$, when the index $n$ satisfies only \eqref{eq:N2 def}, boundary localization of $u|_D$ and $\mathbf{u}^s|_{\mathbb{R}^3 \setminus \overline{D}}$ cannot be guaranteed.

Based on the analysis in Theorems \ref{thm:th3.1} and \ref{thm:nabla u in thm}, if the index $n$ of the incident wave $\mathbf{u}^i$ satisfies
\begin{align}\label{eq:355}
    n > \max\{N_1, N_2\},
\end{align}
then the interior total field $u$ and the exterior scattering field $\mathbf{u}^s$ exhibit quasi-Minnaert resonance in $\mathscr{S}_{-}^{1 - \zeta_1}(\partial D)$ and $\mathscr{S}_{+}^{\zeta_2 - 1}(\partial D)$, respectively. Here, $N_1$ and $N_2$ are defined in \eqref{eq:n1 n2 def} and \eqref{eq:N2 def}, respectively. The parameter $N_1$ is associated with the level of boundary localization $\eta$ and the parameters $\zeta_1$ and $\zeta_2$, while $N_2$ depends on the wave speed ratio parameter $\tau$, the density ratio $\delta$, the parameter $M$ of high oscillation of the wave, and the parameters $\zeta_1$ and $\zeta_2$.

\end{rem}

In the next remark, we separately consider the conditions under which the exterior scattered field $\mathbf{u}^s$ exhibits quasi-Minnaert resonance, as this phenomenon significantly influences the occurrence of stress concentration for  the exterior total elastic field $\mathbf{u}|_{\mathbb{R}^3 \setminus \overline{D}}$ in Section \ref{sec:summary of major findings}. A detailed discussion of the relationship between the quasi-Minnaert resonance of $\mathbf{u}^s$ and the stress concentration in the exterior elastic total field $\mathbf{u}|_{\mathbb{R}^3 \setminus \overline{D}}$ will be presented in Remark \ref{rem:41rem}.

\begin{rem}\label{rem:32rem}

From Theorems \ref{thm:th3.1} and \ref{thm:nabla u in thm}, we conclude that when the index $n$ of the incident wave $\mathbf{u}^i$, defined in \eqref{eq:ui def}, satisfies
\begin{align}\label{eq:356}
    n > \max\{n_2, n_4\},
\end{align}
where $n_2$ and $n_4$ are defined in \eqref{eq:n1 n2 def} and \eqref{eq:n3n4 def}, respectively, the exterior scattering field $\mathbf{u}^s$ exhibits quasi-Minnaert resonance, encompassing both boundary localization and surface resonance. Similarly, the interior total field $u$ exhibits quasi-Minnaert resonance  when the index $n$ of the incident wave $\mathbf{u}^i$ satisfies
\begin{align}\label{eq:357}
    n > \max\{n_1, n_3\},
\end{align}
where $n_1$ and $n_3$ are defined in \eqref{eq:n1 n2 def} and \eqref{eq:n3n4 def}, respectively. This analysis demonstrates that it is very flexible to independently achieve quasi-Minnaert resonance in the interior total field $u$ and the exterior scattering field $\mathbf{u}^s$ through 
appropriately selecting the incident wave $\mathbf{u}^i$ with different indices $n$.  

\end{rem}

 \begin{rem}\label{rem:33rem}

 In this remark, we discuss the fact that imposing a prior condition on the high-oscillation parameter $M$ for the wave allows surface resonance to induce boundary localization.
For any given level $\eta$ of boundary localization (as described in Theorem~\ref{thm:th3.1}), if the index $n$ of $\mathbf{u}^i$ and the high-oscillation parameter $M$ satisfy
\begin{align}\label{eq:M0 def}
n > n_3 \quad \text{and} \quad M > M_0(k, \zeta_1, \eta, \tau, \delta) = \frac{\left( \log_{\zeta_1} \left( \frac{\eta}{\zeta_1^3} \right) \right)^2 (1 - \zeta_1)^{1/2} \tau^{-\frac{1}{2} \log_{\zeta_1} (\zeta_1 \eta)}}{12 \delta},
\end{align}
then Theorem~\ref{thm:nabla u in thm} implies that the surface resonance of $u|_D$ induces boundary localization of $u|_D$, provided that the index $n$ associated with the incident wave $\mathbf{u}^i$ satisfies $n > n_3$.
Similarly, if
\begin{align}\label{eq:M1 def}
n > n_4 \quad \text{and} \quad M > M_1(k, \zeta_2, \eta, \tau, \lambda, \mu) = \frac{k \log_{\zeta_2} \left( \frac{\zeta_2}{\eta} \right) \sqrt{\frac{10 (\zeta_2 - 1)}{3 \zeta_2}} \tau^{\frac{1}{2} \log_{\zeta_2} (\zeta_2 \eta)}}{6 (\lambda + 2\mu)^{3/2}},
\end{align}
then both the boundary localization and the surface resonance of $\mathbf{u}^s|_{\mathbb{R}^3 \setminus \overline{D}}$ are induced.

 \end{rem}

\section{Stress concentration of exterior total field}\label{sec:summary of major findings}

In this section, we analyze the stress concentration in the exterior total field $\mathbf{u}|_{\mathbb{R}^3 \setminus \overline{D}}$, which arises due to the quasi-Minnaert resonance. The following theorem characterizes this phenomenon in the domain $\mathscr{S}_{+}^{\zeta_2 - 1}(\partial D)$, defined in \eqref{eq:Sdef} as the stress concentration region, where $\zeta_2 \in (1, R)$. Under the physical setup, the wave speed ratio satisfies $\tau < 1$. For technical reasons, we assume in this section that $\zeta_2 \tau < 1$. This condition is readily satisfied when $\zeta_2 - 1 \ll 1$, implying that the stress concentration region lies very close to the bubble's boundary.

Meanwhile, in Remarks \ref{rem:41rem} and  \ref{rem:4.2}, we provide detailed discussions for the relationship between stress concentration and quasi-Minnaert resonance. In Remark \ref{rem:43rem}, we summarize the phenomena exhibited by the interior total field $\mathbf{u}|_D$, the exterior scattering field $\mathbf{u}^s|_{\mathbb{R}^3 \setminus \overline{D}}$, and the exterior total field $\mathbf{u}|_{\mathbb{R}^3 \setminus \overline{D}}$, which depend on the index $n$ of the incident wave $\mathbf{u}^i$ or the oscillation parameter $M$ satisfying different conditions.

\begin{thm}\label{thm:Eu definition in thm}

Consider the bubble-elastic scattering problem \eqref{eq:system2}, where $(D; \kappa, \rho_b)$ represents an air bubble embedded in a homogeneous elastic medium $(\mathbb{R}^3 \setminus \overline{D}; \lambda, \mu, \rho_e)$. Under the assumptions \eqref{eq:assume} and \eqref{eq:ks kp defination}, for the stress concentration region $\mathscr{S}_{+}^{\zeta_2 - 1}(\partial D)$ with $\zeta_2 \in (1, R)$, suppose that $\zeta_2 \tau < 1$, where $\tau$ is the wave speed ratio defined in \eqref{eq:delta tau} with $\tau < 1$. Recall that $n_4$ is defined in \eqref{eq:n3n4 def}, which depends on the high-oscillation parameter $M \gg 1$ and the stress $E(\mathbf{u})$ as defined in \eqref{eq:Eu}. If the high-oscillation parameter $M$ satisfies 
\begin{align}\label{eq:M1 in 4.1}
    M > M_1,
\end{align}
where $M_1$ is defined by \eqref{eq:M1 def}, and the index $n$ associated with the incident wave $\mathbf{u}^i$ fulfills 
\begin{align}\label{eq:n4 in 4.1}
    n > n_4,
\end{align}
then 
\begin{equation}\label{eq:Eus th4.2}
    \frac{E(\mathbf{u})}{\|\mathbf{u}^i\|_{L^2(D)^3}^2} \geqslant \frac{n^{2} (\zeta_2 - 1) k^{2}}{27 \zeta_2 (\lambda + 2\mu)^{2} \tau^{2n - 2}} > M \gg 1.
\end{equation}
Namely, the exterior total wave $\mathbf{u}|_{\mathbb{R}^3 \setminus \overline{D}}$ exhibits stress concentration in the region $\mathscr{S}_{+}^{\zeta_2 - 1}(\partial D)$ in the sense of Definition~\ref{def:stress concentration}. Here $M$ also characterize the stress concentration level of $E(\mathbf u)$.

\end{thm}

	\begin{proof}

	Since the incident wave $\mathbf u^i$ is given by \eqref{eq:ui def}, by Lemma \ref{lem:3.6} we have the formula \eqref{eq:scattered field} for  $\mathbf u^s|_{\mathbb R^3 \setminus \overline D}$  with respect to $k$, consider the \eqref{eq:u express}, we rewrite \eqref{eq:Eu} as: 
	\begin{align}\label{eq:Eu=def}
		E(\mathbf{u})=E(\mathbf{u}^s)+E(\mathbf{u}^i)+ Rest,
	\end{align}
where 
\begin{align}\label{eq:rest=}
	Rest&= \int_{\mathscr{S}_{+}^{\zeta_2 - 1}(\partial D)} \left[ \lambda (\nabla \cdot \mathbf{u}^s) \mathbf{I} + \mu \left( \nabla \mathbf{u}^s + (\nabla \mathbf{u}^s)^{\top} \right)\right]  : \nabla \overline{\mathbf{u}^i} \, \mathrm{d} \mathbf{x} \notag\\ &+\int_{\mathscr{S}_{+}^{\zeta_2 - 1}(\partial D)} \left[ \lambda (\nabla \cdot \mathbf{u}^i) \mathbf{I} + \mu \left( \nabla \mathbf{u}^i + (\nabla \mathbf{u}^i)^{\top} \right)\right]  : \nabla \overline{\mathbf{u}^s} \, \mathrm{d} \mathbf{x}.   
\end{align}
	First, by the definition of \(E(\mathbf{u}^s)\), we have
	\begin{align}\label{eq:Eus419}
		E(  \mathbf{u}^s)
		&=\lambda\int_{\mathscr{S}_{+}^{\zeta_2 - 1}(\partial D)} \left(\nabla \cdot \mathbf{u}^s\right) \cdot 
		{\rm tr}( \nabla  \overline{\mathbf{u}^s})~{\rm d} \mathbf{x}
		+ \mu  \|\nabla  \mathbf{u}^s\|_{L^2(\mathscr{S}_{+}^{\zeta_2 - 1}(\partial D))^3}^2 \notag\\
		&\quad+\int_{\mathscr{S}_{+}^{\zeta_2 - 1}(\partial D)} \mu \cdot {\rm tr}(\nabla  \mathbf{u}^s \nabla {\overline{\mathbf{u}^s}})~\rm d \mathbf{x}. 
	\end{align}
	Hence, due to \eqref{eq:ushn}, one can obtain that 
	\begin{align}\label{eq:divus}
		\nabla \cdot \bmf u^s= \frac{-\rmi n^2 k^{2n + 2} \tau^{n+2}}{[(2n + 1)!!]^2(\lambda + 2\mu)^{n+2}} \left(\mathcal{O}(1)+ \mathcal{O}(\delta\tau^{2}k^{2})\right)\sum_{m = -n}^{n} f_{n,m} \nabla \cdot \left(h_n(k_s|\mathbf{x}|)Y_n^m(\theta,\varphi) \boldsymbol{\nu}\right).
	\end{align}
     To derive the asymptotic analysis for \(\nabla \cdot \mathbf{u}^s|_{\mathscr{S}_{+}^{\zeta_2 - 1}(\partial D)}\) with respect to \(k\), we need to compute \(\nabla \cdot \left(h_n(k r) Y_n^m(\theta, \varphi) \hat{r}\right)\), where \(Y_n^m(\theta, \varphi)\) is defined in \eqref{eq:ynm and cnm def}. Below, we derive the explicit expression for \(\nabla \cdot \left(h_n(k r) Y_n^m(\theta, \varphi) \hat{r}\right)\).
	\begin{align}\label{eq:divhn}
		\nabla \cdot \left(h_n(k_p r) Y_n^m(\theta,\varphi)\hat{r}\right) 
		=\frac{1}{r^2}\frac{\partial\left(r^2 h_n(k_p r)\right)}{\partial r} Y_n^m\left(\theta,\varphi\right)=\frac{2 h_n(k_p r)+k_p r h_n^{\prime}(k_p r)}{r}Y_n^m\left(\theta,\varphi\right).
	\end{align}
	 Furthermore, combing the \eqref{eq:4.12} and \eqref{eq:nablahn}, we have:
	 \begin{align}\label{eq:tr422}
	   \mathrm{tr}(\nabla \overline{\mathbf{u}^s})
	   =\sum_{m = -n}^{n}\frac{\rmi \overline{f_{n,m}} n^{2} k^{2n + 2} \tau^{n+2}}{[(2n + 1)!!]^2(\lambda + 2\mu)^{n+2}} \frac{2\overline{h_n(k_pr)}+k_pr\overline{h_n^{\prime}(k_pr)}}{r}\overline{Y_n^m(\theta,\varphi)}
	   (\mathcal{O}(1)+\mathcal{O}(\delta \tau^{2}k^{2})).
	 \end{align} 
	 According to \eqref{eq:divus}, \eqref{eq:divhn}, and \eqref{eq:tr422} it yields that
	\begin{align}\label{eq:423}
	     &\lambda\int_{\mathscr{S}_{+}^{\zeta_2 - 1}(\partial D)}(\nabla\cdot \mathbf{u}^s)\cdot\mathrm{tr}(\nabla\overline{\mathbf{u}^s})~\rmd \mathbf x \\
		=&\int_{\mathscr{S}_{+}^{\zeta_2 - 1}(\partial D)}\sum_{m=-n}^{n}\frac{\lambda|f_{n,m}|^{2}n^{4}k^{4n + 4}\tau^{2n+4}}{[(2n + 1)!!]^4(\lambda + 2\mu)^{2n+4}}
		\frac{|2h_n(k_pr)+k_prh_n^{\prime}(k_pr)|^{2}}{r^{2}}|Y_n^m(\theta,\varphi)|^{2} \notag\\
	    &\times(\mathcal{O}(1)+\mathcal{O}(\delta\tau^{2}k^{2}))^{2}\rmd \mathbf x \notag\\
	    =&\int_{0}^{2\pi}\int_{0}^{\pi}\int_{1}^{\zeta_2}\sum_{m = -n}^{n}\frac{\lambda|f_{n,m}|^{2}n^{4}k^{4n + 4}\tau^{2n+4}\sin\theta}{[(2n + 1)!!]^4(\lambda + 2\mu)^{2n+4}}
	    \left|k_p rh_n^{\prime}(k_pr)+2h_n(k_pr)\right|^{2} \mathrm{d}r\mathrm{d}\theta\mathrm{d}\varphi \notag\\
	    &\times(\mathcal{O}(1)+\mathcal{O}(\delta\tau^{2}k^{2}))\notag\\
	    =&\sum_{m = -n}^{n}\frac{4\pi\lambda|f_{n,m}|^{2}n^{4}k^{4n + 4}\tau^{2n+4}}{[(2n + 1)!!]^4(\lambda + 2\mu)^{2n+4}}
	    \left[\frac{[(2n - 1)!!]^2(n+3)^2(\lambda + 2\mu)^{n+1}}{k^{2n+2}\tau^{2n+2}}(1 - \mathcal{O}(k^{2}\tau^{2}))\int_{r_1}^{r_2}\frac{1}{r^{2n+2}}\mathrm{d}r\right]
	    \notag\\
	    &\times(\mathcal{O}(1)+\mathcal{O}(\delta\tau^{2}k^{2})) \notag\\
	    =&\sum_{m = -n}^{n}\frac{4\pi\lambda|f_{n,m}|^{2}n^{4}k^{2n + 2}\tau^{2}(n+3)^{2}(\zeta_2^{2n+1}-1)}{[(2n + 1)!!]^2(2n + 1)^{3}(\lambda + 2\mu)^{n+3}\zeta_2^{2n+1}}(\mathcal{O}(1)-\mathcal{O}(k^{2}\tau^{2})). \notag
	\end{align}
	Utilizing the similar argument for \eqref{eq:tr422}, we have
	\begin{align}\label{eq:424}
		&\rm{tr}(\nabla \mathbf{u}^s\nabla\overline{\mathbf{u}^s})\\
		=&\sum_{m=-n}^{n}\frac{|f_{n,m}|^{2}n^{4}k^{4n + 4}\tau^{2n+4}}{[(2n + 1)!!]^4(\lambda + 2\mu)^{2n+4}}\frac{|2h_n(k_pr)+k_prh_n^{\prime}(k_pr)|^{2}}{r^{2}}|Y_n^m(\theta,\varphi)|^{2}
		(\mathcal{O}(1)+\mathcal{O}(\delta\tau^{2}k^{2})). \notag
	\end{align}
	Hence, by a similar argument for \eqref{eq:423}, due to \eqref{eq:424}, one can obtain that  
	\begin{align}\label{eq:mu425}
	\int_{\mathscr{S}_{+}^{\zeta_2 - 1}(\partial D)} \mu \cdot {\rm tr}(\nabla {\bmf u}^s \nabla \overline {\mathbf{u^s}})~\rmd  \mathbf{x} =\sum_{m = -n}^{n}\frac{4\pi\mu|f_{n,m}|^{2} n^{4}k^{2n + 2}\tau^{2}(n + 3)^{2}(\zeta_2^{2n+1}-1)}{[(2n + 1)!!]^2(2n + 1)^{3}(\lambda + 2\mu)^{n+3}\zeta_2^{2n+1}}(\mathcal{O}(1)-\mathcal{O}(k^{2}\tau^{2})). 
	\end{align}
  Substituting \eqref{eq:4.14}, \eqref{eq:423} and \eqref{eq:mu425} into \eqref{eq:Eus419}, we can derive that
	%
	\begin{align}\label{eq:EusL2}
		E(\bmf u^s)=\sum_{m = -n}^{n}  \frac{4\pi\left|f_{n,m}\right|^2  (\lambda+3\mu)n^6  k^{2n+2} \tau^2 (\zeta_2^{2n + 1}-1)}{[(2n + 1)!!]^2(\lambda + 2\mu)^{n+3}(2n + 1)^3\zeta_2^{2n+1}} \left(\mathcal{O}(1)-\mathcal{O}(k^2\tau^2)\right).
	\end{align}
    %
    %
    %
    Similar to \eqref{eq:Eus419}, we have the following expansion of $E(\mathbf{u}^i)$,
		\begin{align}\label{eq:Eui}
		E(  \mathbf{u}^i)
		&=\lambda\int_{\mathscr{S}_{+}^{\zeta_2 - 1}(\partial D)} \left(\nabla \cdot \mathbf{u}^i \right) \cdot 
		{\rm{tr}}( \nabla  \overline{\mathbf{u}^{i} })~{\rm d} \mathbf{x}
		+ \mu  \|\nabla  {\mathbf{u}}^i\|_{L^2(\mathscr{S}_{+}^{\zeta_2 - 1}(\partial D))^3}^2 \notag\\
		&\quad+\int_{\mathscr{S}_{+}^{\zeta_2 - 1}(\partial D)} \mu \cdot {\rm tr}(\nabla  {\mathbf{u}}^i \nabla {\overline{\mathbf{u}^i}})~\rm d \mathbf{x}. 
	\end{align}
	%
	By leveraging \eqref{eq:ui def}, the definition of $k_p$ as given in \eqref{eq:ks kp defination}, and the expansion of $j_n(z)$ as presented in \eqref{eq:jn expansion}, direct computation yields the desired result.
	\begin{align}\label{eq:429}
		\lambda\int_{\mathscr{S}_{+}^{\zeta_2 - 1}(\partial D)} \left(\nabla \cdot \mathbf{u}^i \right) \cdot 
		{\rm tr}( \nabla  \overline{\mathbf{u}^{i} })~\rm d \mathbf{x}
		&=  \sum_{m=-n}^{n} \frac{ 4\pi |f_{nm}|^2 n^4 (k \tau)^{2n} \left( \zeta_2^{2n+1} - 1 \right)}{(2n+1) \cdot \left( (2n+1)!! \right)^2 (\lambda + 2\mu)^n}\left(1+\mathcal{O}\left(k^2\tau^2\right)\right).
	\end{align}
	Building upon the definition of $\mathbf{u}^i$ given in \eqref{eq:ui def} and employing the proof technique from \eqref{eq:4.14}, we establish that the following equation holds.
	\begin{align}\label{eq:430}
		& \quad \|\nabla \mathbf{u}^i\|_{L^2(\mathscr{S}_{+}^{\zeta_2 - 1}(\partial D))^3}^2 \notag\\ 
		&= \int_{0}^{2\pi}\int_{0}^{\pi}\int_{1}^{\zeta_2} 
		\sum_{m=-n}^n |f_{n,m}|^2 |Y_n^m|^2 \Bigg[ 
		n(n+1)k_p^2 j_{n-1}^2(k_p r)  - \frac{2n^2(n+1)k_p}{r} j_{n-1}(k_p r)j_n(k_p r) \notag\\ 
		&\quad + \frac{n^3(n^2 + 3n + 4)}{(n+1)r^2} j_n^2(k_p r) 
		\Bigg]\cdot r^2 \sin \theta   \mathrm{d}r   \mathrm{d}\theta   \mathrm{d}\phi \notag\\
		&=\sum_{m=-n}^{n}    \frac{4\pi |f_{n,m}|^2 (n+1)n^2 (k \tau)^{2n} \left( \zeta_2^{2n+1} - 1 \right)}{(2n+1) \left(  (2n-1)!! \right)^2 (\lambda + 2\mu)^n}\left(1+\mathcal{O}\left(k^2\tau^2\right)\right).
	\end{align}
	We adopt a proof technique analogous to that presented in \eqref{eq:mu425}. By integrating the definition of $k_p$ given in \eqref{eq:ks kp defination} and the expansion of $j_n(z)$ presented in \eqref{eq:jn expansion}, we derive the following formula.
	\begin{align}\label{eq:431}
		&\int_{\mathscr{S}_{+}^{\zeta_2 - 1}(\partial D)} \mu \cdot {\rm tr}(\nabla  {\mathbf{u}}^i \nabla {\overline{\mathbf{u}^i}})~{\rm d} \mathbf{x}
		=\mu\int_{0}^{2\pi}\int_{0}^{\pi}\int_{1}^{\zeta_2}\sum_{m=-n}^n |f_{n,m}|^2   \notag\\
		&  \cdot\Bigg[ n^2 k_p^2 j_{n-1}^2(k_p r) - \frac{2 n^3 k_p j_{n-1}(k_p r) j_n(k_p r)}{r} + \frac{\left( 2n^4 + 2n^3 + n^2 - 2n \right) j_n^2(k_p r)}{r^2} \Bigg] \cdot r^2 \sin \theta  \mathrm{d}r   \mathrm{d}\theta   \mathrm{d}\phi \notag\\
		&= \sum_{m=-n}^n    \frac{4\pi\mu|f_{n,m}|^2 k_p^{2n}}{2n+1} \cdot \frac{n^2}{(2n-1)!!^2} \int_{1}^{\zeta_2} r^{2n} \mathrm{d}r \left(1+\mathcal{O}\left(k^2\tau^2\right)\right) \notag \\
		&= \sum_{m=-n}^n\frac{4\pi \mu|f_{n,m}|^2 n^2 (k\tau)^{2n} (\zeta_2^{2n+1} - 1)}{(2n+1)^2 (2n-1)!!^2 (\lambda+2\mu)^n}  \left(1+\mathcal{O}\left(k^2\tau^2\right)\right). 
	\end{align}
  %
  %
  %
   Finally, by substituting \eqref{eq:429}, \eqref{eq:430}, and \eqref{eq:431} into \eqref{eq:Eui}, the derivation leads to the conclusion that
\begin{align}\notag
      E(\mathbf{u}^i)= \sum_{m=-n}^{n} \frac{4\pi |f_{nm}|^2 n^2 (k\tau)^{2n} \left( \zeta_2^{2n+1} - 1 \right)[4n^3 + (\lambda + 8)n^2 + (2\mu + 5)n + (\mu + 1)]}{(2n+1)^3 \cdot \left( (2n-1)!! \right)^2 (\lambda + 2\mu)^n}     \left(1+\mathcal{O}\left(k^2\tau^2\right)\right).
\end{align}
 It yields that
 \begin{equation}\label{eq:Eui Eus}
    {\sf r}_{E(\mathbf{u}^i), E(\mathbf{u}^s)  }= \frac{E(\mathbf{u}^i)}{E(\mathbf{u}^s)}=\mathcal{O}\left( \frac{n (\tau \zeta_2)^{2n} \zeta_2}{k^2 \tau^2} \right).
 \end{equation}

We adopt a proof technique analogous to those presented in \eqref{eq:EusL2} and \eqref{eq:Eui}, as established in \eqref{eq:rest=}. By integrating the definition of $\mathbf{u}^i$ given in \eqref{eq:ui def}, the external scattered field $\mathbf{u}^s$ given in \eqref{eq:scattered field}, the definition of $k_p$ given in \eqref{eq:ks kp defination}, and the expansions of $j_n(z)$ and $h_n(z)$ presented in \eqref{eq:jn expansion} and \eqref{eq:hn expansion}, respectively, we perform an expansion with respect to $k$ to derive the following formula.
\begin{align}\label{eq:Rest}
	&Rest= \int_{\mathscr{S}_{+}^{\zeta_2 - 1}(\partial D)} \left[ \lambda (\nabla \cdot \mathbf{u}^s) \mathbf{I} + \mu \left( \nabla \mathbf{u}^s + (\nabla \mathbf{u}^s)^{\top} \right)\right]  : \nabla \overline{\mathbf{u}^i} \, \mathrm{d} \mathbf{x} \notag\\ 
	&\quad \quad \quad +\int_{\mathscr{S}_{+}^{\zeta_2 - 1}(\partial D)} \left[ \lambda (\nabla \cdot \mathbf{u}^i) \mathbf{I} + \mu \left( \nabla \mathbf{u}^i + (\nabla \mathbf{u}^i)^{\top} \right)\right]  : \nabla \overline{\mathbf{u}^s} \, \mathrm{d} \mathbf{x} \notag \\
    &= \int_{0}^{2\pi}\int_{0}^{\pi}\int_{1}^{\zeta_2}
    \sum_{m=-n}^n |f_{n,m}|^2 \frac{n^2 (2n+1) k^{2n+2} \tau^{2n+2} \left( \lambda n(n+1) + 2\mu (n^2 + n + 1) \right)}{2\pi [(2n+1)!!]^3 (2n-3)!! (\lambda + 2\mu)^{2n+2}} \cdot \sin \theta \, \mathrm{d}r \, \mathrm{d}\theta \, \mathrm{d}\phi \notag\\
    &\quad  \left( \mathcal{O}(1) + \mathcal{O}( k^2 \tau^2 ) \right) \notag \\
    &=\sum_{m=-n}^n |f_{n,m}|^2 \frac{2 n^2 (2n+1) k^{2n+2} \tau^{2n+2} \left( \lambda n(n+1) + 2\mu (n^2 + n + 1) \right) (\zeta_2 - 1)}{[(2n+1)!!]^3 (2n-3)!! (\lambda + 2\mu)^{2n+2}} \cdot \left( \mathcal{O}(1) + \mathcal{O}(  k^2 \tau^2 ) \right).
\end{align}
    	Combining the fact that $\zeta_2 \in (1, R)$ with the inequality $1 - \frac{1}{\zeta_2^{2n+1}} \geqslant 1-\frac{1}{\zeta_2}$,  we can derive that
    \begin{align}\label{eq:427}
    	\frac{\left(2+\frac{3}{n}\right)\left(1 - \frac{1}{\zeta_2^{2n+1}}\right)}{\left(2+\frac{1}{n}\right)^{3}\left(1+\frac{1}{n}\right)} \geqslant \frac{\zeta_2-1}{27\zeta_2}.
    \end{align}
    Under the assumptions given in \eqref{eq:assume} and \eqref{eq:ks kp defination}, consider a given $\tau$ as defined in \eqref{eq:delta tau} with $\tau < 1$. Select an appropriate $\zeta_2 > 1$ such that $\zeta_2 \tau < 1$. Then, leveraging \eqref{eq:ui L2 norm}, \eqref{eq:EusL2}, \eqref{eq:Eui}, \eqref{eq:Rest}, and \eqref{eq:427}, when the index $n$ governing $\mathbf{u}^i$ (as given in \eqref{eq:ui def}) satisfies \eqref{eq:n4 in 4.1} and the high-oscillation parameter $M$ satisfies \eqref{eq:M1 in 4.1}, direct computation yields the desired result.
    \begin{align}
    	\frac{E(\mathbf{u})}{\|\bmf u^i\|_{L^2(D)^3}^2}
    	=&\frac{E(\mathbf{u}^s)+E(\mathbf{u}^i)+ Rest}{\sum_{m = - n}^{n}\frac{4 \pi \left|f_{n,m}\right|^{2}n(n + 1)(k\tau)^{2n}}{(2n + 3)[(2n+1)!!]^{2}(\lambda+2\mu)^{n}}\left(1 - \mathcal{O}\left(k^{2}\tau^{2}\right)\right)} \notag\\
    	=&\frac{E(\mathbf{u}^s)\left(1+\mathcal{O}\left(\frac{n(\tau\zeta_{2})^{2n}\zeta_2}{k^2\tau^2}\right)\right)}{\sum_{m = - n}^{n}\frac{4 \pi \left|f_{n,m}\right|^{2}n(n + 1)(k\tau)^{2n}}{(2n + 3)[(2n+1)!!]^{2}(\lambda+2\mu)^{n}}\left(1 - \mathcal{O}\left(k^{2}\tau^{2}\right)\right)} \label{eq:420}\\
    	=&\frac{\sum_{m = -n}^{n}  \frac{4\pi\left|f_{n,m}\right|^2  (\lambda+3\mu)n^6  k^{2n+2} \tau^2 (\zeta_2^{2n + 1}-1)}{[(2n + 1)!!]^2(\lambda + 2\mu)^{n+3}(2n + 1)^3\zeta_2^{2n+1}} \left(\mathcal{O}(1)+\mathcal{O}\left(\frac{n(\tau\zeta_{2})^{2n}\zeta_2}{k^2\tau^2}\right)\right)}{\sum_{m = - n}^{n}\frac{4 \pi \left|f_{n,m}\right|^{2}n(n + 1)(k\tau)^{2n}}{(2n + 3)[(2n+1)!!]^{2}(\lambda+2\mu)^{n}}\left(1 - \mathcal{O}\left(k^{2}\tau^{2}\right)\right)}\notag\\
    	=&\frac{(\lambda + 3\mu)n^{5}(2n + 3)k^{2}(\zeta_2^{2n+1}-1)}{(\lambda + 2\mu)^{3}(2n + 1)^{3}(n + 1)\zeta_2^{2n+1}\tau^{2n - 2}}
    	\left(\mathcal{O}(1)+\mathcal{O}\left(\frac{n(\tau\zeta_{2})^{2n}\zeta_2}{k^2\tau^2}\right)\right)
    	\left(1 + \mathcal{O}\left(k^{2}\tau^{2}\right)\right) \notag\\
    	\geqslant& \frac{n^{2}k^{2}}{(\lambda + 2\mu)^{2}\tau^{2n - 2}}\frac{\left(2+\frac{3}{n}\right)\left(1 - \frac{1}{\zeta_2^{2n+1}}\right)}{\left(2+\frac{1}{n}\right)^{3}\left(1+\frac{1}{n}\right)} \geqslant \frac{n^{2}\left(\zeta_2-1\right)k^{2}}{27\zeta_2(\lambda + 2\mu)^{2}\tau^{2n - 2}}
        > M \gg 1. \label{eq:435}
    \end{align}

    The proof is complete.
    \end{proof}

\begin{rem}\label{rem:41rem}

The stress concentration behavior of the exterior total field $\mathbf{u}|_{\mathbb{R}^3 \setminus \overline{D}}$ is characterized by analyzing $E(\mathbf{u})$, as defined in \eqref{eq:Eu}, which represents the physical stress state \cite{LB09,LLH06,ROO,ZWH}. From \eqref{eq:420}, the stress concentration of $\mathbf{u}|_{\mathbb{R}^3 \setminus \overline{D}}$ is governed by the exterior scattering field $\mathbf{u}^s|_{\mathbb{R}^3 \setminus \overline{D}}$. Recall that the ratio ${\sf r}{E(\mathbf{u}^i), E(\mathbf{u}^s)}$, which characterizes the relationship between $E(\mathbf{u}^i)$ and $E(\mathbf{u}^s)$, is defined by \eqref{eq:Eui Eus}. Since the wave speed ratio $\tau$, as defined in \eqref{eq:delta tau}, is fixed and satisfies $\tau < 1$, and $\zeta_2$ fulfills $\zeta_2 \tau < 1$, even for a fixed $k \ll 1$ in the sub-wavelength regime, when the index $n$ of the incident wave $\mathbf{u}^i$ is sufficiently large, the quantity
$$\frac{n (\tau \zeta_2)^{2n} \zeta_2}{k^2 \tau^2}$$
remains bounded. Consequently, the concentration of $E(\mathbf{u})$ is primarily contributed by that of $E(\mathbf{u}^s)$.

\end{rem}

\begin{rem}\label{rem:4.2}

From the discussion in Remark~\ref{rem:33rem}, if the index $n$ of the incident wave $\mathbf{u}^i$ satisfies \eqref{eq:n4 in 4.1} and $M$, which characterizes the stress concentration level of $E(\mathbf{u})$, satisfies \eqref{eq:M1 in 4.1}, then the quasi-Minnaert resonance of  the exterior scattering field $\mathbf{u}^s|_{\mathbb{R}^3 \setminus \overline{D}}$ occurs. Furthermore, from Remark~\ref{rem:41rem}, we conclude that this quasi-Minnaert resonance of $\mathbf{u}^s|_{\mathbb{R}^3 \setminus \overline{D}}$ induces the concentration of $E(\mathbf{u})$, since the stress concentration of the exterior total field $\mathbf{u}|_{\mathbb{R}^3 \setminus \overline{D}}$ is primarily due to that of the scattered field $\mathbf{u}^s|_{\mathbb{R}^3 \setminus \overline{D}}$.

\end{rem}

 \begin{rem}\label{rem:43rem}

From Theorems~\ref{thm:th3.1},~\ref{thm:nabla u in thm}, and~\ref{thm:Eu definition in thm}, along with Remarks~\ref{rem:rem32}--\ref{rem:33rem}, we observe that varying conditions on the index $n$ of the incident wave $\mathbf{u}^i$, the high-oscillation parameter, and the stress concentration level $M\gg 1$ give rise to diverse physical phenomena in the interior total field $\mathbf{u}|_D$, the exterior scattering field $\mathbf{u}^s|_{\mathbb{R}^3\setminus\overline{D}}$, and the exterior total field $\mathbf{u}|_{\mathbb{R}^3\setminus\overline{D}}$. In particular, these conditions can trigger boundary localization, surface resonance, and stress concentration-either independently or in combination-within the respective wave fields. The key results are summarized in Table~\ref{tab:condition}. The abbreviations in Table~\ref{tab:condition} are defined as follows: BL for boundary localization, SR for surface resonance, SC for stress concentration, and QMR for quasi-Minnaert resonance. Additionally, the symbol $\times$ denotes cases that are not investigated in the present work.

\begin{table}[h]
\centering
\caption{Boundary localization, surface resonance, and stress concentrations of $\mathbf{u}|_D$, $\mathbf{u}^s|_{\mathbb{R}^3 \setminus \overline{D}}$, and $\mathbf{u}|_{\mathbb{R}^3 \setminus \overline{D}}$ under different conditions of $n$ or $M$.}
\begin{tabular}{|c|c|c|c|}
\hline
     Conditions for  $n$ or $M$   &  $\mathbf{u}|_D$ & $\mathbf{u}^s|_{\mathbb{R}^3 \setminus \overline{D}}$   &  $\mathbf{u}|_{\mathbb{R}^3 \setminus \overline{D}}$ \\
\hline
\eqref{eq:n1 n2 def} & BL &  BL &  $\times$  \\   
\hline
 \eqref{eq:N2 def} & SR  & SR  &  $\times$  \\   
\hline
 \eqref{eq:355} & QMR    & QMR  &  $\times$  \\  
\hline
 \eqref{eq:356} &  $\times$ & QMR  &  $\times$  \\   
\hline
 \eqref{eq:357} & QMR  &  $\times$ &  $\times$  \\  
\hline
 \eqref{eq:M0 def} & QMR   &  $\times$ &  $\times$  \\  
 \hline
 \eqref{eq:M1 def} &  $\times$  & QMR &  $\times$  \\    
  \hline
\eqref{eq:M1 in 4.1}  \mbox{and}  \eqref{eq:n4 in 4.1} & $\times$ & QMR and SC & SC   \\  
\hline
\end{tabular}
\label{tab:condition}
\end{table}

\end{rem}

\section{Numerical experiments}\label{sec:5}

 This section presents extensive numerical examples to validate the theoretical results established in the preceding section. As established in Theorem \ref{thm:Eu definition in thm}, we have rigorously demonstrated that in the three-dimensional case, for a model with high-contrast densities between the bubble \(D\) and the surrounding elastic medium \(\mathbb{R}^3 \backslash \overline{D}\), by selecting an appropriate incident wave, we can observe the stress concentration phenomenon in the external total field \(\mathbf{u}\), as shown in \eqref{eq:Eus th4.2}. Stress concentration can be achieved via quasi-Minnaert resonance, which involves boundary localization and surface resonance. In Theorem \ref{thm:th3.1} and Theorem \ref{thm:nabla u in thm}, we have respectively proven that, by selecting an appropriate $\mathbf{u}^i$, the internal total field \(\mathbf{u}\) and the external scattered field \(\mathbf{u}^s\) exhibit boundary localization as shown in \eqref{eq:328}, and surface resonance as demonstrated in \eqref{eq:nablau nabla us gg 1} and \eqref{eq:43}.

 The above conclusion can also be extended to the two-dimensional case. The subsequent numerical examples illustrate these effects for radially  geometries in $\mathbb{R}^2$. Furthermore, they demonstrate that, under general conditions, the quasi-Minnaert resonance and stress concentration phenomena persist through the selection of appropriate incident waves $\mathbf{u}^i$, as defined in \eqref{eq:ui in 2D}, by exploiting the high-contrast between media. Consequently, by strategically designing high-contrast structures and incident waves $\mathbf{u}^i$, stress concentration can be induced for a general shape bubble  in $\mathbb{R}^2$.

 As noted in Remark \ref{rem:2.2}, in the sub-wavelength regime, the bubble $D$ is considerably smaller than the incident wavelength. To quantify this effect, in subsequent numerical experiments, we normalize the $L^2$-norm of the incident wave $\mathbf{u}^i$, as defined in \eqref{eq:ui in 2D}, within $D$. We utilize numerical methods implemented in COMSOL Multiphysics to solve the integral equation \eqref{eq:system}. Through several numerical examples, we demonstrate stress concentration in the external total field $\mathbf{u}$, achieved through quasi-Minnaert resonance. These findings are evident in both two-dimensional radially   geometries and other two-dimensional general  bubbles, and are validated in three-dimensional spherical scenarios.

 In this section, we consider test domains for the high-contrast bubble \( D \), including a disk, a corner-shaped domain, and an apple-shaped domain in two dimensions. In three dimensions, the validity of the theoretical results is verified using the spherical geometry. The boundaries of the bubbles  are parametrized  as follows:
\begin{equation}
	\label{eq:test_domains}
	\begin{aligned}
		&\text{Circle:} \quad &\{(x_1,x_2) \mid x_1^2 + x_2^2 = 1\}, \\
		&\text{Corner:} \quad &\{(x_1,x_2) \mid \text{boundary characterized by parameterization } g(t), \, 0 \leq t \leq 2\pi\}, \\
		&\text{Apple:} \quad &\{(x_1,x_2) \mid \text{boundary characterized by parameterization } h(t), \, 0 \leq t \leq 2\pi\}, \\
		&\text{Sphere:} \quad &\{(x_1,x_2,x_3) \mid x_1^2 + x_2^2 + x_3^2 = 1\},
	\end{aligned}
\end{equation}
where $g(t)$ and $h(t)$ are given by:
\begin{align}\label{eq:boundary_param}
	g(t) =(1 - \cos t)\bigl( \cos t,\, -\sin t \bigr), \quad h(t) =\frac{0.48 \left(1 + 0.9\cos t + 0.1\sin 2t\right)}{1 + 0.75\cos t} \left( \cos t,\, \sin t \right).
\end{align}

Next, we adopt the physical model from \cite{CTL12}, which involves thin layers of polydimethylsiloxane (PDMS) enclosing a bubble. This configuration yields the following physical parameters:
\begin{align}\label{eq:parameter}
	\text{Elastic material parameters of PDMS:} &\ \rho_e = 1042 \ \mathrm{kg/m^3}, \ \tilde{\lambda} = 1.083 \times 10^9 \ \mathrm{Pa}, \ \tilde{\mu} = 6.5 \times 10^5 \ \mathrm{Pa}; \notag\\
	\text{Bubble physical parameters:} & \ \rho_b = 1.2 \ \mathrm{kg/m^3}, \quad \kappa = 1.412 \times 10^5 \ \mathrm{Pa}.
\end{align}

For the purpose of the study, we consider non-dimensionalizing the aforementioned parameters as shown in \eqref{eq:non dim}, while combining the definitions of parameters $\delta$, $\tau$, $k_p$, and $k_s$ from \eqref{eq:delta tau} and \eqref{eq:ks kp defination}. We select an incident frequency of $\omega=0.1$ Hz, and naturally, we can obtain the following dimensionless parameters.
 \begin{align}\label{eq:5.4}
 	k &= 2.9152 \times 10^{-4}, & \tau &= 0.33627, \notag\\
 	k_p &= 9.803 \times 10^{-5}, & k_s &= 1.2159 \times 10^{-7}, & \delta &= 1.1516 \times 10^{-3}. 
 \end{align}

 In $\mathbb{R}^2$, we consider the incident wave as shown below
 \begin{align}\label{eq:ui in 2D}
 \mathbf{u}^i=\frac{\widetilde{ \mathbf{u}^i } } { \|\widetilde{ \mathbf{u}^i }\|_{L^2(D)^2} },\quad 
 	\widetilde{\mathbf{u}^i} =   k_p J_n'(k_p \lvert \mathbf{x} \rvert) e^{\rmi n\theta} \cdot \hat{\boldsymbol{r}}  +   \frac{\rmi n}{\lvert \mathbf{x} \rvert} J_n(k_p \lvert \mathbf{x} \rvert) e^{\rmi n\theta} \cdot \hat{\boldsymbol{\theta}},
 \end{align}
 with 
   \[
   \hat{\boldsymbol{r}} := \begin{pmatrix} \cos\theta \\ \sin\theta \end{pmatrix}, \quad \text{and} \quad 
   \hat{\boldsymbol{\theta}}:= \begin{pmatrix} -\sin\theta \\ \cos\theta \end{pmatrix},
   \]
 and  $J_n(z)$ denotes the Bessel function of the first kind.

   In $\mathbb{R}^3$, the incident wave $\mathbf{u}^i$ is given by:  
 	\begin{align}
            \mathbf{u}^i&=\frac{\widetilde{ \mathbf{u}^i } } { \|\widetilde{ \mathbf{u}^i }\|_{L^2(D)^3} },\label{eq:ui in 3D}  \\
 		\widetilde {\mathbf{u}^i} =   j_{n}^{\prime}\left(k_{p}|\mathbf{x}|\right) e^{\mathrm{i}n\varphi} P_{n}^{|n|}(\cos\theta)  \hat{r} 
 		&+   \frac{j_{n}\left(k_{p}|\mathbf{x}|\right)}{k_{p}|\mathbf{x}|} e^{\mathrm{i}n\varphi} \frac{dP_{n}^{|n|}(\cos\theta)}{d\theta}  \hat{\theta}   +   \frac{j_{n}\left(k_{p}|\mathbf{x}|\right)}{k_{p}|\mathbf{x}|} \frac{\mathrm{i}n}{\sin\theta} e^{\mathrm{i}n\varphi} P_{n}^{|n|}(\cos\theta)  \hat{\varphi}, \notag
 	\end{align}
 where 
\begin{equation}\notag
	\hat{r} := \begin{pmatrix}
		\sin\theta \cos\varphi \\
		\sin\theta \sin\varphi \\
		\cos\theta
	\end{pmatrix}, \quad
	\hat{\theta} := \begin{pmatrix}
		\cos\theta \cos\varphi \\
		\cos\theta \sin\varphi \\
		-\sin\theta
	\end{pmatrix}, \quad
	\hat{\varphi} := \begin{pmatrix}
		-\sin\varphi \\
		\cos\varphi \\
		0
	\end{pmatrix},
\end{equation}
 and $P_n^{|n|}$ is the associated Legendre function of degree $n$ and order $n$. The $j_n(z)$ denotes the spherical Bessel function of the first kind.

 \subsection{$D$ is a unit disk}

   In this subsection, we examine the scenario where the bubble $D$ is a unit disk. By analyzing the physical model defined in \eqref{eq:parameter} and selecting the incident wave specified in \eqref{eq:ui in 2D}, we demonstrate that, through suitable adjustments to the incident wave and the bubble-elastic structure, the exterior total field $\mathbf{u}|_{\mathbb{R}^2 \setminus \overline{D}}$ exhibits stress concentration, while simultaneously  the interior total field $\mathbf{u}|_D$ and the exterior scattered field $\mathbf{u}^s|_{\mathbb{R}^2 \setminus \overline{D}}$ exhibit quasi-Minnaert resonance. This confirm that the  quasi-Minnaert resonance can induce  stress concentration.

  	\begin{exm}	\label{exm:1}
  	\textbf{Stress concentration of $E(\mathbf{u})|_{B_2 \setminus \overline{D}}$}.

    In this example, Figure~\ref{fig:5} illustrates the stress concentration phenomenon for elastic waves corresponding to different indices $n=5,15,25$ of the incident wave $\mathbf{u}^i$. The figure depicts the numerical values of the stress $\mathcal{E}(\mathbf{u})$, defined in \eqref{eq:Eu223}, within the domain $B_2 \setminus \overline{D}$. Blue regions represent small values of $\mathcal{E}(\mathbf{u})$, while red regions indicate large values. For $n=5,15,25$, the stress $\mathcal{E}(\mathbf{u})$ concentrates near $\partial D$. Notably, as $n$ increases, the maximum values of $\mathcal{E}(\mathbf{u})$ grow from $\mathcal{O}(10)$ to $\mathcal{O}(10^{17})$. This verifies the estimate in \eqref{eq:Eus th4.2} from Theorem~\ref{thm:Eu definition in thm}. Furthermore, as the index $n$ of $\mathbf{u}^i$ increases, local regions of relatively high stress emerge near $\partial D$, symmetrically distributed along the $x_1$-axis. Rigorous analysis of this observation will be pursued in future research.

  \end{exm}

 \begin{exm}\label{exm:2}
 	\textbf{Boundary localization of  $u|_D$ and $\mathbf{u}^s|_{B_2 \setminus \overline{D}}$ }.

In this example, we consider incident waves $\mathbf{u}^i$ with indices $n=5,15,25$, as defined in \eqref{eq:ui in 2D}. The $L^2$-norms of the interior total field $u|_D$ and the exterior scattered field $\mathbf{u}^s|_{B_2 \setminus \overline{D}}$ are depicted in Figures~\ref{fig:u_circle} and~\ref{fig:us_circle}, respectively, where $B_2$ denotes the ball of radius $2$ centered at the origin. From the first column of Figure~\ref{fig:u_circle}, it is evident that as $n$ increases, the blue region expands toward $\partial D$, indicating boundary localization of $u|_D$ near $\partial D$. The second column of Figure~\ref{fig:u_circle} illustrates the $L^2$-norm of $u|_D$ near the point $(1,0)$; as $n$ increases, the red regions narrow and approach $\partial D$, signifying that the $L^2$-norm is predominantly concentrated near the boundary. Similarly, the dense contour lines near the boundary in the third column of Figure~\ref{fig:u_circle} support this conclusion. In addition, from Figure~\ref{fig:us_circle}, we observe that the exterior scattered field $\mathbf{u}^s|_{B_2 \setminus \overline{D}}$ also exhibits boundary localization near $\partial D$. More importantly, local red regions exhibiting higher values than the surrounding periodic patterns appear in image (i) of Figure~\ref{fig:us_circle} when $n=25$.

For $\zeta_1=0.9$ and $\zeta_2=1.1$, the boundary localization ratios, as defined in \eqref{eq:226}, are reported for various values of $n$ in Table~\ref{tab:1}. These values indicate that the boundary localization levels $\eta_{\mathbf{u}}$ and $\eta_{\mathbf{u}^s}$ decrease as $n$ increases. Smaller values of these ratios correspond to stronger boundary localization for $\mathbf{u}$ and $\mathbf{u}^s$, respectively. From Table~\ref{tab:1}, we conclude that for $n=60$, approximately $93\%$ of the energy of $\mathbf{u}$ is concentrated in $D \backslash B_{0.9}$, and approximately $77\%$ of the energy of $\mathbf{u}^s$ is concentrated in $B_{1.1} \backslash \overline{D}$. This confirms that both $\mathbf{u}$ and $\mathbf{u}^s$ exhibit boundary localization near $\partial D$. These numerical results strongly validate the findings in Theorem~\ref{thm:th3.1}.

 	 \end{exm}

 	\begin{exm}\label{exm:3}
 		\textbf{Surface resonance of  $\nabla u|_D$  and   $\nabla\mathbf{u}^s|_{B_2 \setminus \overline{D}}$ }.

      Figures~\ref{fig:3} and~\ref{fig:nus_grid} present the values of $\|\nabla u\|_{L^2(D)^2}$ and $\|\nabla\mathbf{u}^s\|_{L^2(B_2\setminus\overline{D})^2}$, respectively, for different indices $n$ ($n=5,15,25$) of the incident wave $\mathbf{u}^i$. The red regions indicate high-oscillation domains, while the blue regions represent areas where the gradient is nearly zero. As the index $n$ increases, the blue domains expand and gradually approach the boundary, leading to the conclusion that the high-oscillation phenomenon occurs near $\partial D$. Although the index $n$ of $\mathbf{u}^i$ increases from $5$ to $25$, we observe that the $L^2$-norm of $\nabla u|_D$ increases from $\mathcal{O}(10^2)$ to $\mathcal{O}(10^{17})$ in Figure~\ref{fig:3}, and the $L^2$-norm of $\nabla\mathbf{u}^s|_{B_2\setminus\overline{D}}$ increases from $\mathcal{O}(10^2)$ to $\mathcal{O}(10^8)$ in Figure~\ref{fig:nus_grid}. We can thus conclude that the degree of oscillation for $\nabla u|_D$ is stronger than that for $\nabla\mathbf{u}^s|_{B_2\setminus\overline{D}}$. These numerical values are consistent with the mathematical analysis in \eqref{eq:nablau nabla us gg 1} and \eqref{eq:43} of Theorem~\ref{thm:nabla u in thm}.

 	\end{exm}

 	\begin{figure}
 		\centering
 		\begin{subfigure}[t]{0.3\textwidth}
 			\centering
 			\includegraphics[width=\textwidth]{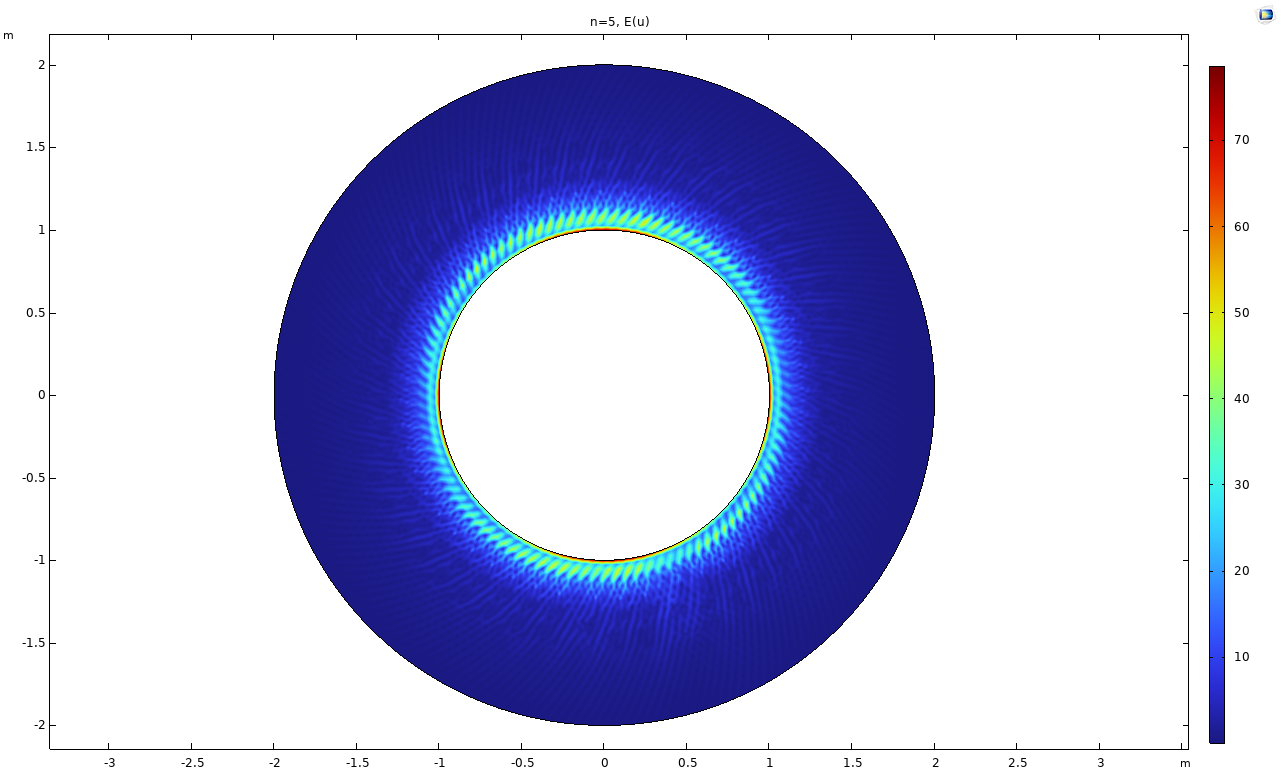} 
 			\caption{$n=5$,\ ${\mathcal E}(\mathbf{u})$}
 		\end{subfigure} 
 		\hfill 
 		\begin{subfigure}[t]{0.3\textwidth}
 			\centering
 			\includegraphics[width=\textwidth]{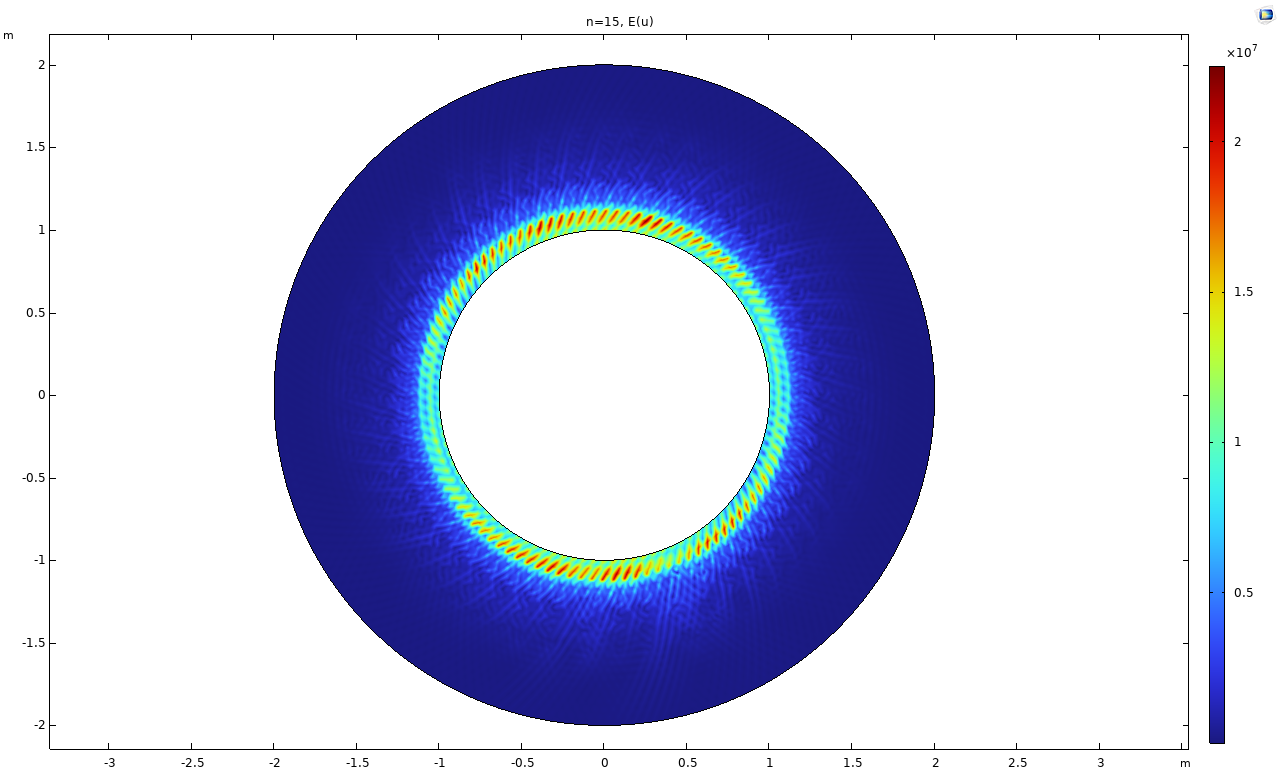}
 			\caption{$n=15$,\ ${\mathcal E}(\mathbf{u})$}
 		\end{subfigure}%
 		\hfill%
 		\begin{subfigure}[t]{0.3\textwidth}
 			\centering
 			\includegraphics[width=\textwidth]{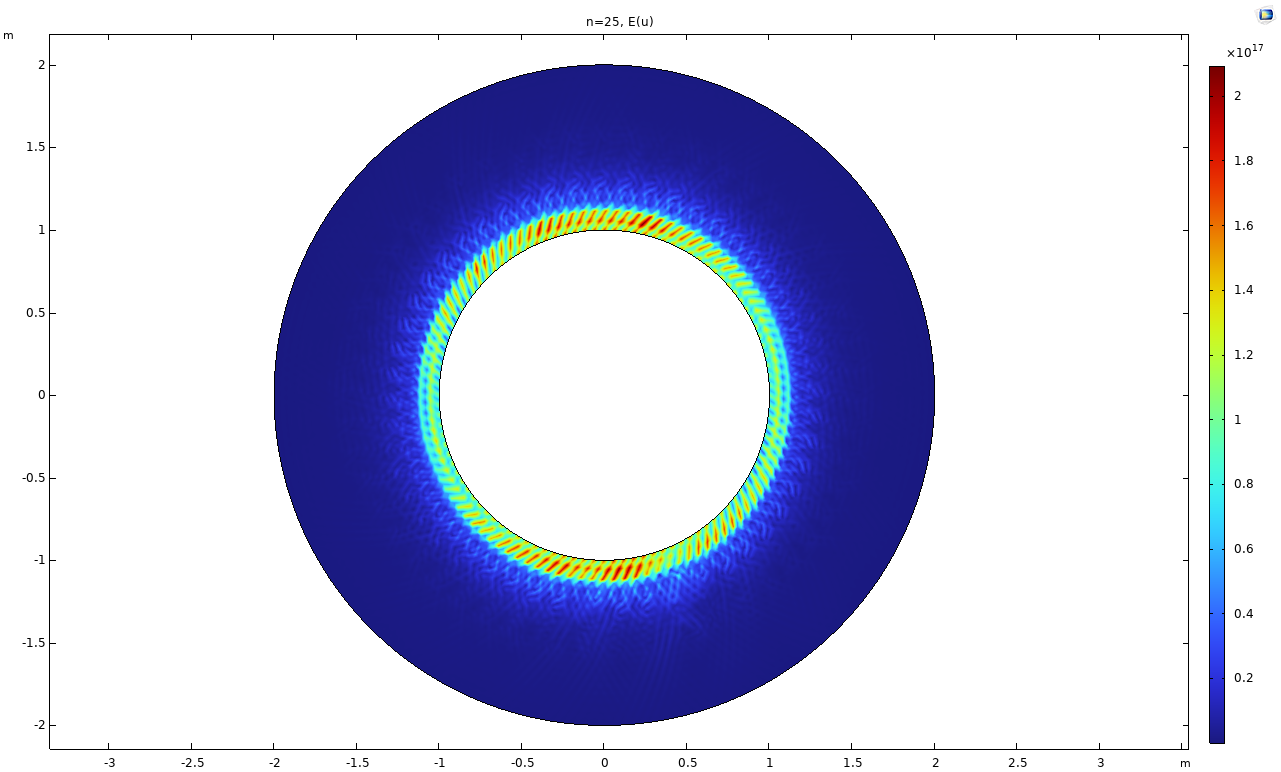}
 			\caption{$n=25$,\ ${\mathcal E}(\mathbf{u})$}
 		\end{subfigure}
 		\caption{The stress $\mathcal{E}(\mathbf{u})$ of the exterior total field $\mathbf{u}$ for the incident wave $\mathbf{u}^i$ with different indices $n$ ($n=5, 15, 25$).}
 		\label{fig:5}
 	\end{figure}

 \begin{figure}  
 	\centering
 	\subfloat[$n=5,  \|{u}\|_{L^2\left(D\right)^2}$ ]{\includegraphics[width=0.29\textwidth]{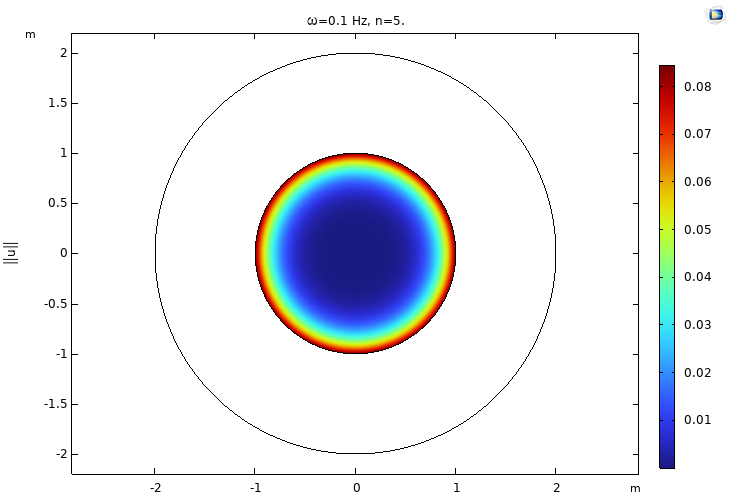}}
 	\hfill
 	\subfloat[ boundary localization of $u$ near the point $(1,0)$ for $n=5$ ]{\includegraphics[width=0.29\textwidth]{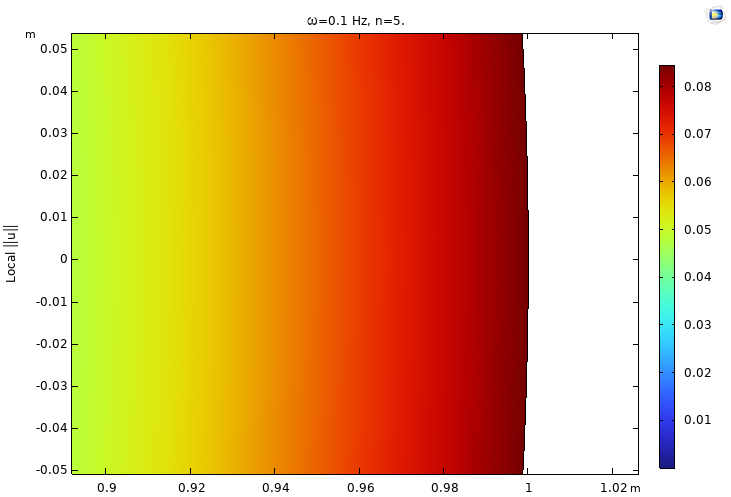}}
 	\hfill
 	\subfloat[for $n=5$: energy contours of $u$ near the boundary at $(1,0)$]{\includegraphics[width=0.29\textwidth]{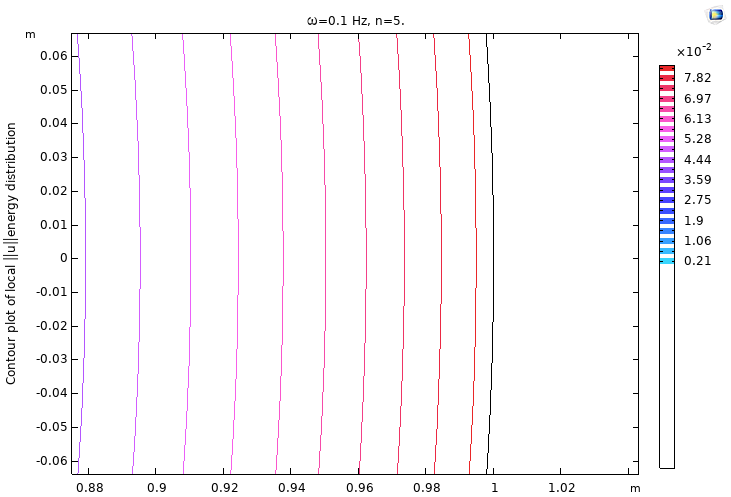}}
 	\vspace{0.1cm}  
 	\subfloat[$n=15,  \|{u}\|_{L^2\left(D\right)^2}$]{\includegraphics[width=0.29\textwidth]{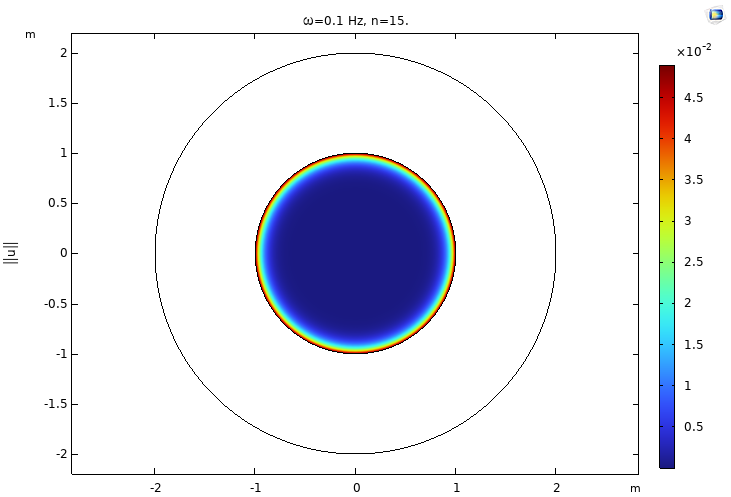}}
 	\hfill
 	\subfloat[ boundary localization of $u$ near the point $(1,0)$ for $n=15$]{\includegraphics[width=0.29\textwidth]{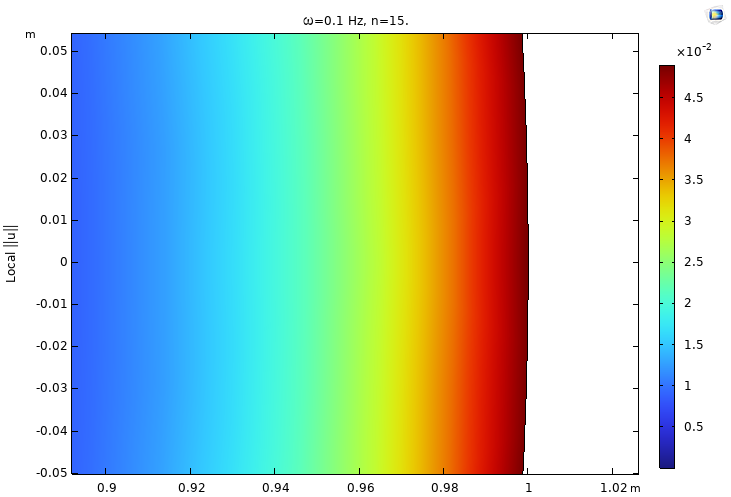}}
 	\hfill
 	\subfloat[for $n=15$: energy contours of $u$ near the boundary at $(1,0)$]{\includegraphics[width=0.29\textwidth]{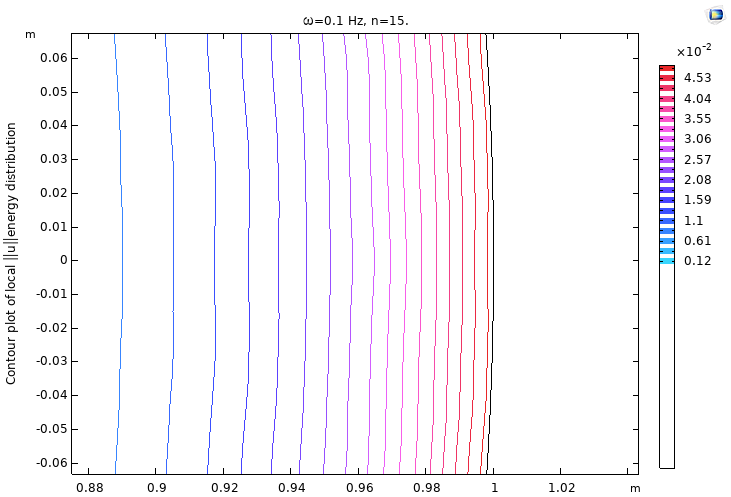}}
 	
 	\vspace{0.1cm}
 	\subfloat[$n=25,  \|{u}\|_{L^2\left(D\right)^2}$]{\includegraphics[width=0.29\textwidth]{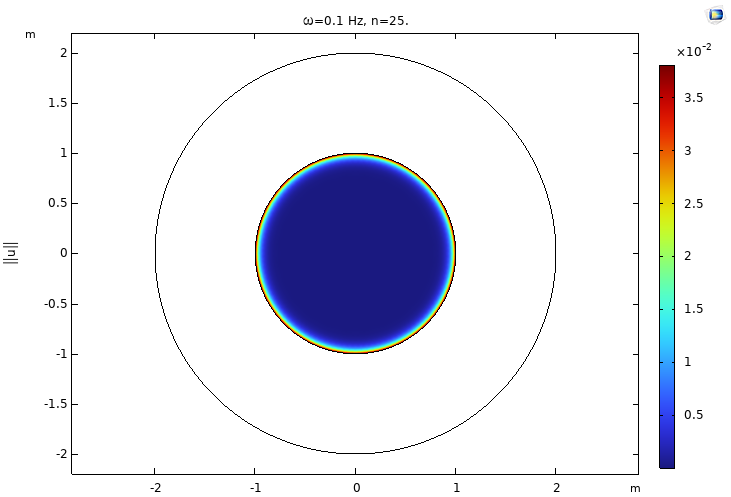}}
 	\hfill
 	\subfloat[boundary localization of $u$ near the point $(1,0)$ for $n=25$]{\includegraphics[width=0.29\textwidth]{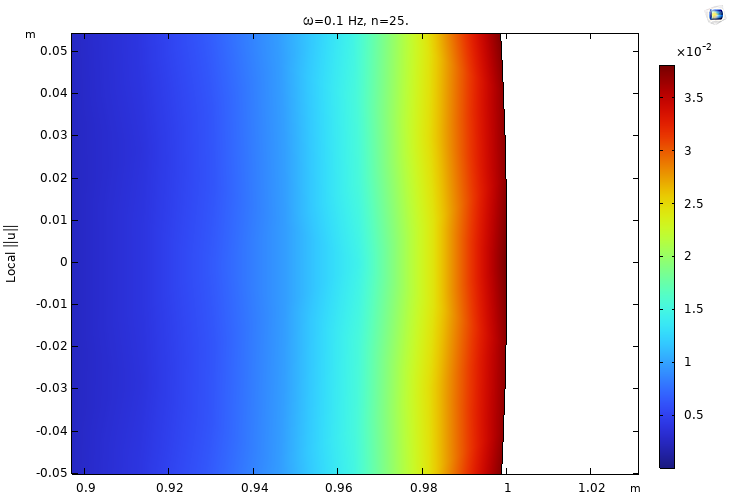}}
 	\hfill
 	\subfloat[for $n=25$: energy contours of $u$ near the boundary at $(1,0)$]{\includegraphics[width=0.29\textwidth]{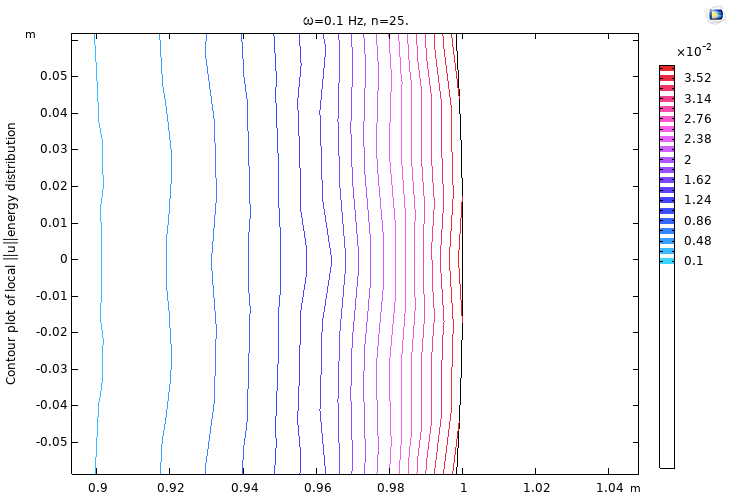}}
 	
 	\caption{ 
 		The $L^2$-norm of $u|_D$ for $\mathbf{u}^i$ with different indices $n$ ($n=5, 15, 25$), alongside the localized energy distribution near the point (1,0) and the corresponding contour plot.
 	}
 	\label{fig:u_circle}
 \end{figure}

 \begin{figure}
 	\centering
 	\subfloat[$n=5,\| \mathbf{u}^s\|_{L^2\left(B_2 \backslash \overline D\right)^2}$]{\includegraphics[width=0.29\textwidth]{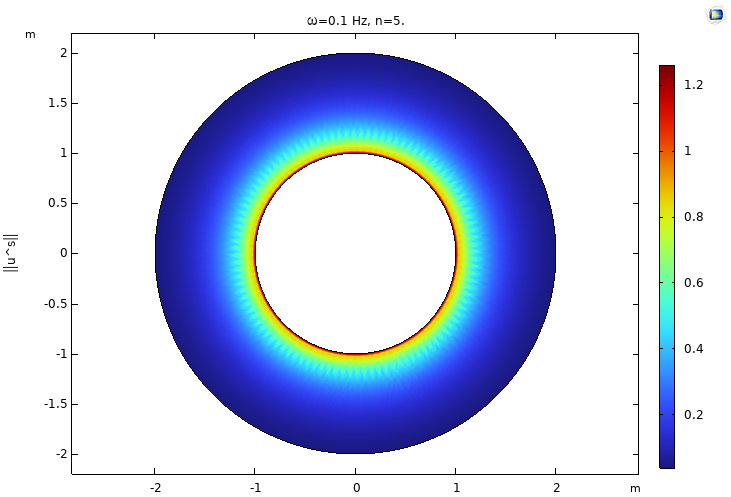}}
 	\hfill
 	\subfloat[ boundary localization of $\mathbf{u}^s$ near the point $(1,0)$ for $n=5$ ]{\includegraphics[width=0.29\textwidth]{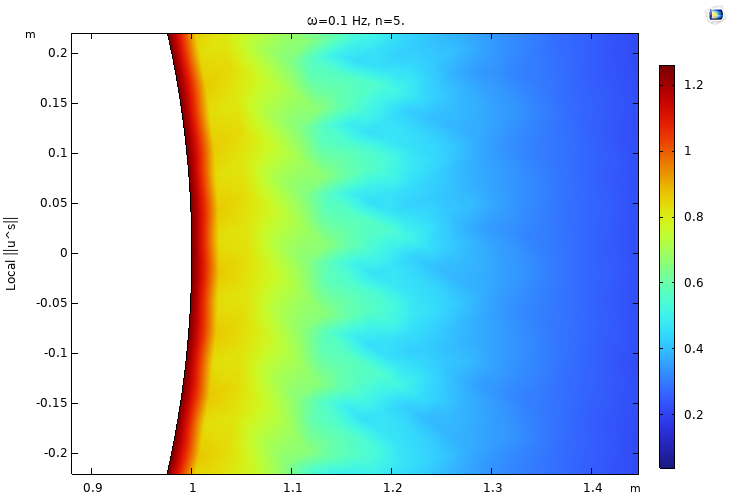}}
 	\hfill
 	\subfloat[for $n=5$: energy contours of $\mathbf{u}^s$ near the boundary at $(1,0)$]{\includegraphics[width=0.29\textwidth]{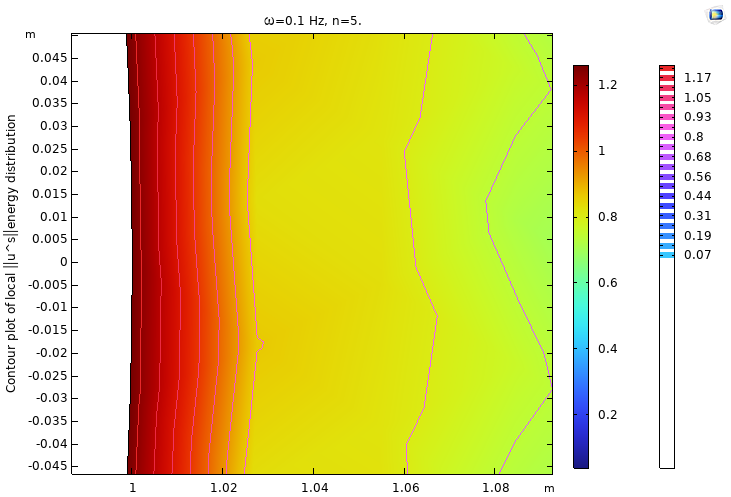}}
 	
 	\vspace{0.1cm}
 	\subfloat[$n=15,\| \mathbf{u}^s\|_{L^2\left(B_2 \backslash \overline D\right)^2}$]{\includegraphics[width=0.29\textwidth]{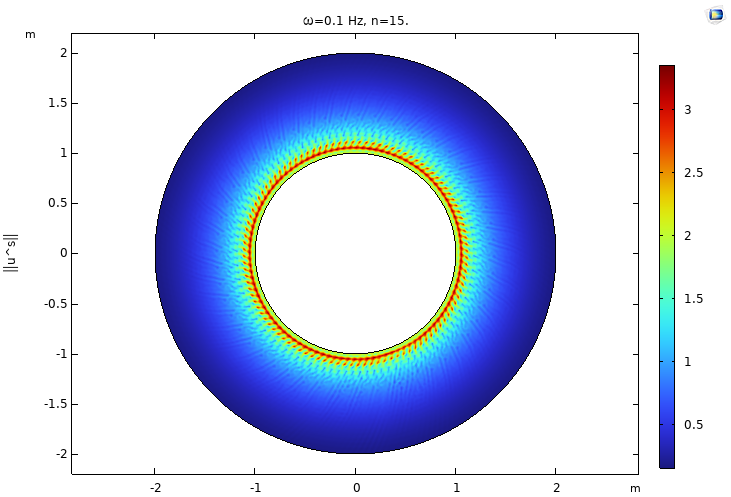}}
 	\hfill
 	\subfloat[boundary localization of $\mathbf{u}^s$ near the point $(1,0)$ for $n=15$]{\includegraphics[width=0.29\textwidth]{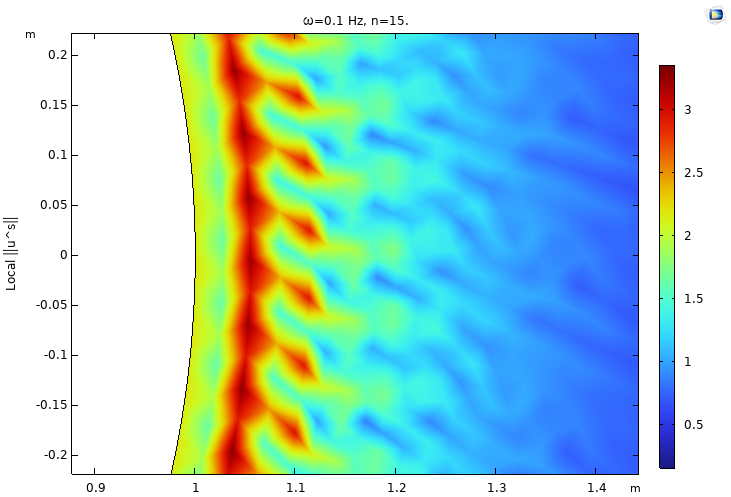}}
 	\hfill
 	\subfloat[ for $n=15$: energy contours of $\mathbf{u}^s$ near the boundary at $(1,0)$]{\includegraphics[width=0.29\textwidth]{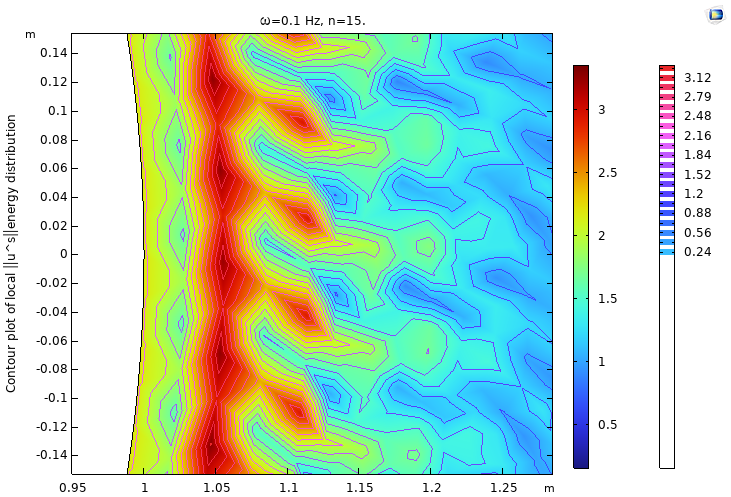}}
 	
 	\vspace{0.1cm}
 	\subfloat[$n=25,\| \mathbf{u}^s\|_{L^2\left(B_2 \backslash \overline D\right)^2}$]{\includegraphics[width=0.29\textwidth]{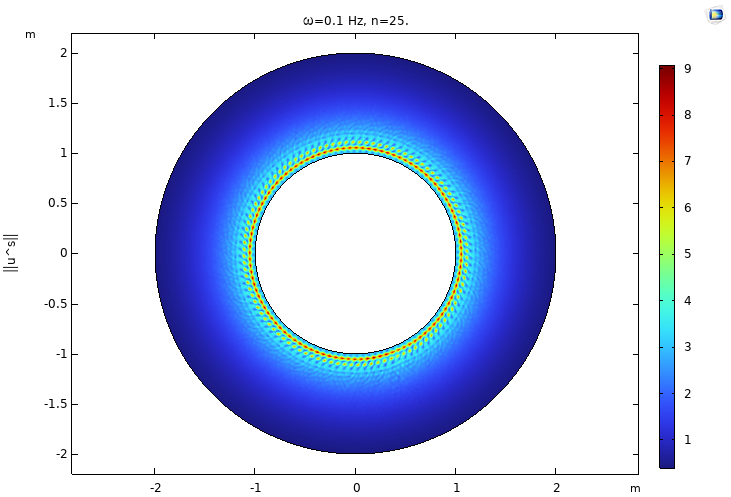}}
 	\hfill
 	\subfloat[boundary localization of $\mathbf{u}^s$ near the point $(1,0)$ for $n=25$]{\includegraphics[width=0.29\textwidth]{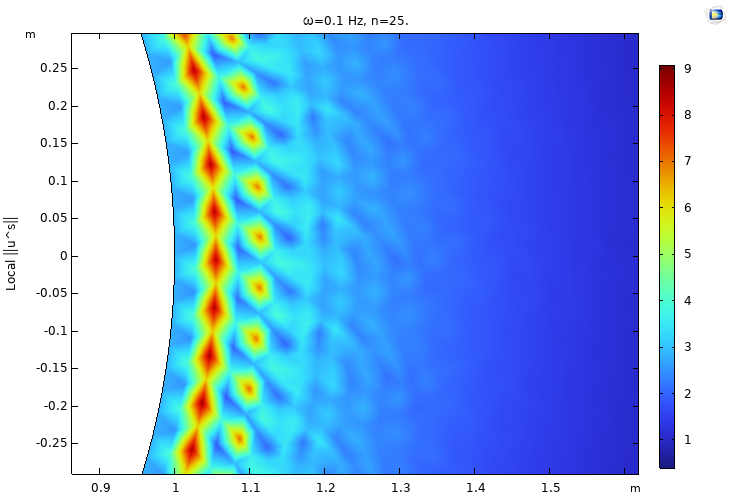}}
 	\hfill
 	\subfloat[ for $n=25$: energy contours of $\mathbf{u}^s$ near the boundary at $(1,0)$]{\includegraphics[width=0.29\textwidth]{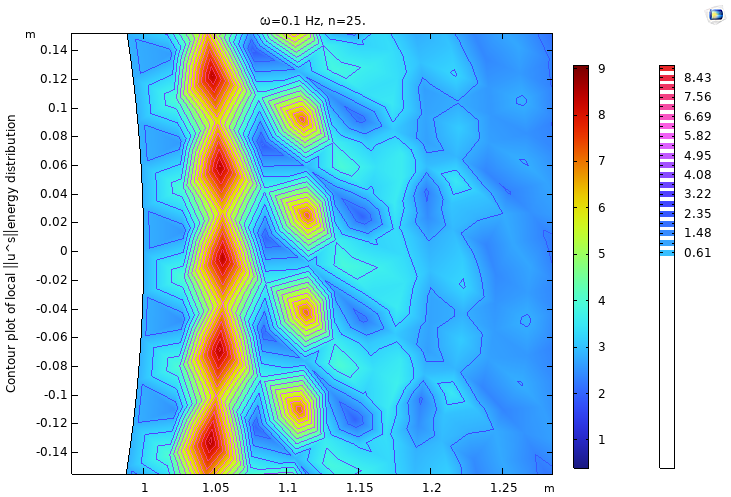}}
 	
 	\caption{
 		The $L^2$-norm of  \(\mathbf{u}^s|_{B_2\setminus\overline{D}}\) for \(\mathbf{u}^i\) with different values \(n\) ($n=5, 15, 25$), together with the localized energy distribution near (1,0) and the corresponding contour plot.
 	}
 	\label{fig:us_circle}
 \end{figure}

\begin{table}
	\centering
	\caption{Boundary localization ratios and  for different  $n$, and given $\zeta_1=0.9$ and $\zeta_2=1.1$. }
	\label{tab:1}
	\begin{tabular}{@{}ccc@{}}
		\toprule
		$n$ & $\eta_{u}$ & $\eta_{\mathbf{u}^s} $ \\
		\midrule
		  20 & 0.5904018869589628 &  0.4590269811218674     \\
		  40 &   0.2051670856315115   & 0.4198613576351685  \\
		  60 & 0.0720115793865058  & 0.2347708781772988 \\
		\bottomrule
	\end{tabular}
\end{table}

  \begin{figure}
  	\centering
  	\begin{subfigure}[t]{0.3\textwidth}
  		\centering
  		\includegraphics[width=\textwidth]{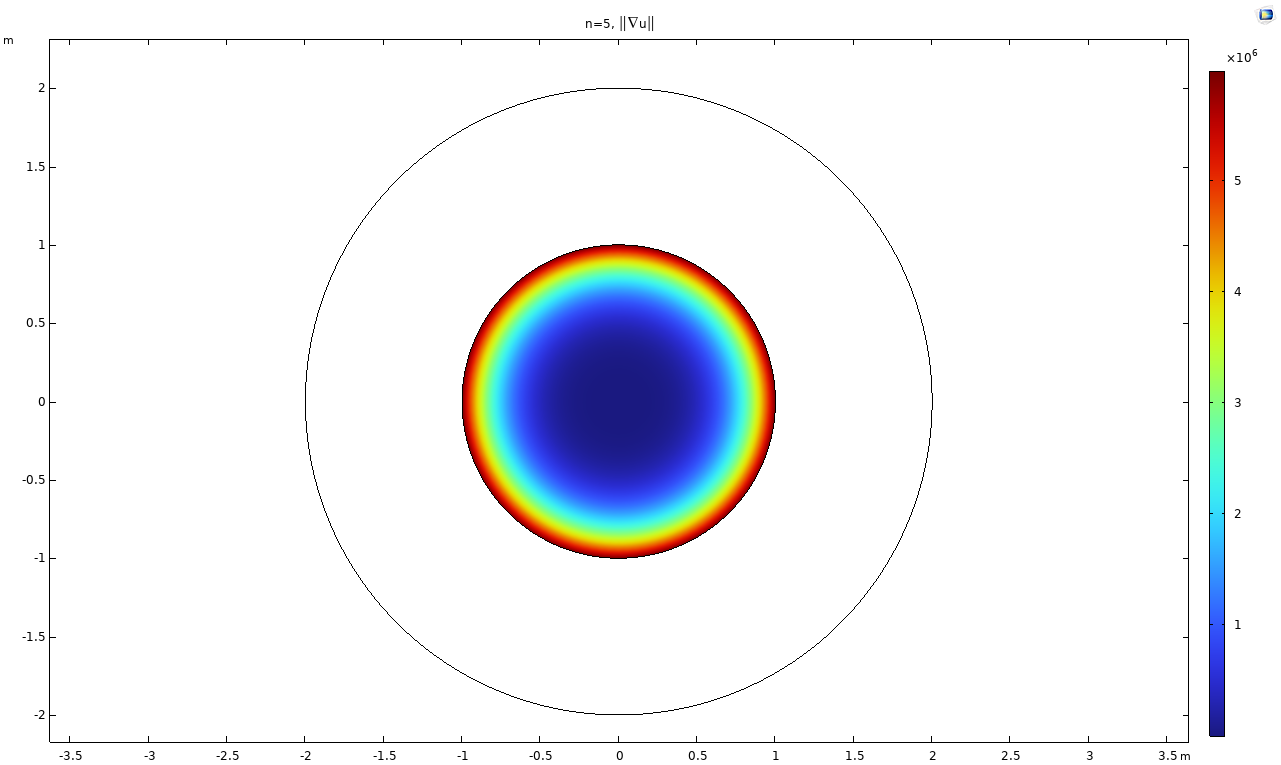} 
  		\caption{$n=5$,\ $\| \nabla u\|_{L^2\left(D\right)^2}$}
  	\end{subfigure}
  	\hfill
  	\begin{subfigure}[t]{0.3\textwidth}
  		\centering
  		\includegraphics[width=\textwidth]{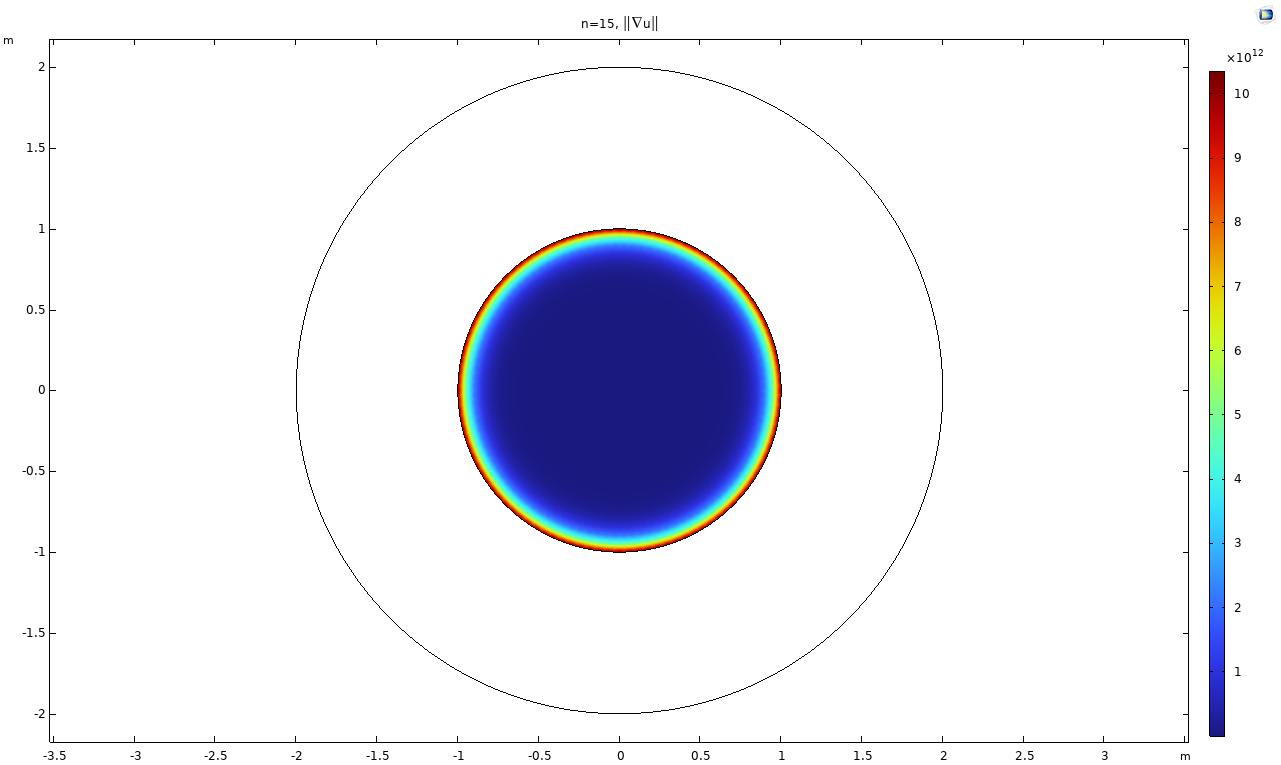}
  		\caption{$n=15$,\ $\| \nabla u\|_{L^2\left(D\right)^2}$}
  	\end{subfigure}%
  	\hfill%
  	\begin{subfigure}[t]{0.3\textwidth}
  		\centering
  		\includegraphics[width=\textwidth]{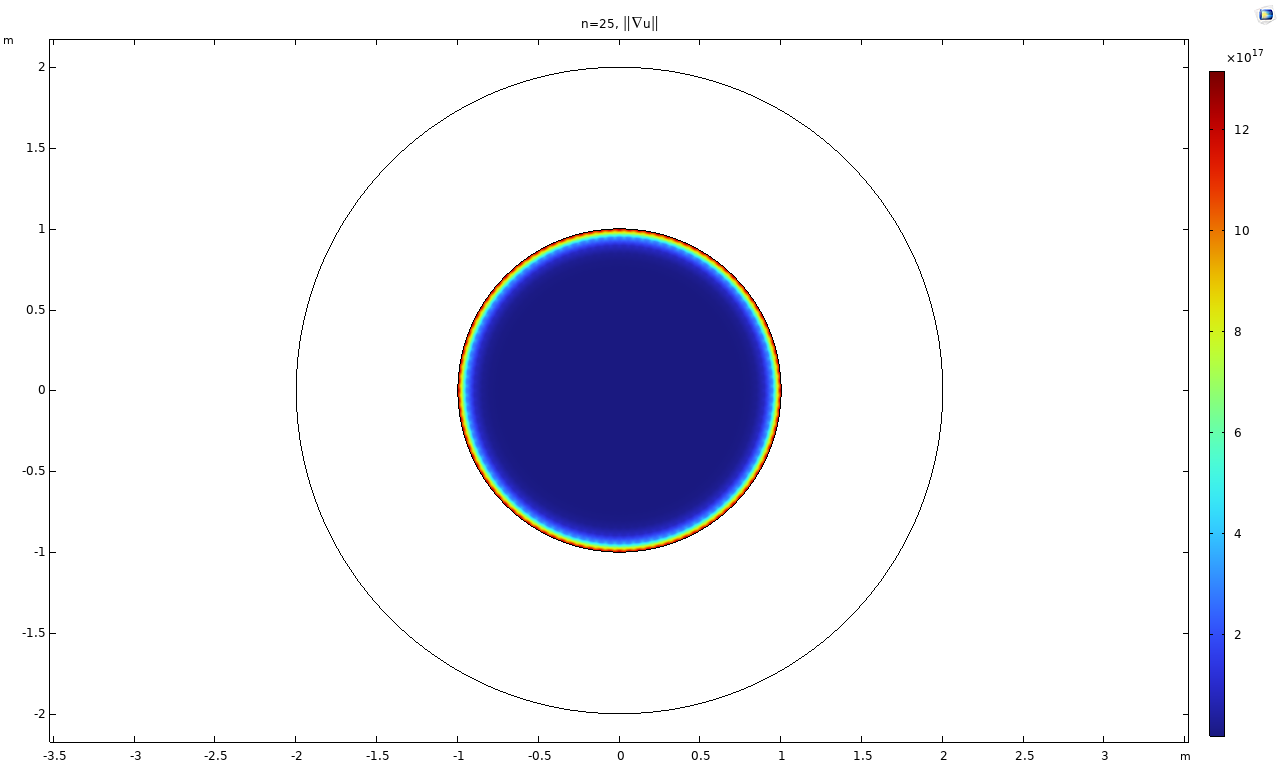}
  		\caption{$n=25$,\ $\| \nabla u\|_{L^2\left(D\right)^2}$}
  	\end{subfigure}
  	\caption{$\| \nabla u \|_{L^2(D)^2}$ for the incident wave $\mathbf{u}^i$ with different indices $n$ ($n=5,15,25$).}
  	\label{fig:3}
  \end{figure}

  \begin{figure}
    \centering
    \subfloat[$n=5,\| \nabla \mathbf{u}^s\|_{L^2\left(B_2 \backslash \overline D\right)^2}$]{\includegraphics[width=0.3\textwidth]{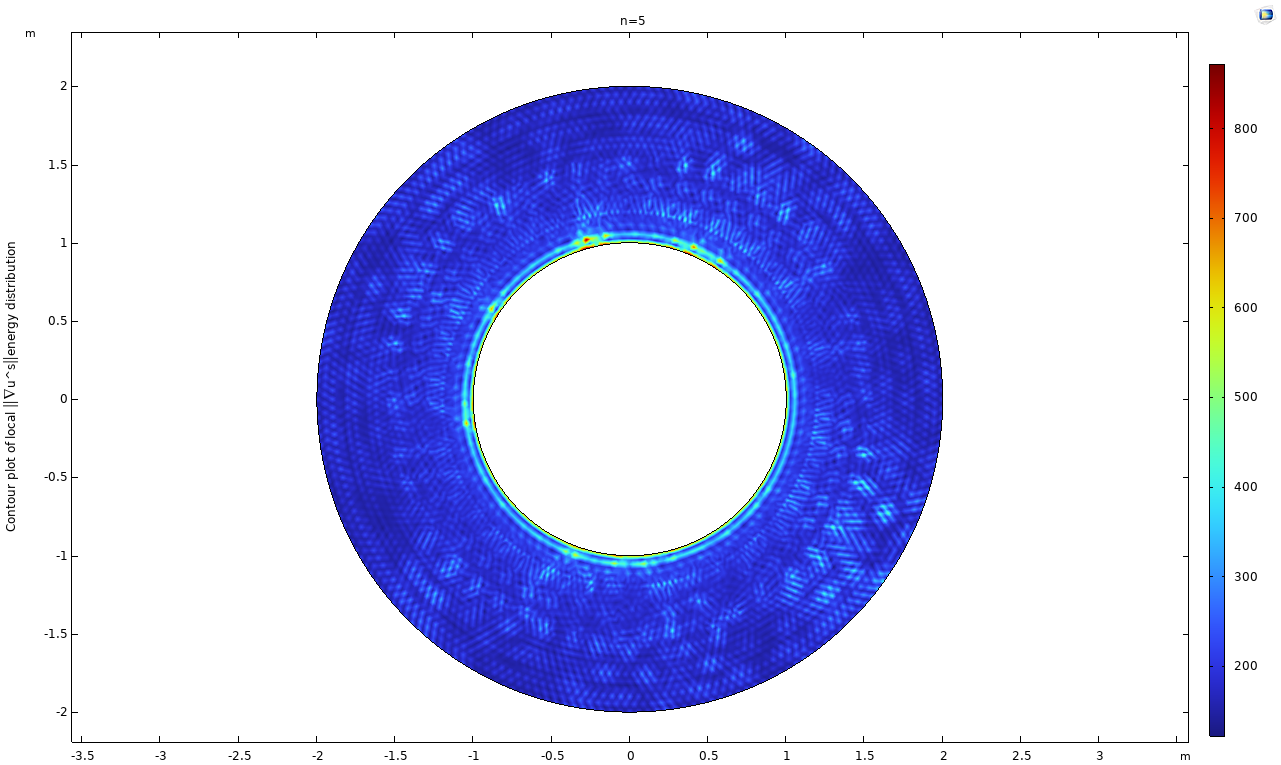}}
    \hfill
    \subfloat[$n=15,\| \nabla \mathbf{u}^s\|_{L^2\left(B_2 \backslash \overline D\right)^2}$]{\includegraphics[width=0.3\textwidth]{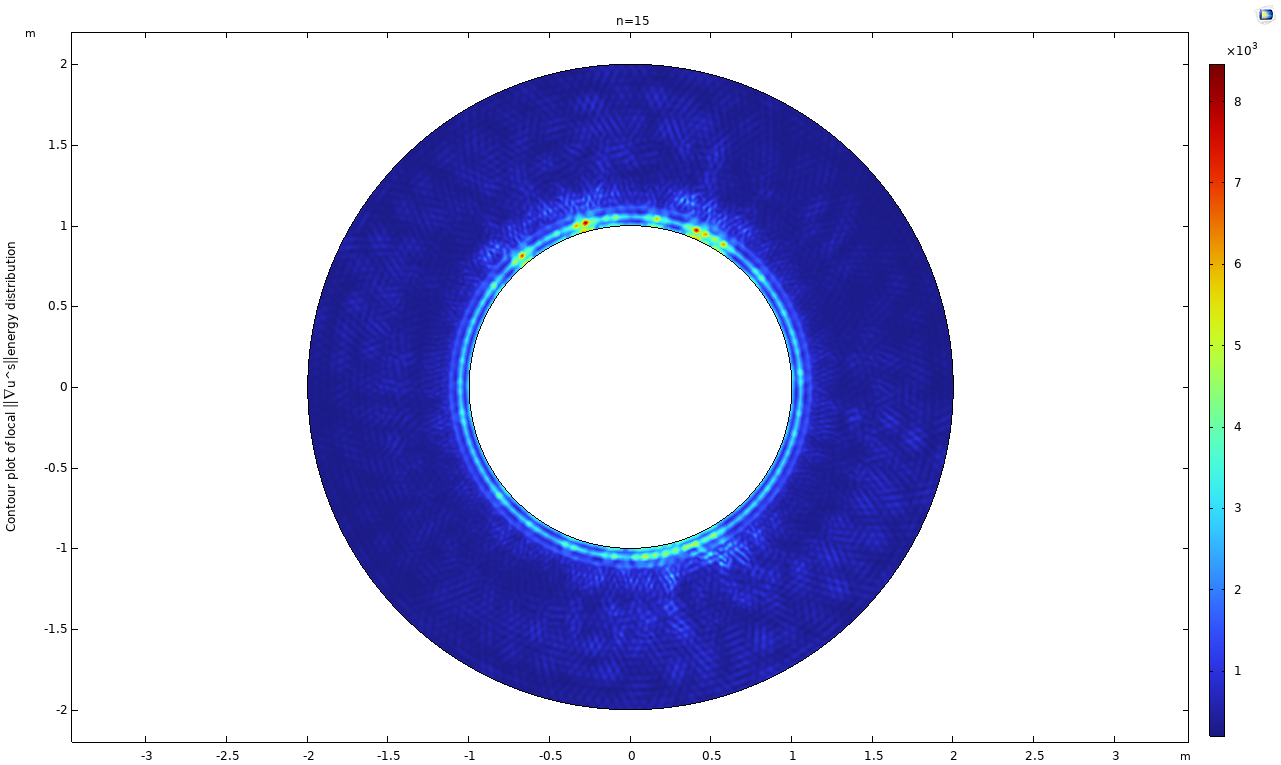}}
    \hfill
    \subfloat[$n=25,\| \nabla \mathbf{u}^s\|_{L^2\left(B_2 \backslash \overline D\right)^2}$]{\includegraphics[width=0.3\textwidth]{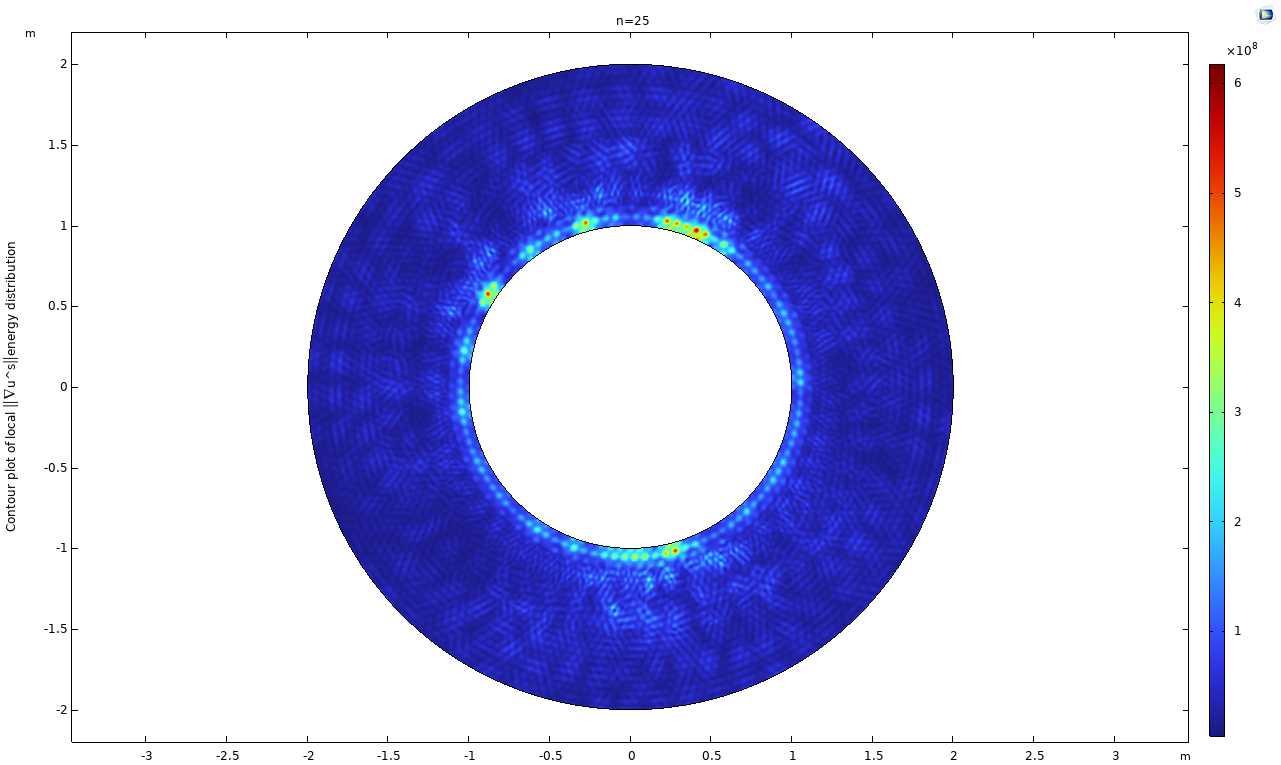}}
    
    \caption{$\| \nabla \mathbf{u}^s\|_{L^2\left(B_2 \backslash \overline D\right)^2}$ for the incident wave $\mathbf{u}^i$ with different indices $n$ ($n=5,15,25$).}
    \label{fig:nus_grid}
\end{figure}

\subsection{The boundary of $ D$ is corner-shaped or apple-shaped curve}

  \begin{exm}

In this example, we consider the case where $\partial D$ is formed by a corner-shaped curve, as parameterized in \eqref{eq:test_domains}. We choose the physical model defined in \eqref{eq:parameter} and the incident wave $\mathbf{u}^i$ defined in \eqref{eq:ui in 2D}. We first show the stress $\mathcal{E}(\mathbf{x})$ for the exterior total field $\mathbf{u}|_{B_2\setminus\overline{D}}$, defined in \eqref{eq:Eu223}, in Figure~\ref{fig:13} for the indices $n=5,15,25$. Meanwhile, the $L^2$-norms of the interior total field $\mathbf{u}|_D$ and the exterior scattered field $\mathbf{u}^s|_{B_2\setminus\overline{D}}$ are presented in Figures~\ref{fig:9} and~\ref{fig:10}, respectively. The values of $\|\nabla u\|_{L^2(D)^2}$ and $\|\nabla \mathbf{u}^s\|_{L^2(B_2\setminus\overline{D})^2}$ are exhibited in Figures~\ref{fig:11} and~\ref{fig:12}, respectively. As the index $n$ of $\mathbf{u}^i$ increases, similar observations and conclusions can be drawn as those presented in Examples \ref{exm:1}--\ref{exm:3} for the unit disk. In particular, it can be directly observed that stress concentration, boundary localization, and surface resonance become stronger near the boundary of the bubble $D$. Additionally, for $n=5$, boundary localization, high oscillation, and stress concentration are observed at the corner point, with the maximum value of $\|\nabla \mathbf{u}\|_{L^2(D)^2}$ in Figure~\ref{fig:11} larger than that in Figure~\ref{fig:3}. For $ n = 25 $, both the interior total field $ \mathbf{u}|_D $ and the exterior scattered field $ \mathbf{u}^s|_{B_2 \setminus \overline{D}} $ exhibit wave field localization, high oscillations, and stress concentrations near the boundary $ \partial D $, particularly at the convex arc vertices on the right and bottom sides. These regions are symmetric with respect to the $ x_1 $-axis.

  \end{exm}
  
  \begin{exm}

In this example, we analyze the case where the bubble $D$ is an apple-shaped domain, as defined in \eqref{eq:test_domains}. For the incident wave $\mathbf{u}^i$, specified in \eqref{eq:ui in 2D} with indices $n = 5, 15, 25$, we fix the frequency $\omega = 0.1$ Hz and adopt the physical model described in \eqref{eq:5.4}. Figure~\ref{fig:18} illustrates the stress $\mathcal{E}(\mathbf{x})$ of the exterior total field $\mathbf{u}$. Additionally, Figures~\ref{fig:14} and \ref{fig:15} present the $L^2$-norms of the internal total field $u|_D$ and the external scattered field $\mathbf{u}^s|_{B_2 \setminus \overline{D}}$, respectively. Figures~\ref{fig:16} and \ref{fig:17} depict the $L^2$-norms of $\nabla u|_D$ and $\nabla \mathbf{u}^s|_{B_2 \setminus \overline{D}}$, respectively.

As the index $n$ of the incident wave $\mathbf{u}^i$ increases, the observations align with those in Examples~\ref{exm:1}--\ref{exm:3} for the unit disk. Notably, stress concentration intensifies in regions of higher curvature, including concave points, for $n = 5$, as shown in Figure~\ref{fig:18}. For larger $n$, the stress increasingly concentrates at the convex corner in the upper-right region of the domain. This behavior supports the applicability of the analysis in Theorem~\ref{thm:Eu definition in thm} to general domains.

Figures~\ref{fig:14} and \ref{fig:15} show that, as $n$ increases, the boundary localization of the interior total field $\mathbf{u}|_D$ and the exterior scattered field $\mathbf{u}^s$ intensifies near the right convex vertex of $D$. Additionally, Figures~\ref{fig:16} and \ref{fig:17} illustrate increasingly oscillatory behavior in $\mathbf{u}|_D$ and $\mathbf{u}^s|_{\mathbb{R}^2 \setminus \overline{D}}$. As $n$ increases, this oscillatory behavior becomes more pronounced near the right convex vertex of $D$, while Figure~\ref{fig:17} also reveals high oscillation at the concave points of the bubble $D$.

\end{exm}

  \begin{figure}
  	\centering
  	\begin{subfigure}[t]{0.3\textwidth}
  		\centering
  		\includegraphics[width=\textwidth]{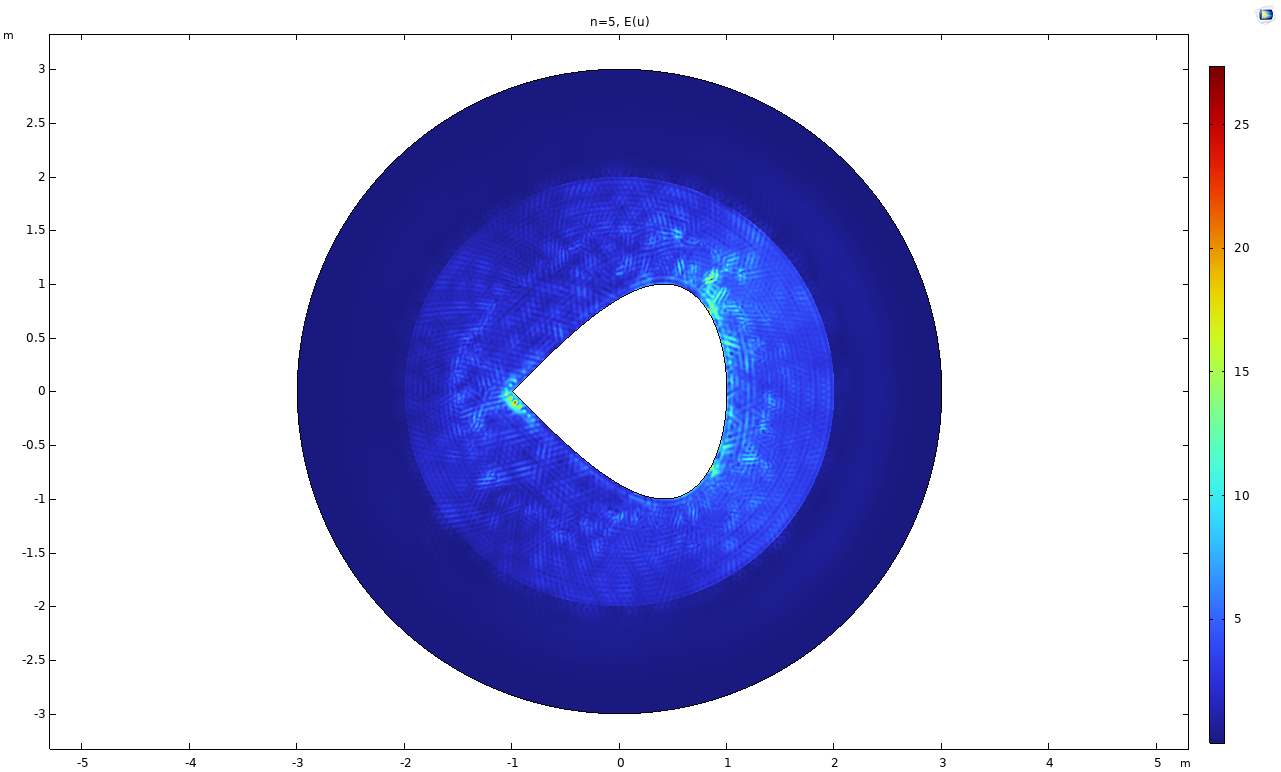} 
  		\caption{$n=5$,\ ${\mathcal E}(\mathbf{u})$}
  	\end{subfigure}
  	\hfill
  	\begin{subfigure}[t]{0.3\textwidth}
  		\centering
  		\includegraphics[width=\textwidth]{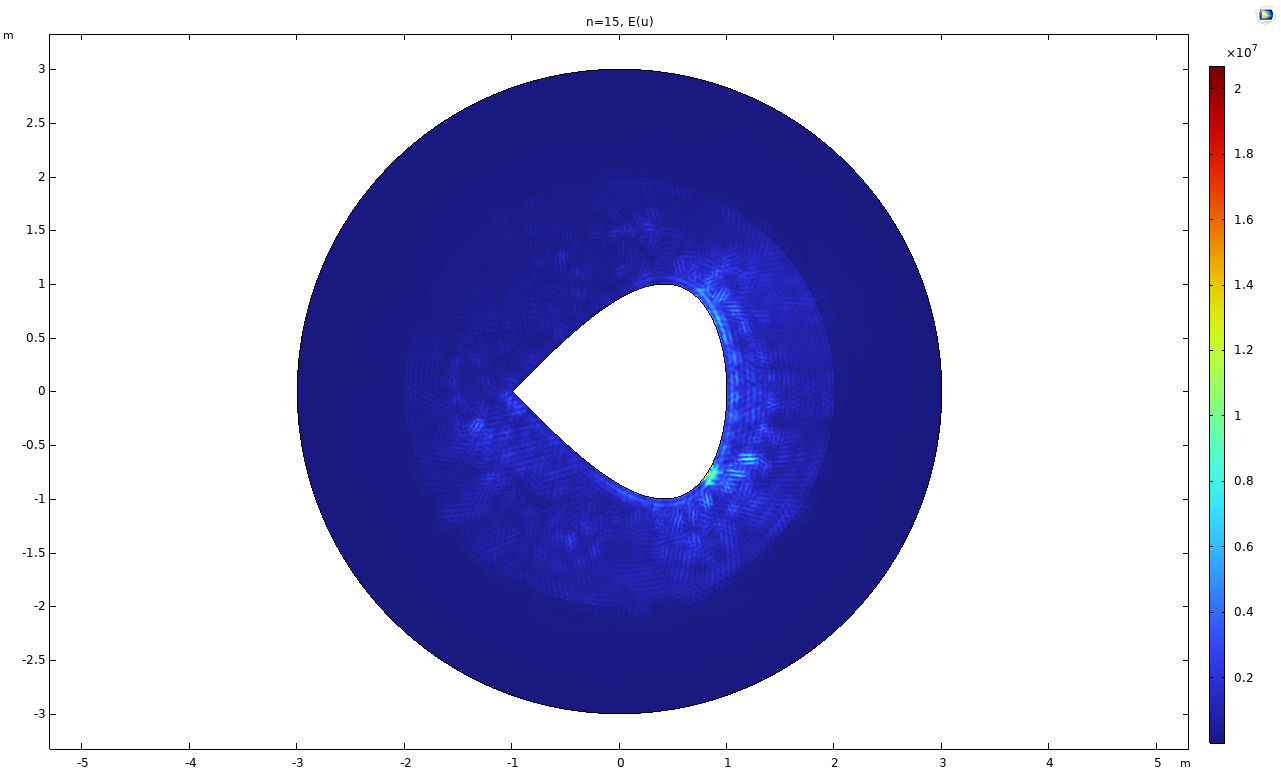}
  		\caption{$n=15$,\ ${\mathcal E}(\mathbf{u})$}
  	\end{subfigure}%
  	\hfill%
  	\begin{subfigure}[t]{0.3\textwidth}
  		\centering
  		\includegraphics[width=\textwidth]{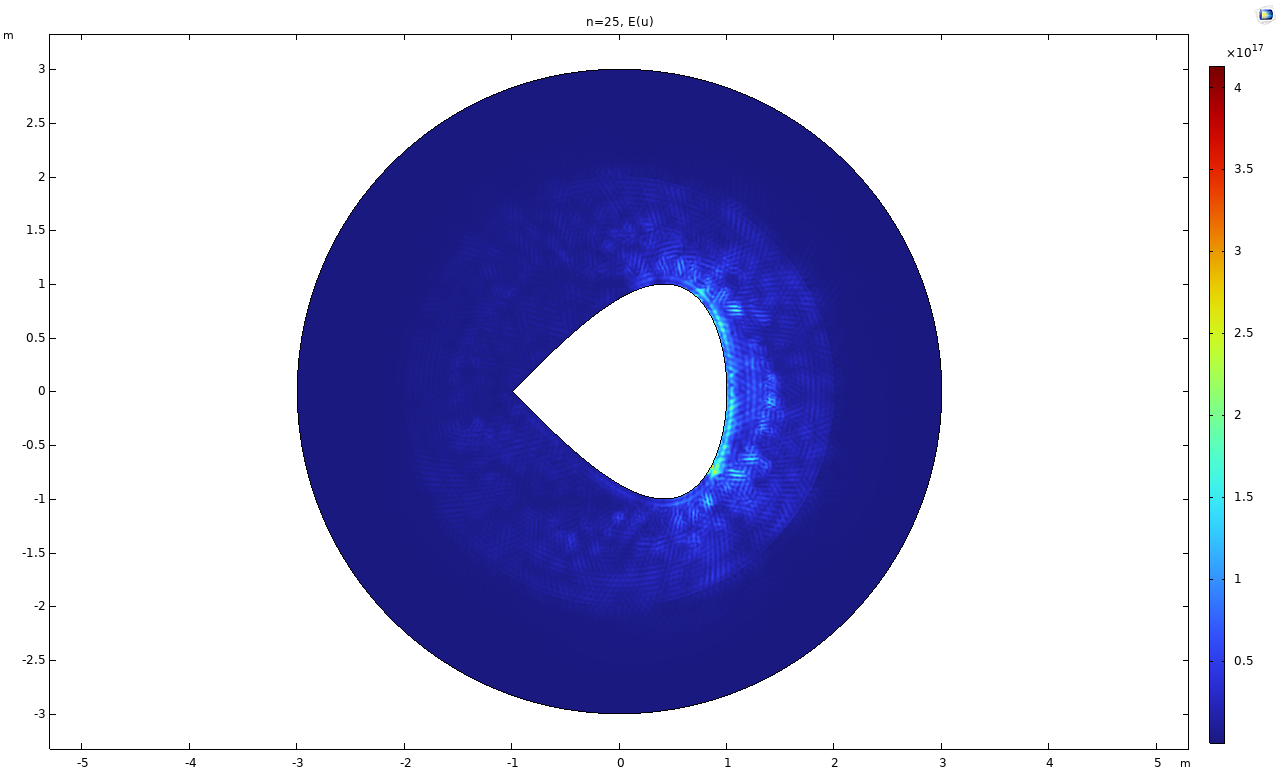}
  		\caption{$n=25$,\ ${\mathcal E}(\mathbf{u})$}
  	\end{subfigure}
  	\caption{ The stress $\mathcal{E}(\mathbf{u})$ of the exterior total field for the incident wave $\mathbf{u}^i$ with indices $n$ ($n=5, 15, 25$).}
  	\label{fig:13}
  \end{figure}

  \begin{figure}
  	\centering
  	\begin{subfigure}[t]{0.3\textwidth}
  		\centering
  		\includegraphics[width=\textwidth]{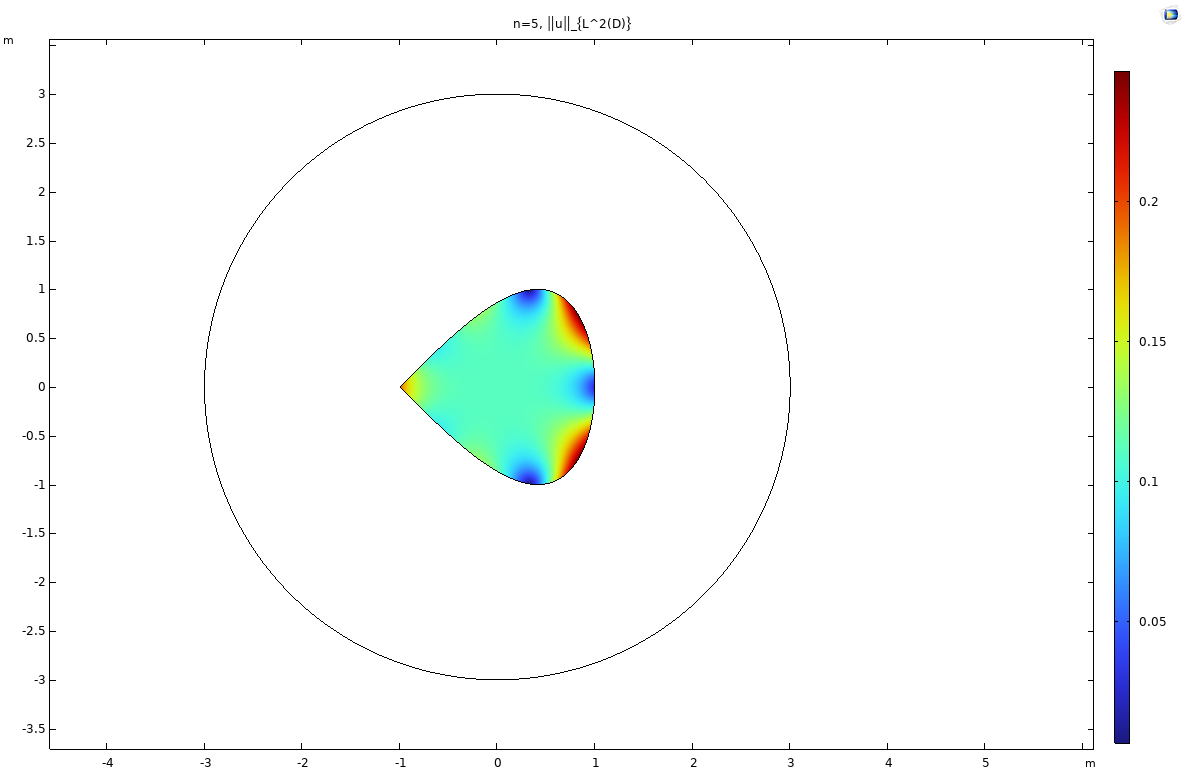} 
  		\caption{$n=5$,\ $\|u\|_{L^2(D)^2}$}
  	\end{subfigure}
  	\hfill
  	\begin{subfigure}[t]{0.3\textwidth}
  		\centering
  		\includegraphics[width=\textwidth]{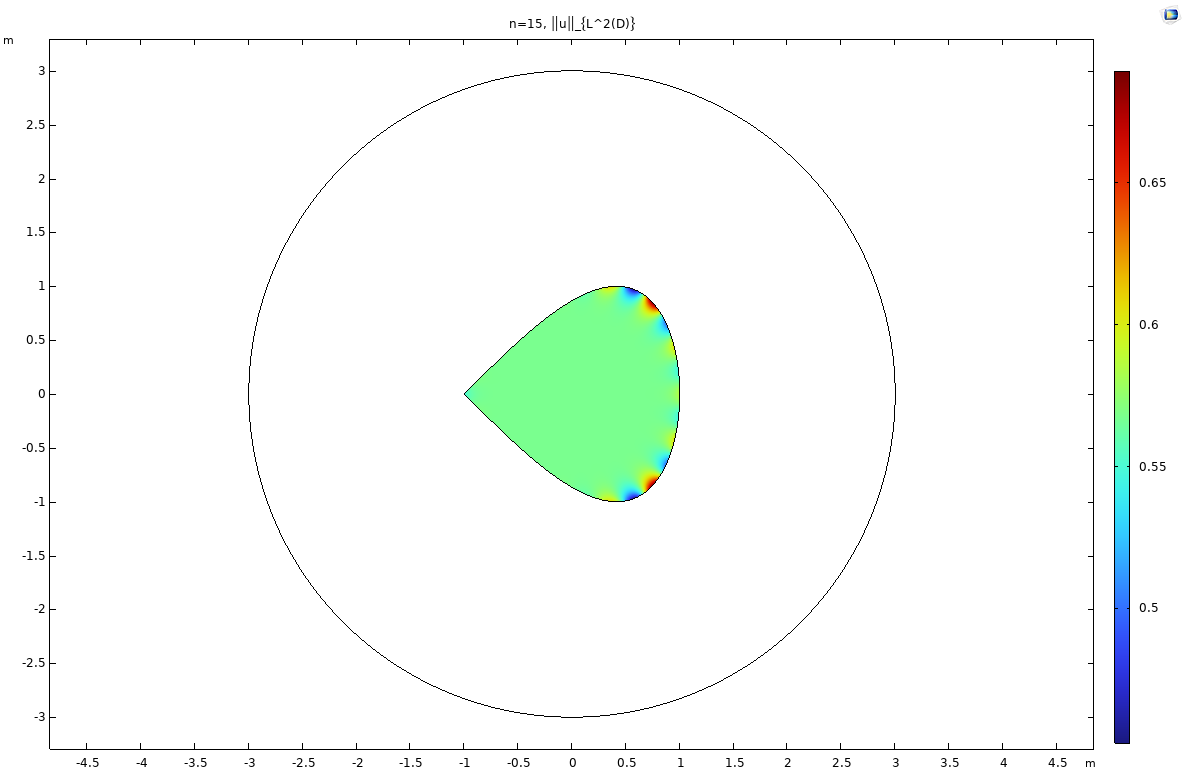}
  		\caption{$n=15$,\ $\|u\|_{L^2(D)^2}$}
  	\end{subfigure}%
  	\hfill%
  	\begin{subfigure}[t]{0.3\textwidth}
  		\centering
  		\includegraphics[width=\textwidth]{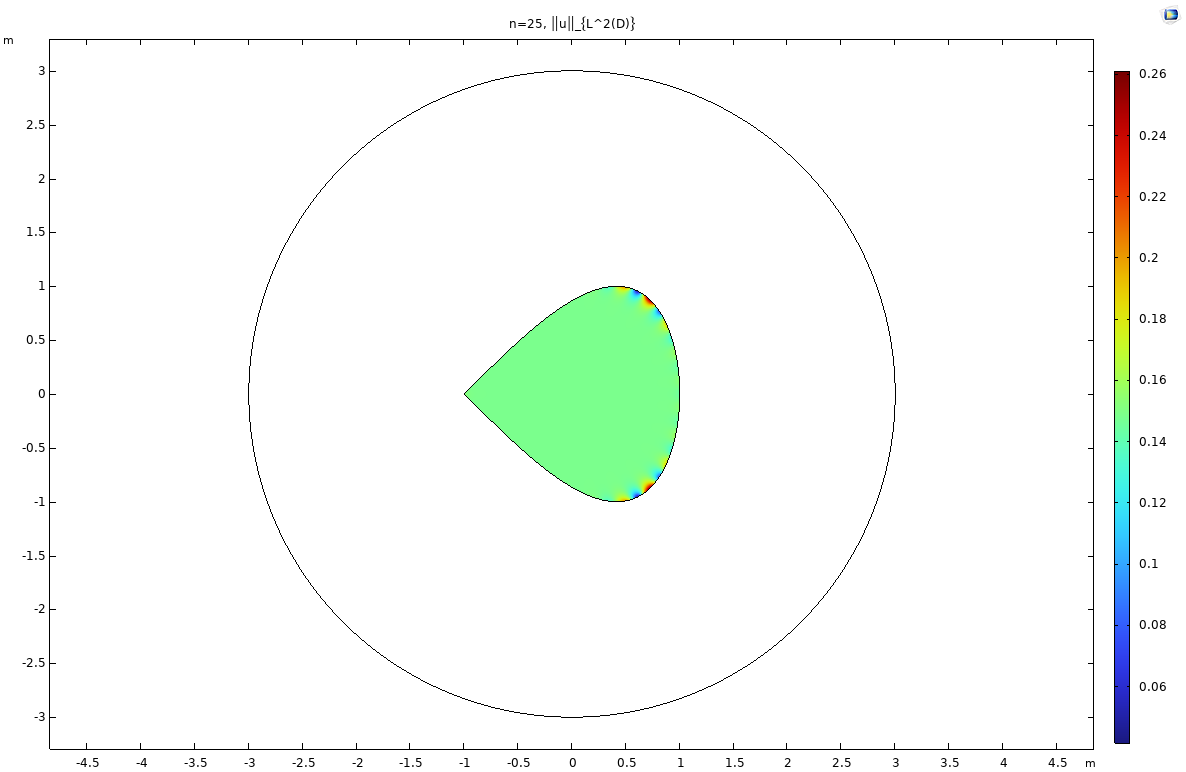}
  		\caption{$n=25$,\ $\|u\|_{L^2(D)^2}$}
  	\end{subfigure}
  	\caption{$\|u\|_{L^2(D)}$ for the incident wave $\mathbf{u}^i$ with different indices $n$ ($n=5,15,25$).}
  	\label{fig:9}
  \end{figure}

  \begin{figure}
  	\centering
  	\begin{subfigure}[t]{0.3\textwidth}
  		\centering
  		\includegraphics[width=\textwidth]{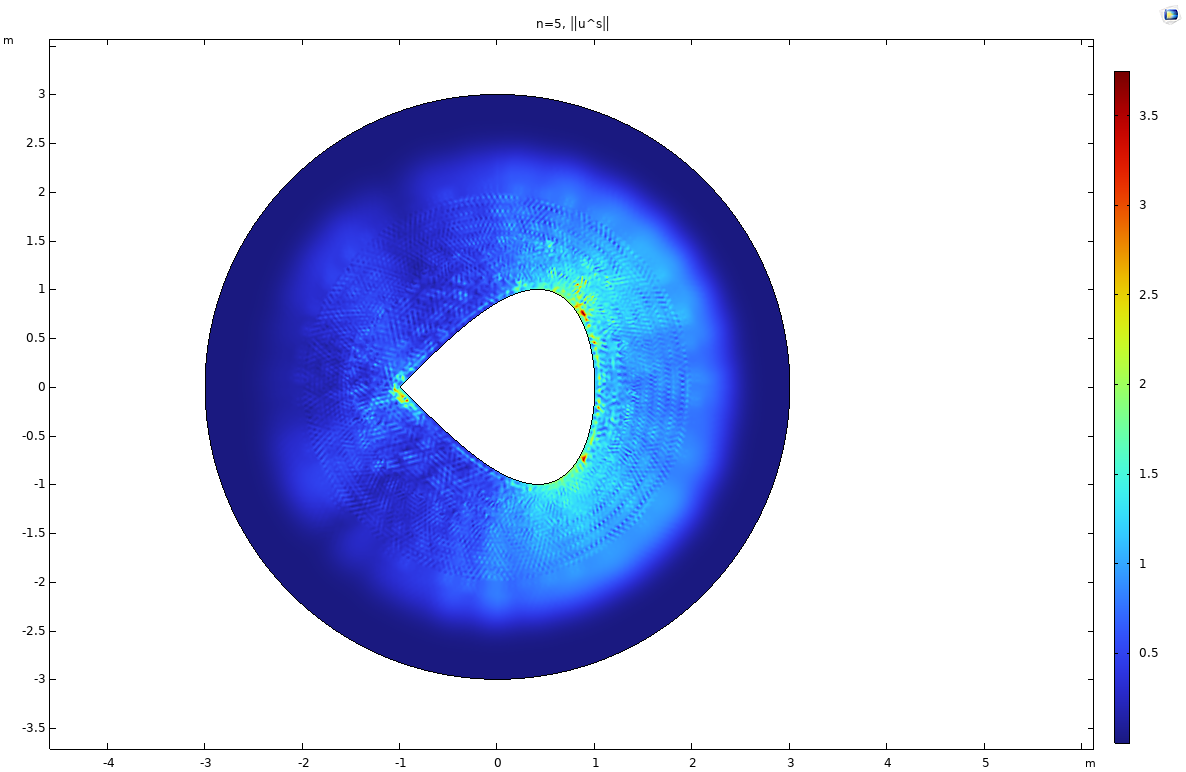} 
  		\caption{$n=5$,\ $\|\mathbf{u}^s\|_{L^2\left(B_2 \backslash \overline D\right)^2}$}
  	\end{subfigure}
  	\hfill
  	\begin{subfigure}[t]{0.3\textwidth}
  		\centering
  		\includegraphics[width=\textwidth]{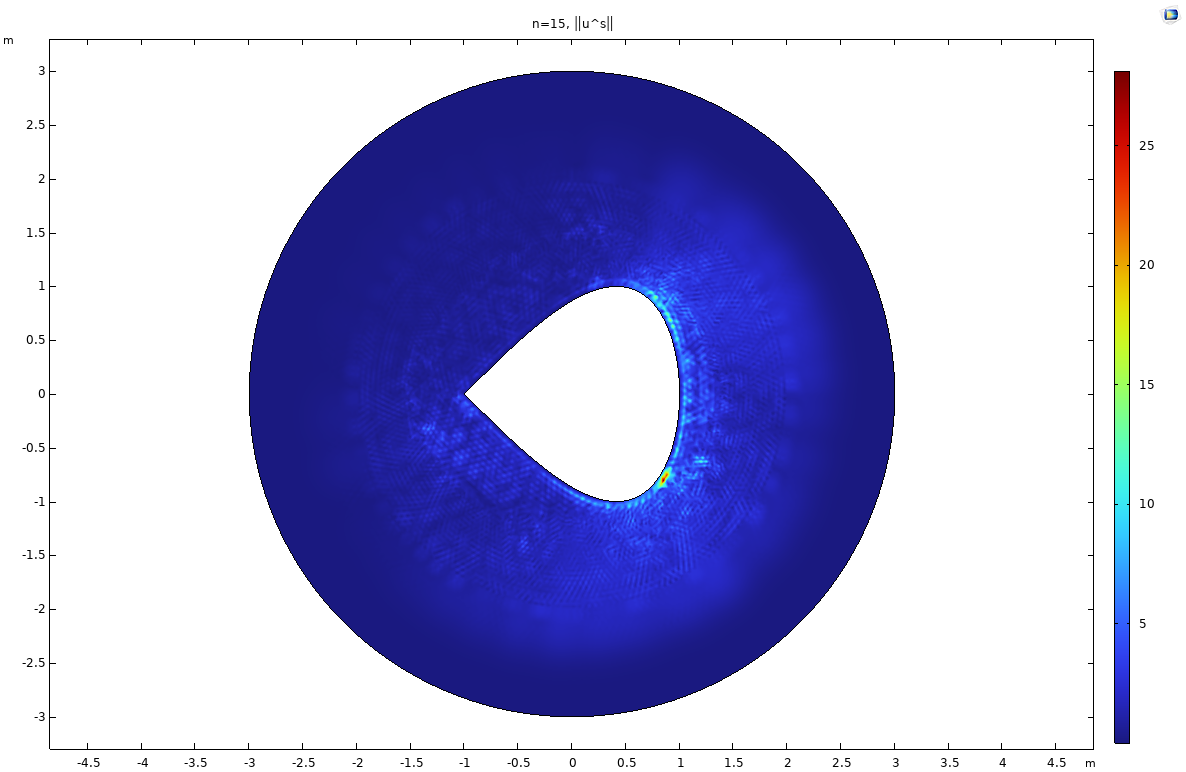}
  		\caption{$n=15$,\ $\|\mathbf{u}^s\|_{L^2\left(B_2 \backslash \overline D\right)^2}$}
  	\end{subfigure}%
  	\hfill%
  	\begin{subfigure}[t]{0.3\textwidth}
  		\centering
  		\includegraphics[width=\textwidth]{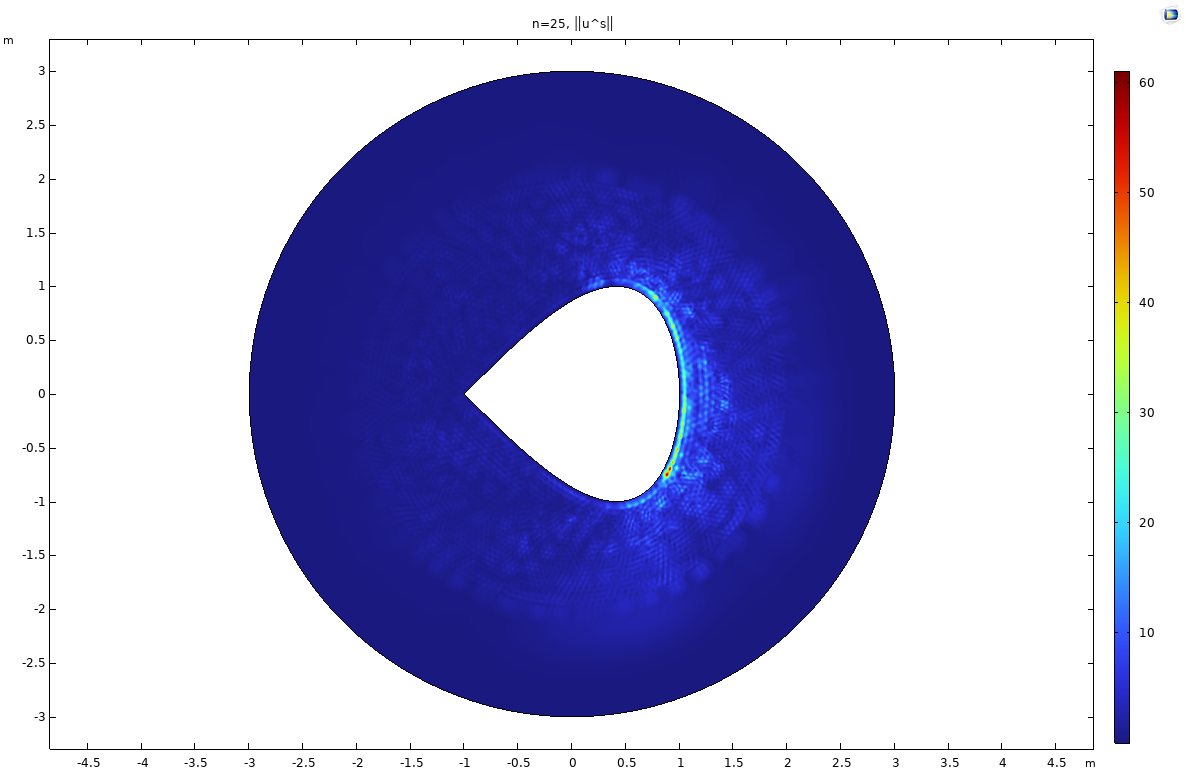}
  		\caption{$n=25$,\ $\|\mathbf{u}^s\|_{L^2\left(B_2 \backslash \overline D\right)^2}$}
  	\end{subfigure}
  	\caption{ $\|\mathbf{u}^s\|_{L^2\left(B_2 \backslash \overline D\right)^2}$ for the incident wave $\mathbf{u}^i$ with different indices $n$ ($n=5,15,25$).}
  	\label{fig:10}
  \end{figure}

   \begin{figure}
   	\centering
   	\begin{subfigure}[t]{0.3\textwidth}
   		\centering
   		\includegraphics[width=\textwidth]{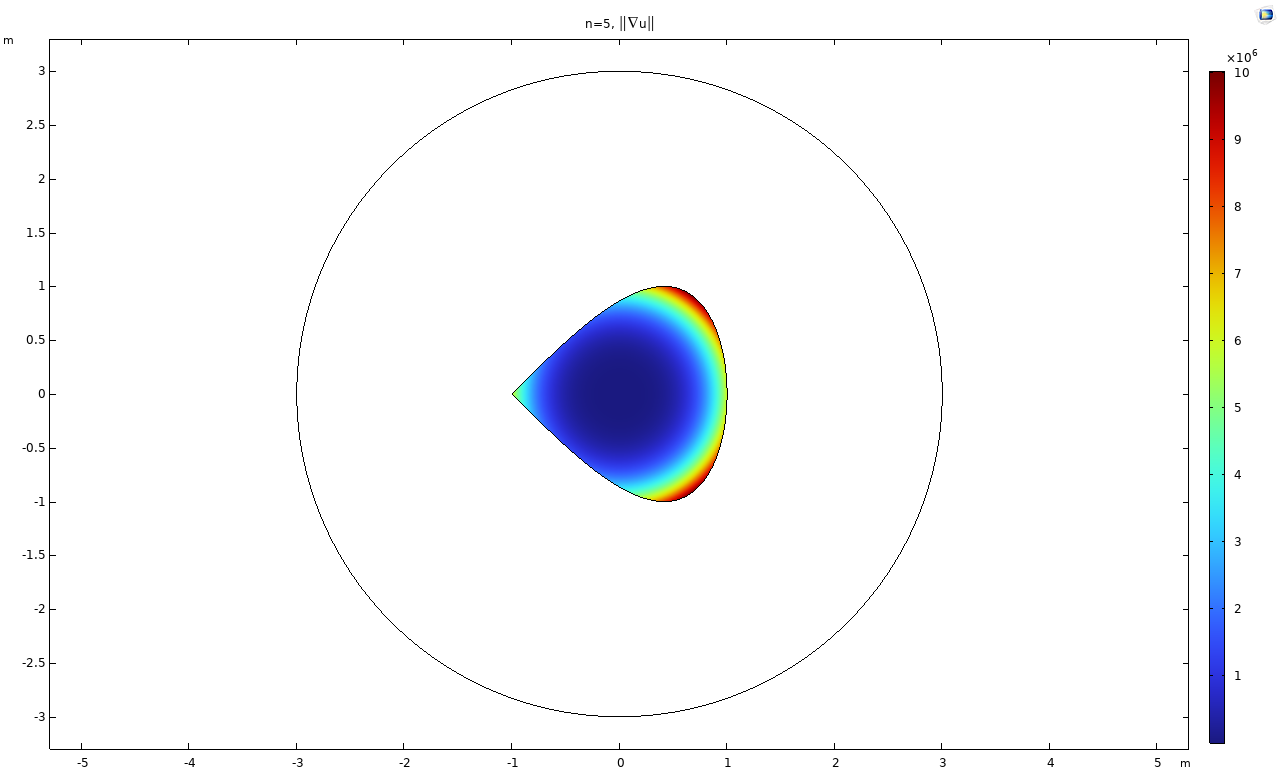} 
   		\caption{$n=5$,\ $\| \nabla u\|_{L^2\left(D\right)^2}$}
   	\end{subfigure}
   	\hfill
   	\begin{subfigure}[t]{0.3\textwidth}
   		\centering
   		\includegraphics[width=\textwidth]{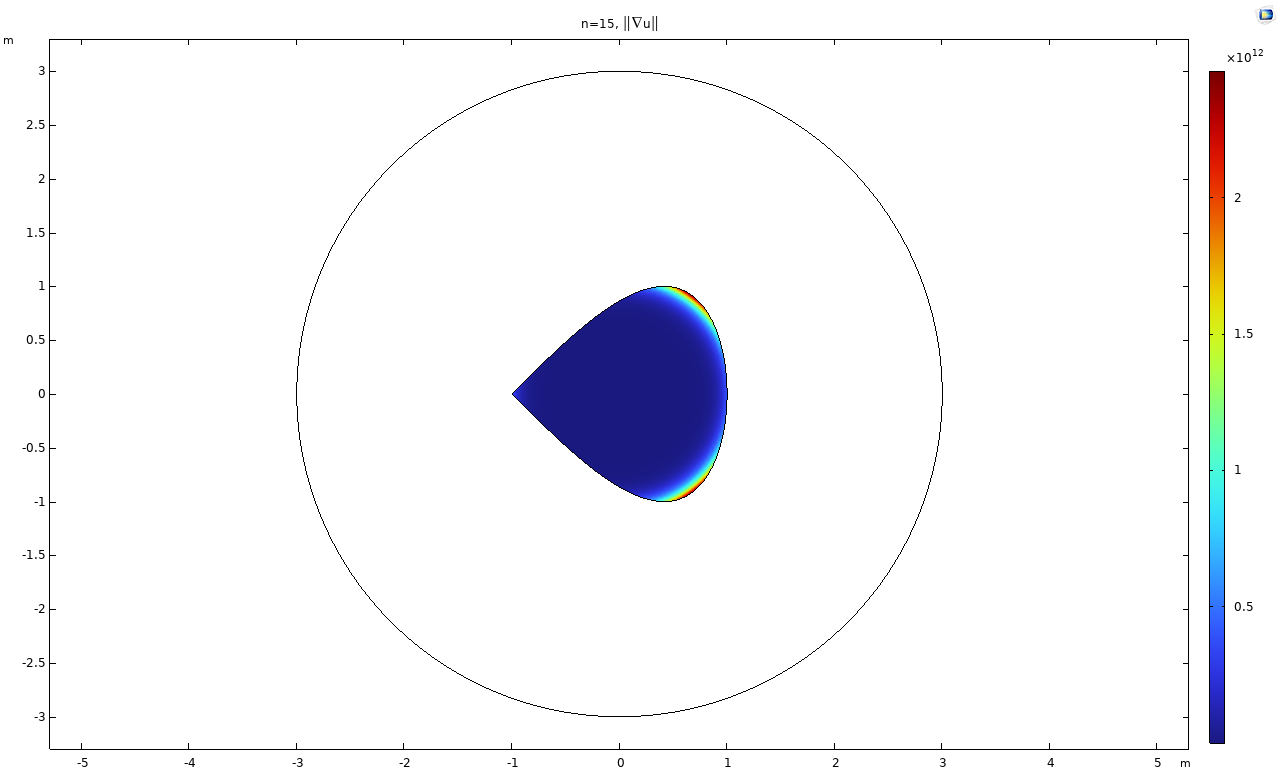}
   		\caption{$n=15$,\ $\| \nabla u\|_{L^2\left(D\right)^2}$}
   	\end{subfigure}%
   	\hfill%
   	\begin{subfigure}[t]{0.3\textwidth}
   		\centering
   		\includegraphics[width=\textwidth]{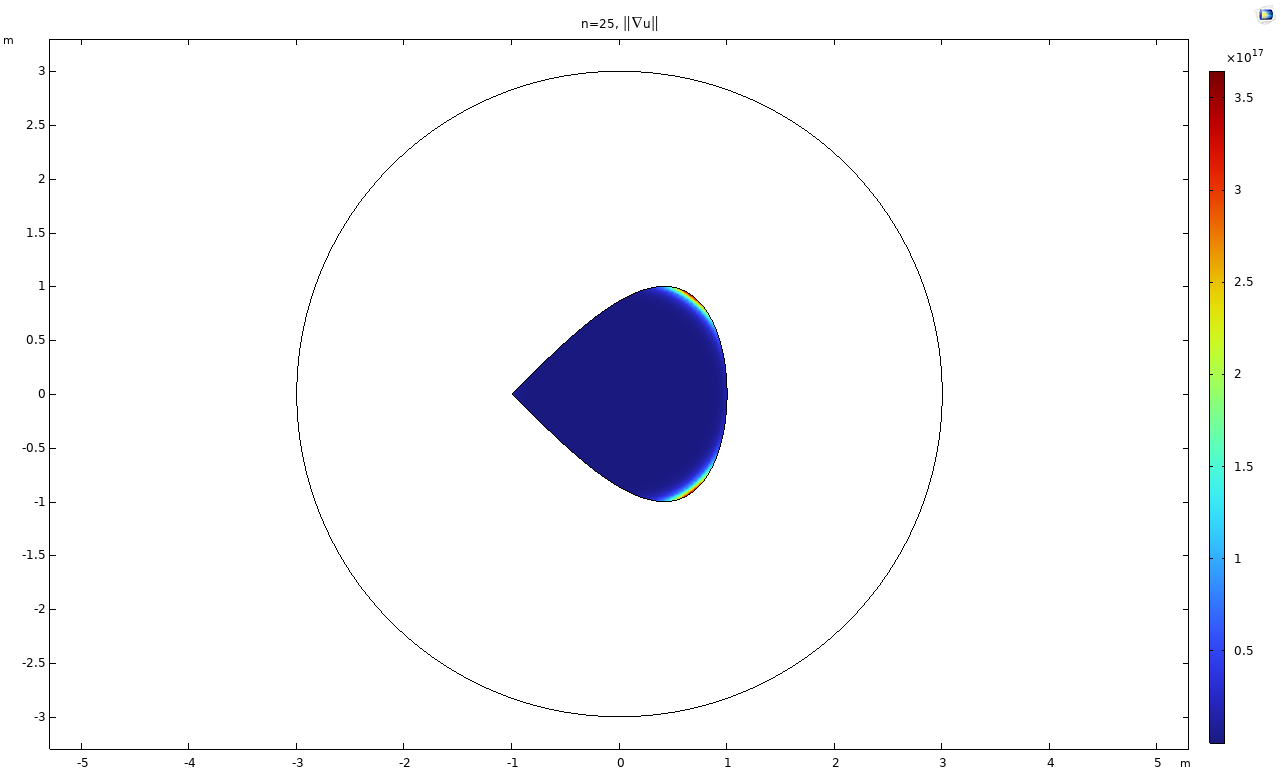}
   		\caption{$n=25$,\ $\| \nabla u\|_{L^2\left(D\right)^2}$}
   	\end{subfigure}
   	\caption{$\| \nabla u \|_{L^2(D)^2}$ for the incident wave $\mathbf{u}^i$ with different indices $n$ ($n=5,15,25$).}
   	\label{fig:11}
   \end{figure}

  \begin{figure}
  	\centering
  	\begin{subfigure}[t]{0.3\textwidth}
  		\centering
  		\includegraphics[width=\textwidth]{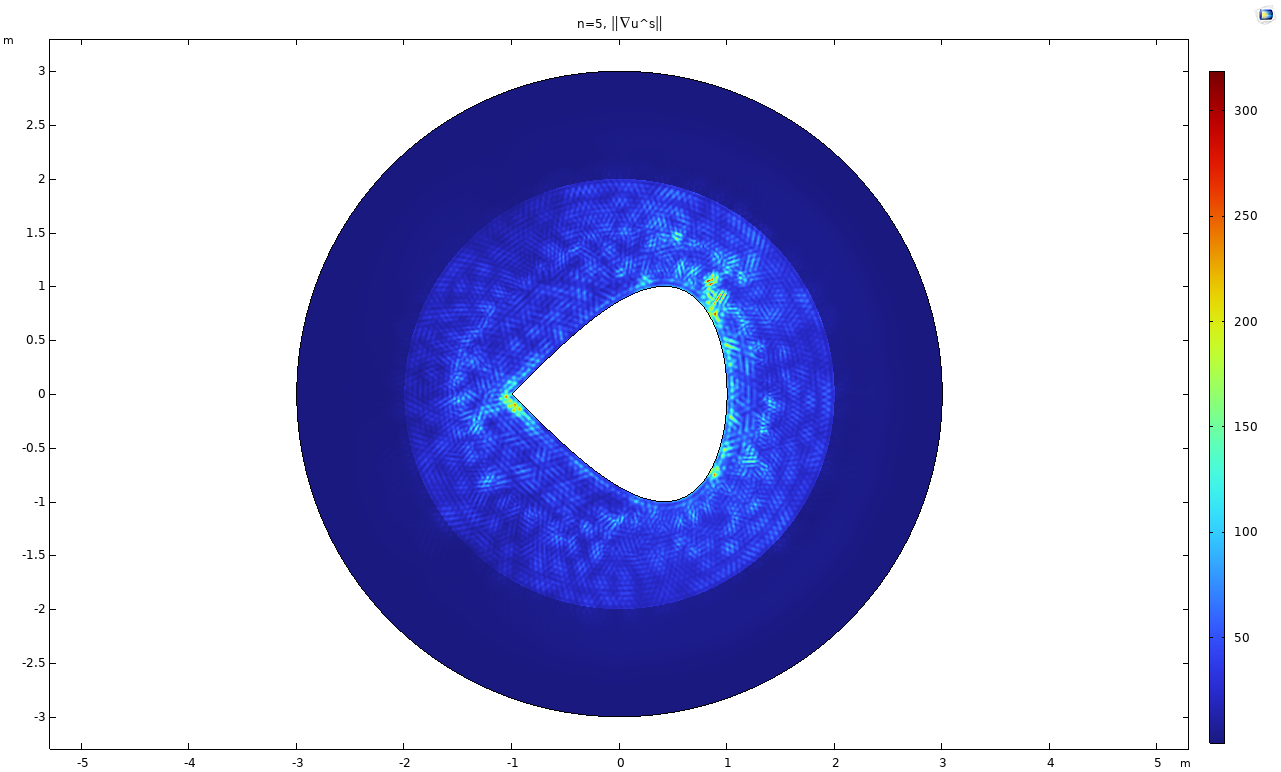} 
  		\caption{$n=5$,\ $\| \nabla \mathbf{u}^s\|_{L^2\left(B_2 \backslash \overline D\right)^2}$}
  	\end{subfigure}
  	\hfill
  	\begin{subfigure}[t]{0.3\textwidth}
  		\centering
  		\includegraphics[width=\textwidth]{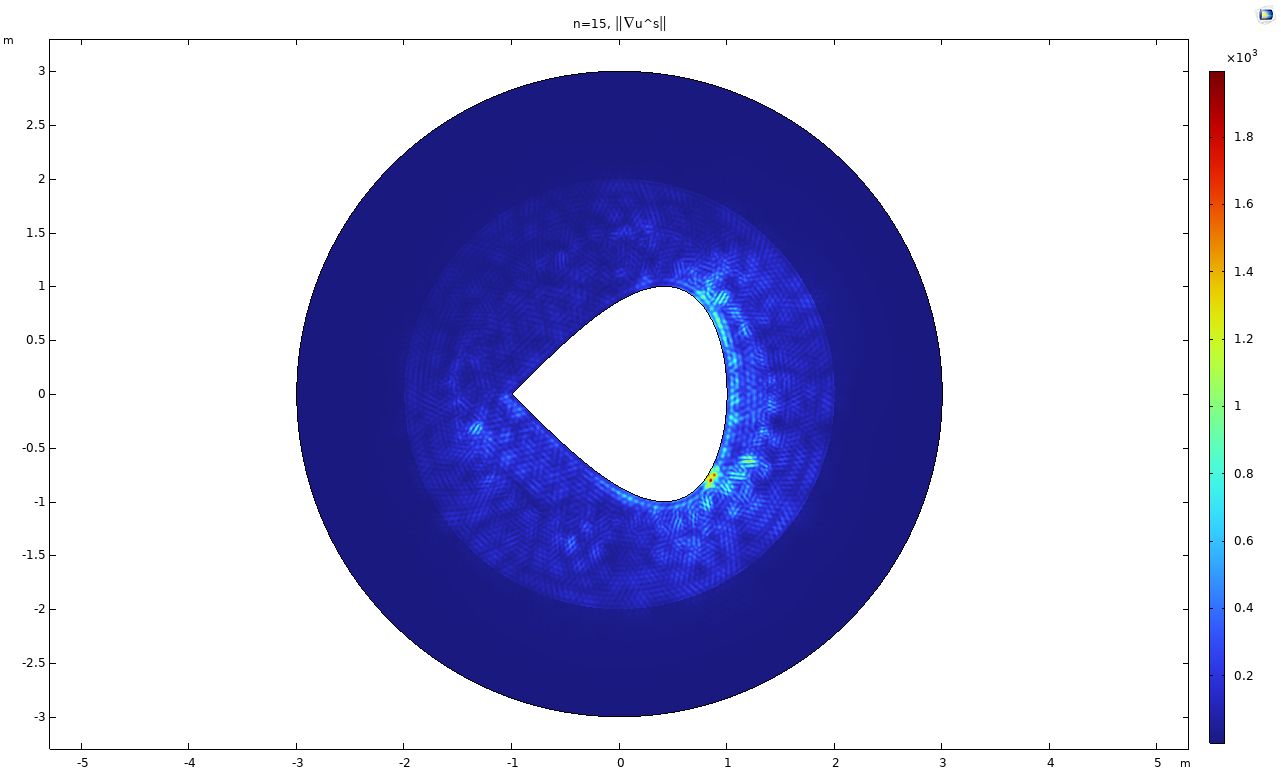}
  		\caption{$n=15$,\ $\| \nabla \mathbf{u}^s\|_{L^2\left(B_2 \backslash \overline D\right)^2}$}
  	\end{subfigure}%
  	\hfill%
  	\begin{subfigure}[t]{0.3\textwidth}
  		\centering
  		\includegraphics[width=\textwidth]{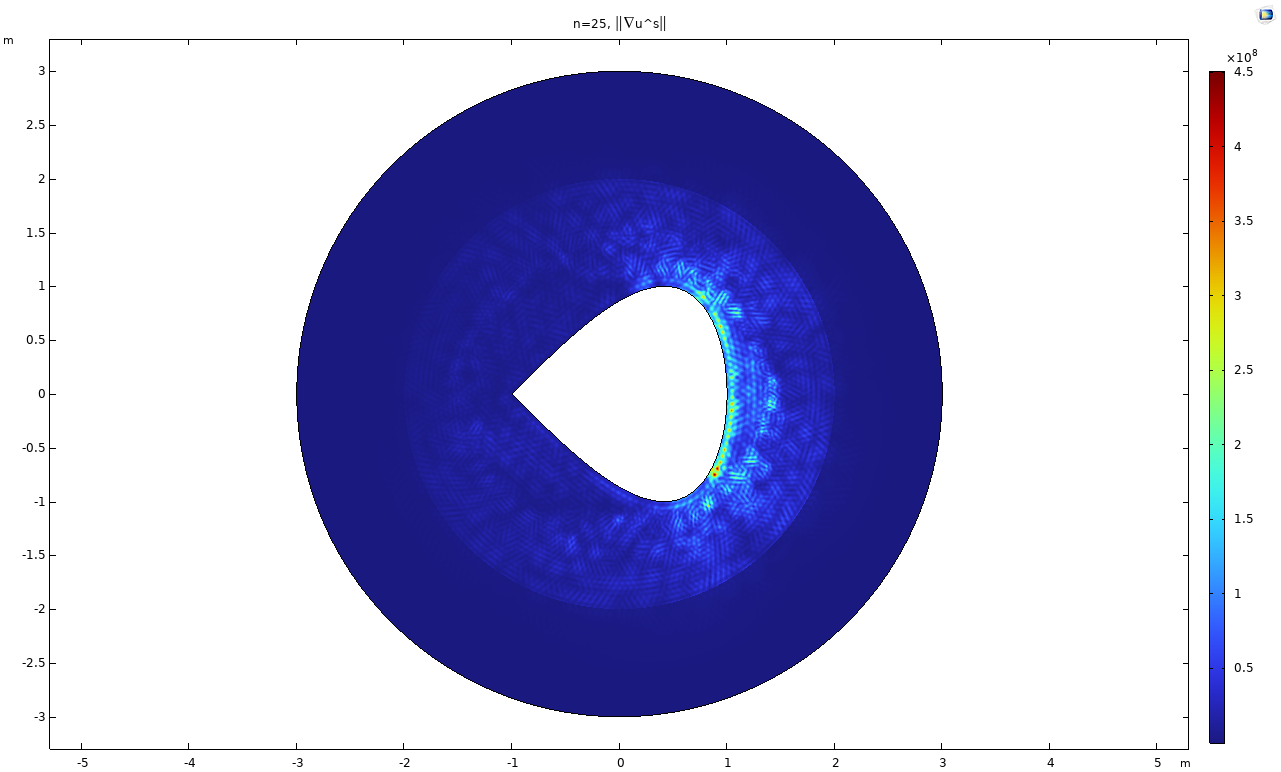}
  		\caption{$n=25$,\ $\| \nabla \mathbf{u}^s\|_{L^2\left(B_2 \backslash \overline D\right)^2}$}
  	\end{subfigure}
  	\caption{$\| \nabla \mathbf{u}^s\|_{L^2\left(B_2 \backslash \overline D\right)^2}$ for the incident wave $\mathbf{u}^i$ with different indices $n$ ($n=5,15,25$).}
  	\label{fig:12}
  \end{figure}


 \begin{figure}
 	\centering
 	\begin{subfigure}[t]{0.3\textwidth}
 		\centering
 		\includegraphics[width=\textwidth]{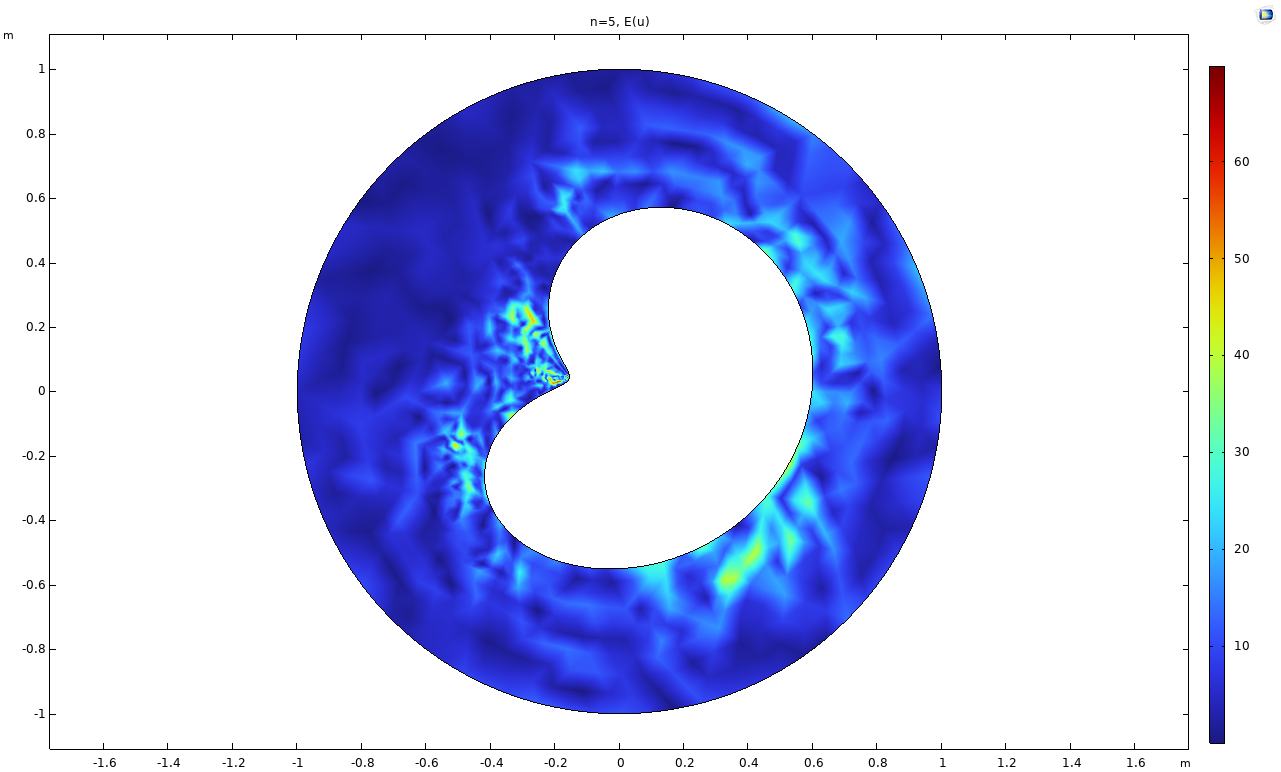} 
 		\caption{$n=5$,\ ${\mathcal E}(\mathbf{u})$}
 	\end{subfigure}
 	\hfill
 	\begin{subfigure}[t]{0.3\textwidth}
 		\centering
 		\includegraphics[width=\textwidth]{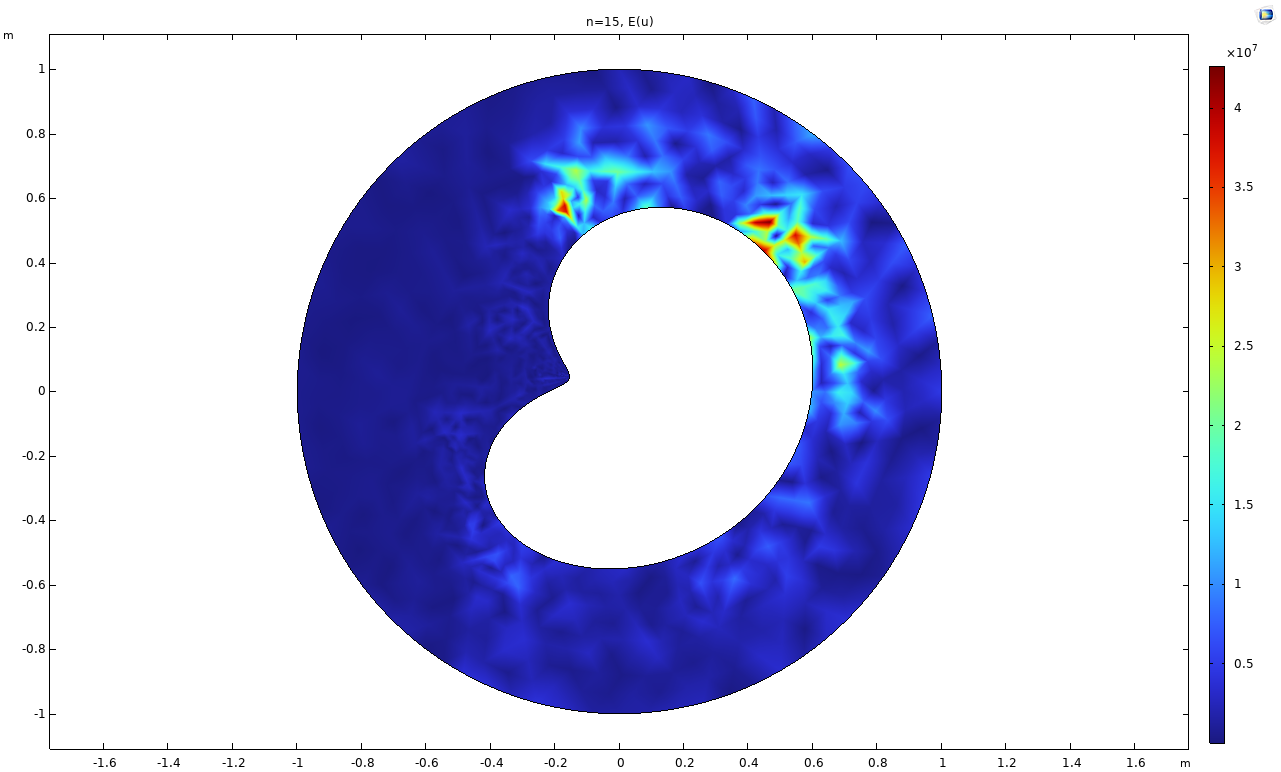}
 		\caption{$n=15$,\ ${\mathcal E}(\mathbf{u})$}
 	\end{subfigure}%
 	\hfill%
 	\begin{subfigure}[t]{0.3\textwidth}
 		\centering
 		\includegraphics[width=\textwidth]{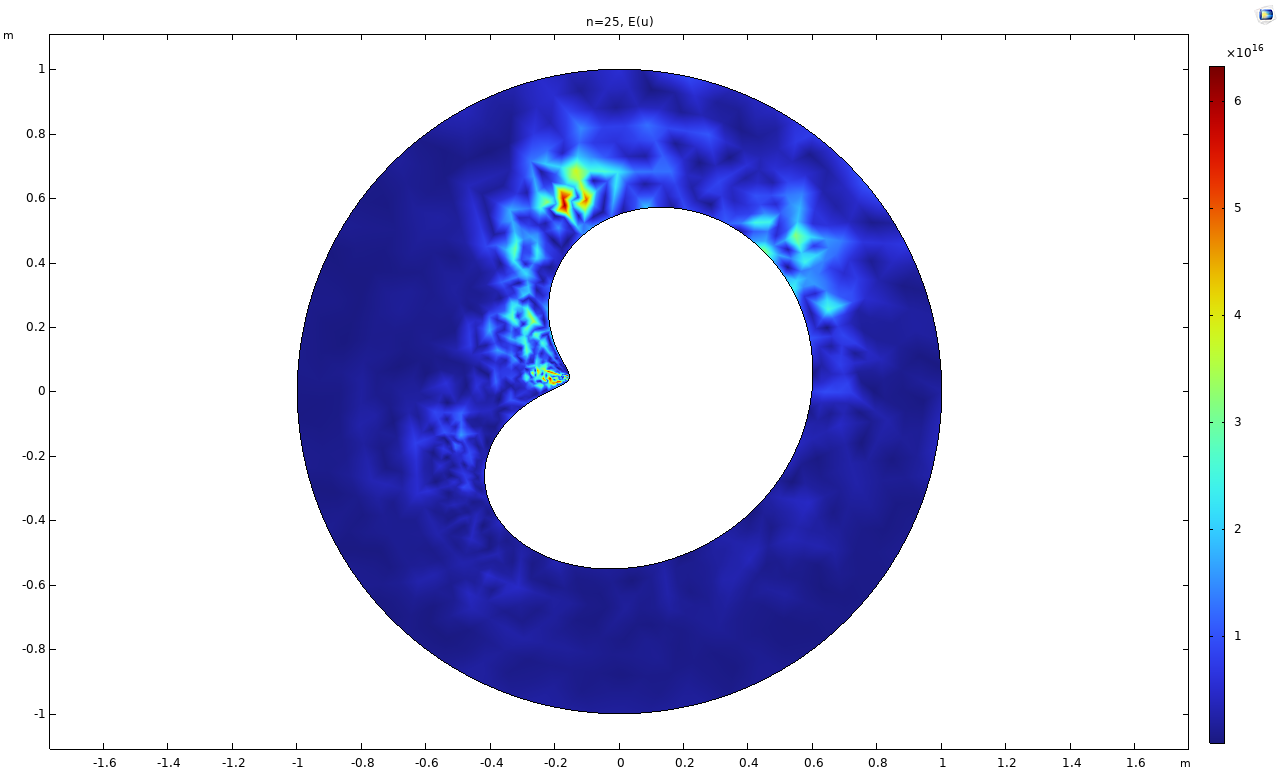}
 		\caption{$n=25$,\ ${\mathcal E}(\mathbf{u})$}
 	\end{subfigure}
 	\caption{ The stress $\mathcal{E}(\mathbf{u})$ of the exterior total field for the incident wave $\mathbf{u}^i$ with indices $n$ ($n=5, 15, 25$).}
 	\label{fig:18}
 \end{figure}

\begin{figure}
	\centering
	\begin{subfigure}[t]{0.3\textwidth}
		\centering
		\includegraphics[width=\textwidth]{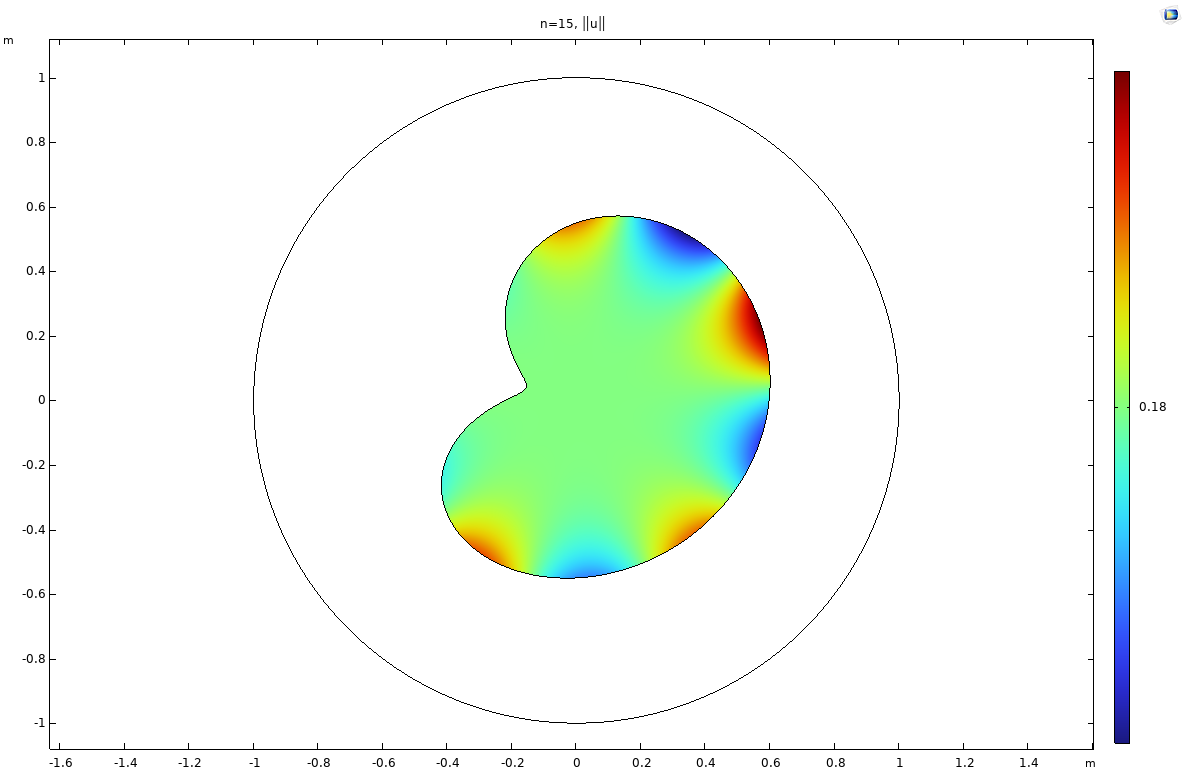} 
		\caption{$n=5$,\ $\|u\|_{L^2(D)^2}$}
	\end{subfigure}
	\hfill
	\begin{subfigure}[t]{0.3\textwidth}
		\centering
		\includegraphics[width=\textwidth]{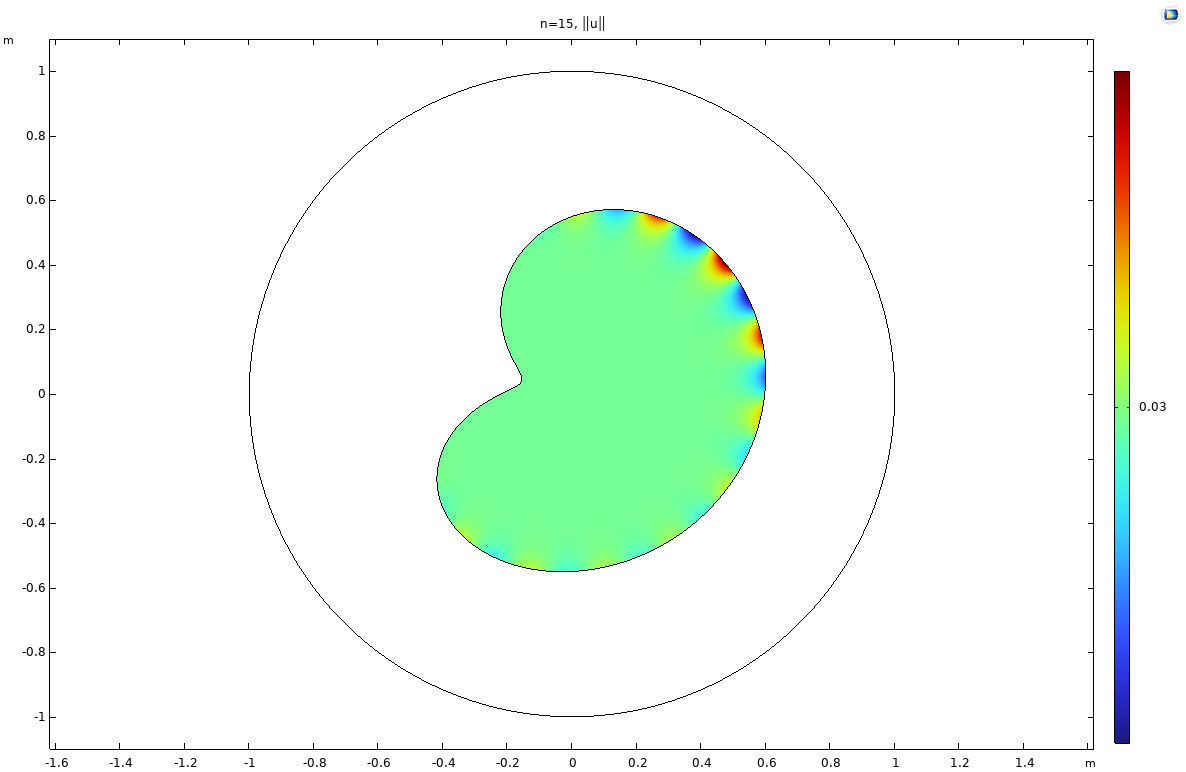}
		\caption{$n=15$,\ $\|u\|_{L^2(D)^2}$}
	\end{subfigure}%
	\hfill%
	\begin{subfigure}[t]{0.3\textwidth}
		\centering
		\includegraphics[width=\textwidth]{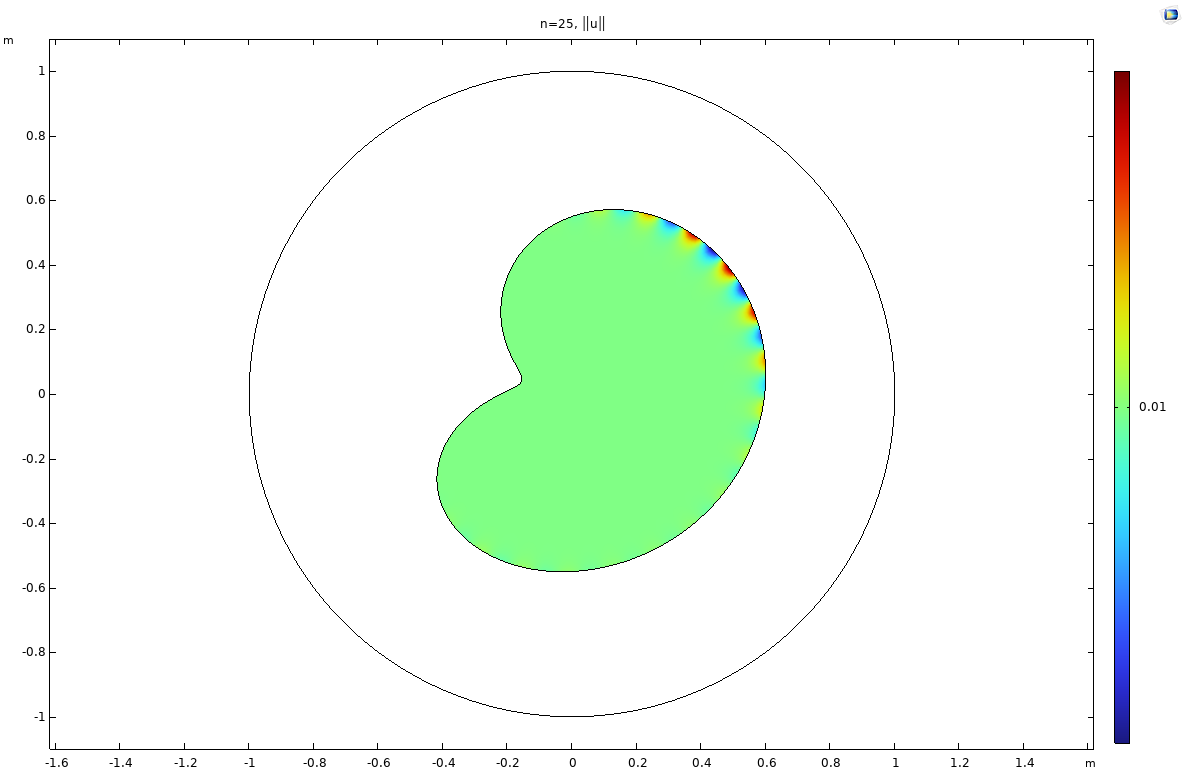}
		\caption{$n=25$,\ $\|u\|_{L^2(D)^2}$}
	\end{subfigure}
	\caption{$\|u\|_{L^2(D)}$ for the incident wave $\mathbf{u}^i$ with different indices $n$ ($n=5,15,25$).}
	\label{fig:14}
\end{figure}

\begin{figure}
	\centering
	\begin{subfigure}[t]{0.3\textwidth}
		\centering
		\includegraphics[width=\textwidth]{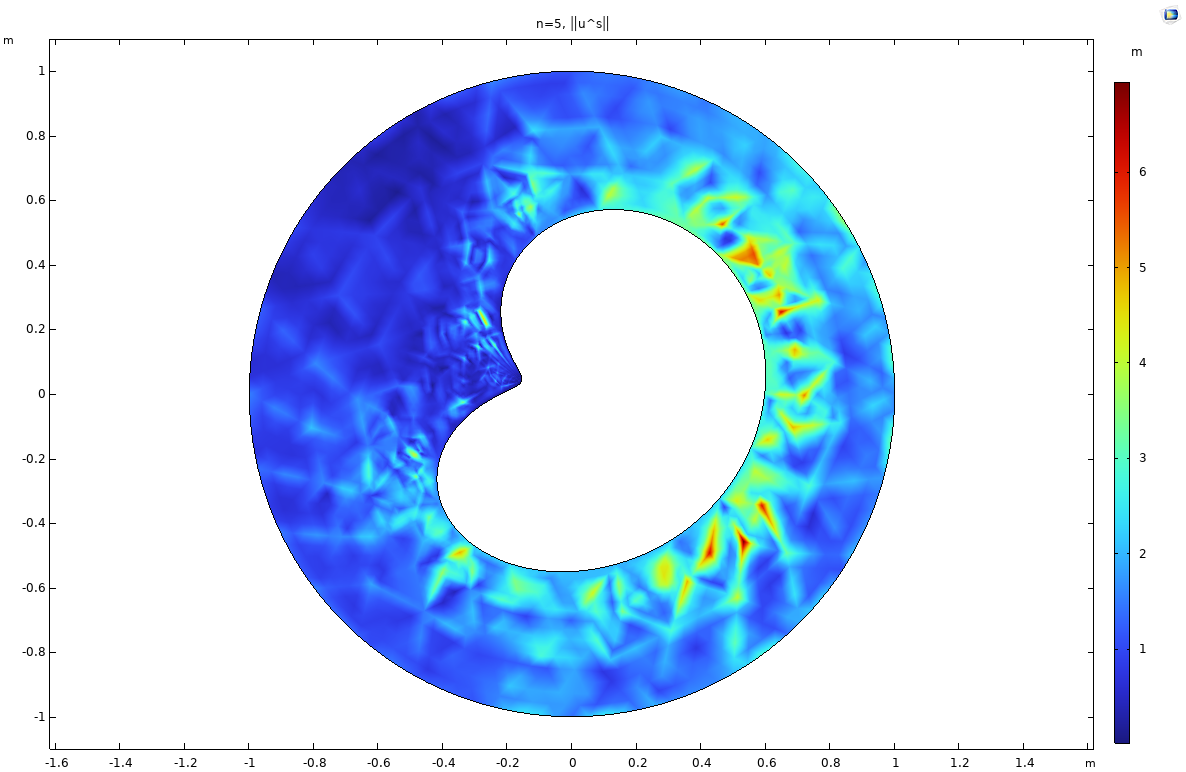} 
		\caption{$n=5$,\ $\|\mathbf{u}^s\|_{L^2\left(B_2 \backslash \overline D\right)^2}$}
	\end{subfigure}
	\hfill
	\begin{subfigure}[t]{0.3\textwidth}
		\centering
		\includegraphics[width=\textwidth]{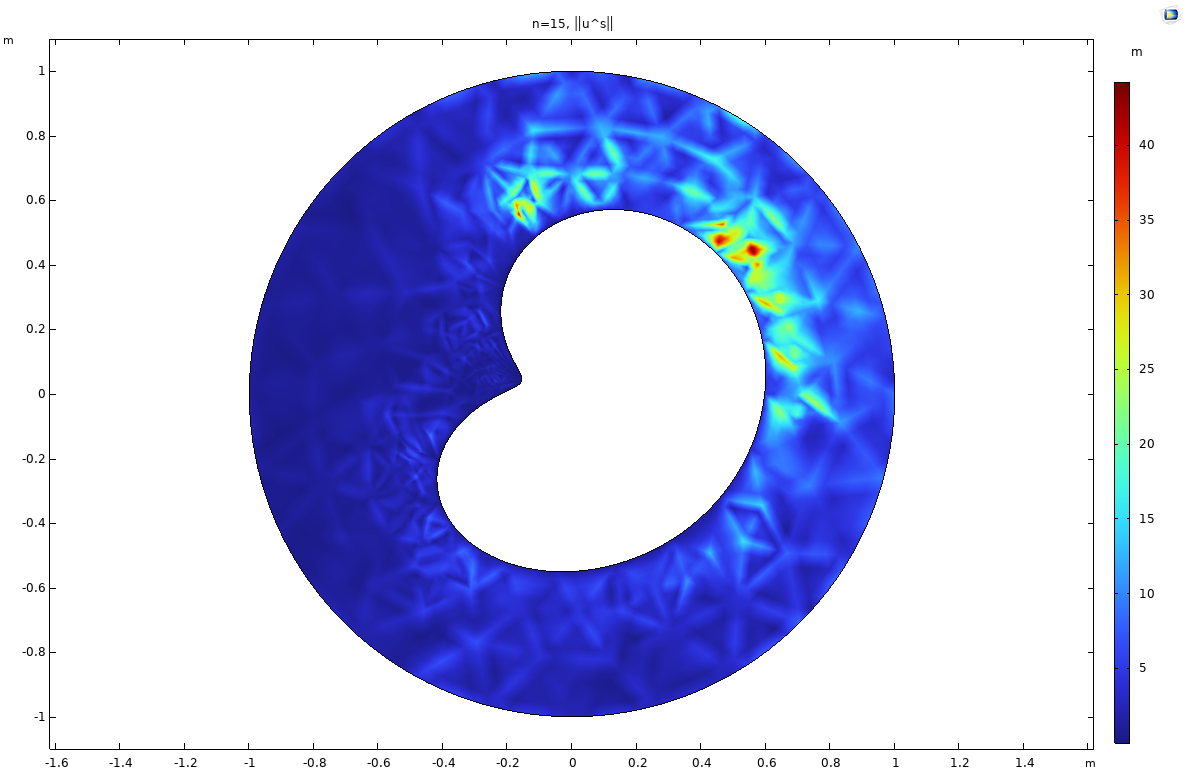}
		\caption{$n=15$,\ $\|\mathbf{u}^s\|_{L^2\left(B_2 \backslash \overline D\right)^2}$}
	\end{subfigure}%
	\hfill%
	\begin{subfigure}[t]{0.3\textwidth}
		\centering
		\includegraphics[width=\textwidth]{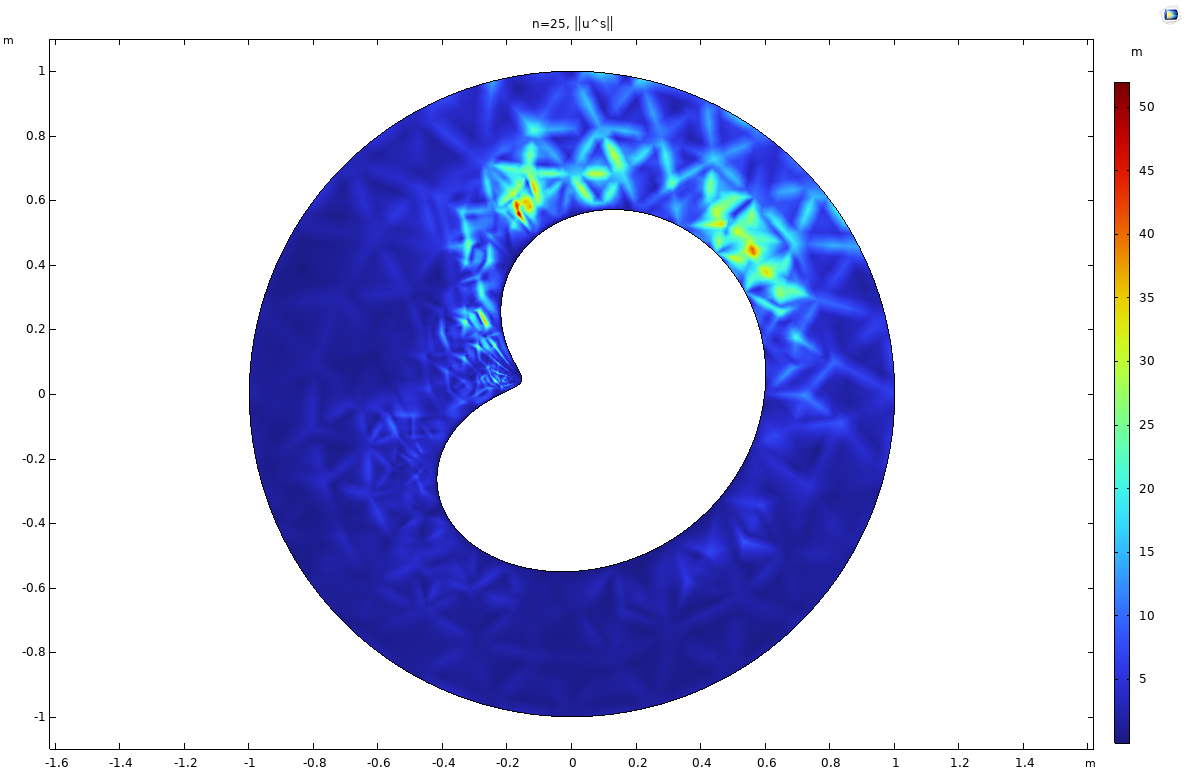}
		\caption{$n=25$,\ $\|\mathbf{u}^s\|_{L^2\left(B_2 \backslash \overline D\right)^2}$}
	\end{subfigure}
	\caption{ $\|\mathbf{u}^s\|_{L^2\left(B_2 \backslash \overline D\right)^2}$ for the incident wave $\mathbf{u}^i$ with different indices $n$ ($n=5,15,25$).}
	\label{fig:15}
\end{figure}

\begin{figure}
	\centering
	\begin{subfigure}[t]{0.3\textwidth}
		\centering
		\includegraphics[width=\textwidth]{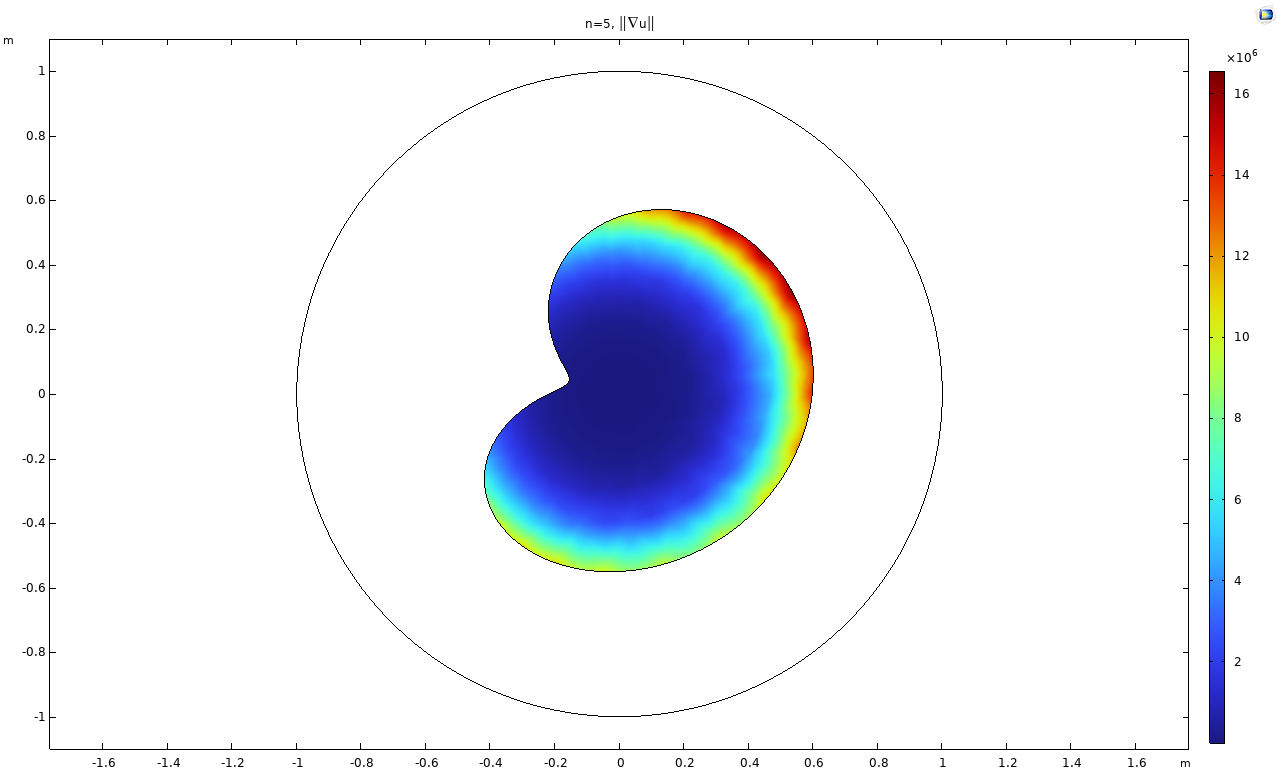} 
		\caption{$n=5$,\ $\| \nabla u\|_{L^2\left(D\right)^2}$}
	\end{subfigure}
	\hfill
	\begin{subfigure}[t]{0.3\textwidth}
		\centering
		\includegraphics[width=\textwidth]{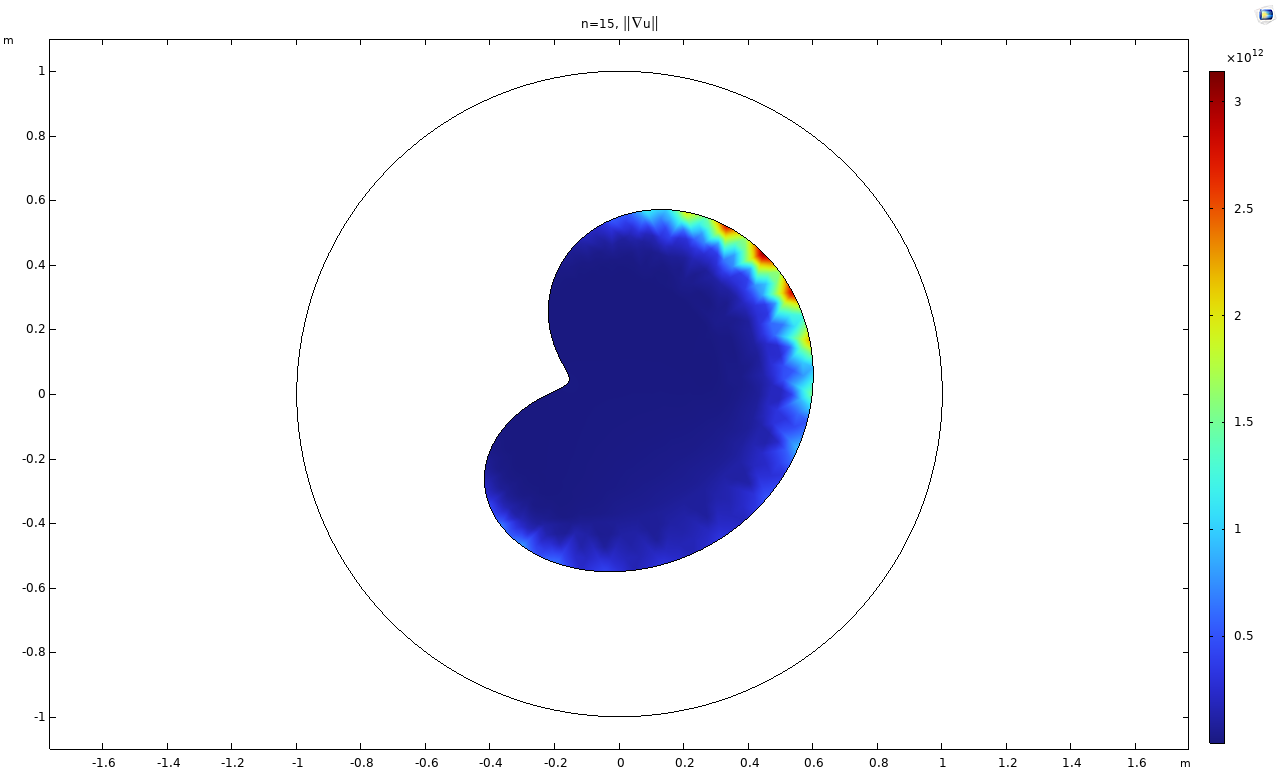}
		\caption{$n=15$,\ $\| \nabla u\|_{L^2\left(D\right)^2}$}
	\end{subfigure}%
	\hfill%
	\begin{subfigure}[t]{0.3\textwidth}
		\centering
		\includegraphics[width=\textwidth]{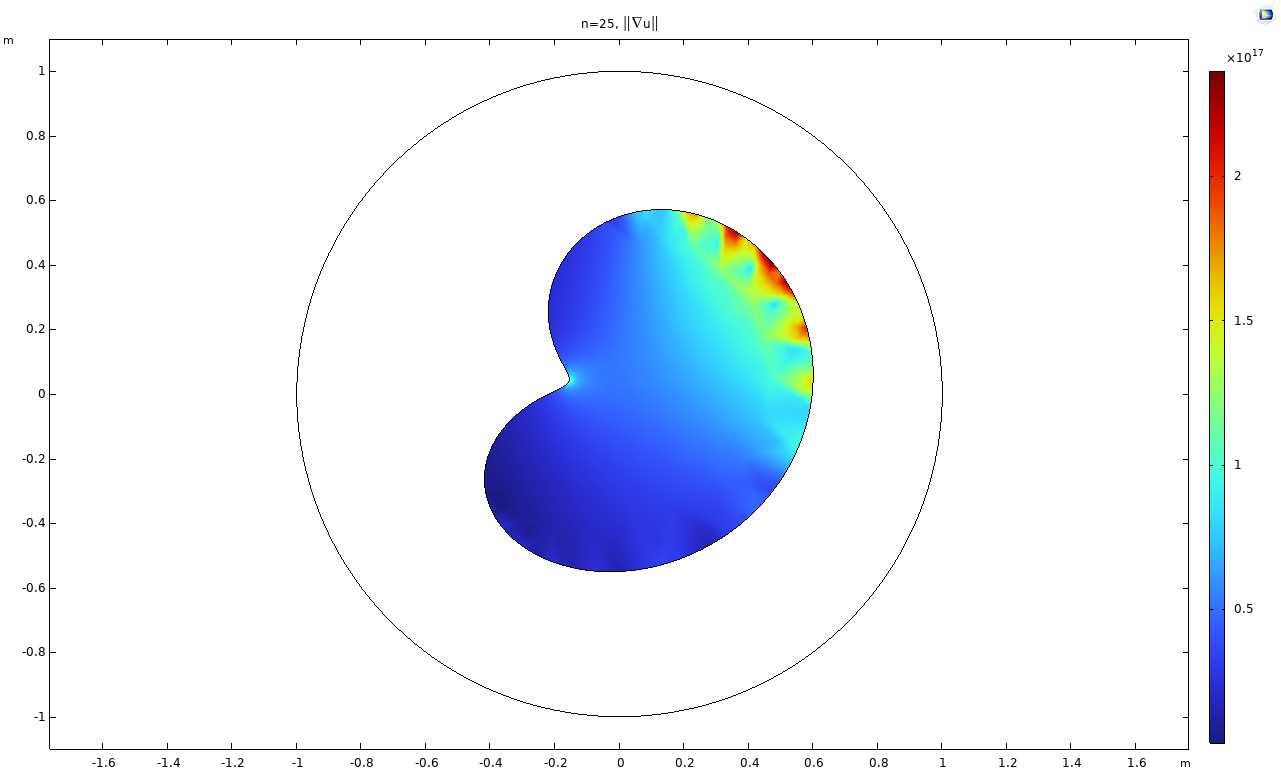}
		\caption{$n=25$,\ $\| \nabla u\|_{L^2\left(D\right)^2}$}
	\end{subfigure}
	\caption{$\| \nabla u \|_{L^2(D)^2}$ for the incident wave $\mathbf{u}^i$ with different indices $n$ ($n=5,15,25$).}
	\label{fig:16}
\end{figure}

\begin{figure}
	\centering
	\begin{subfigure}[t]{0.3\textwidth}
		\centering
		\includegraphics[width=\textwidth]{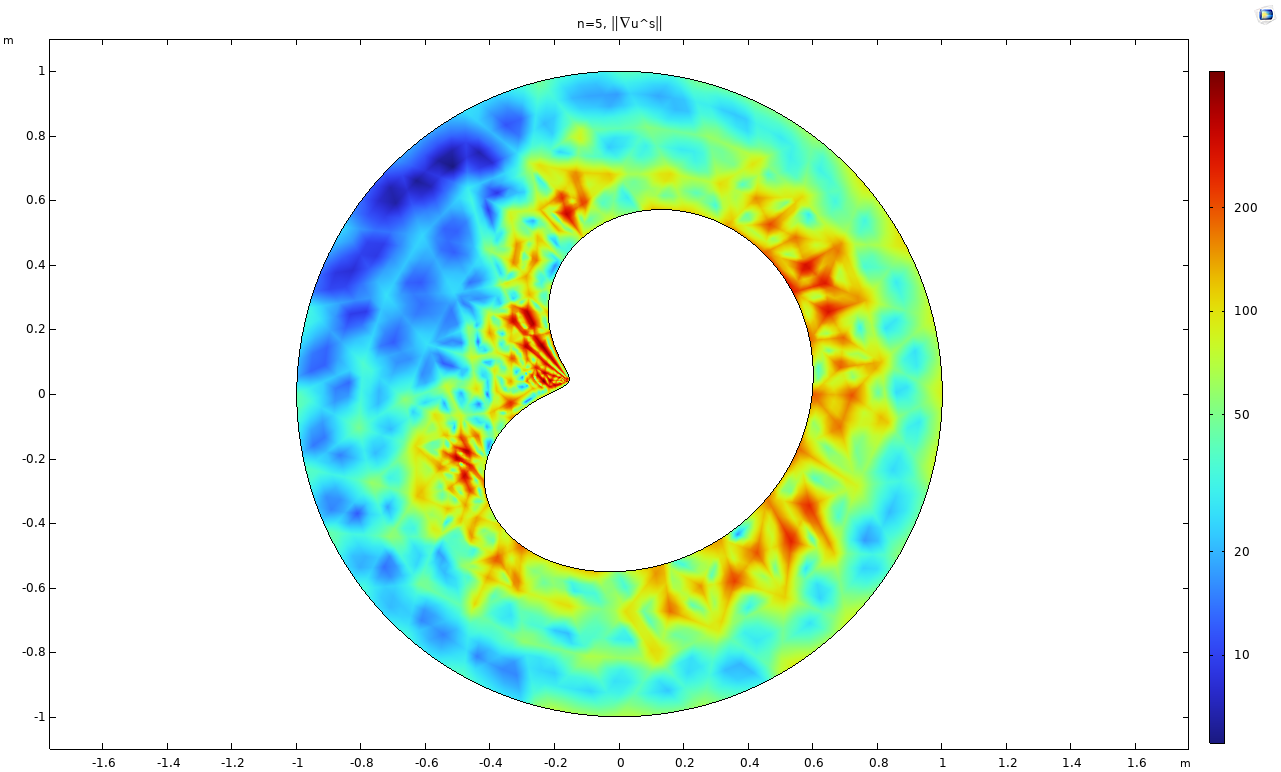} 
		\caption{$n=5$,\ $\| \nabla \mathbf{u}^s\|_{L^2\left(B_2 \backslash \overline D\right)^2}$}
	\end{subfigure}
	\hfill
	\begin{subfigure}[t]{0.3\textwidth}
		\centering
		\includegraphics[width=\textwidth]{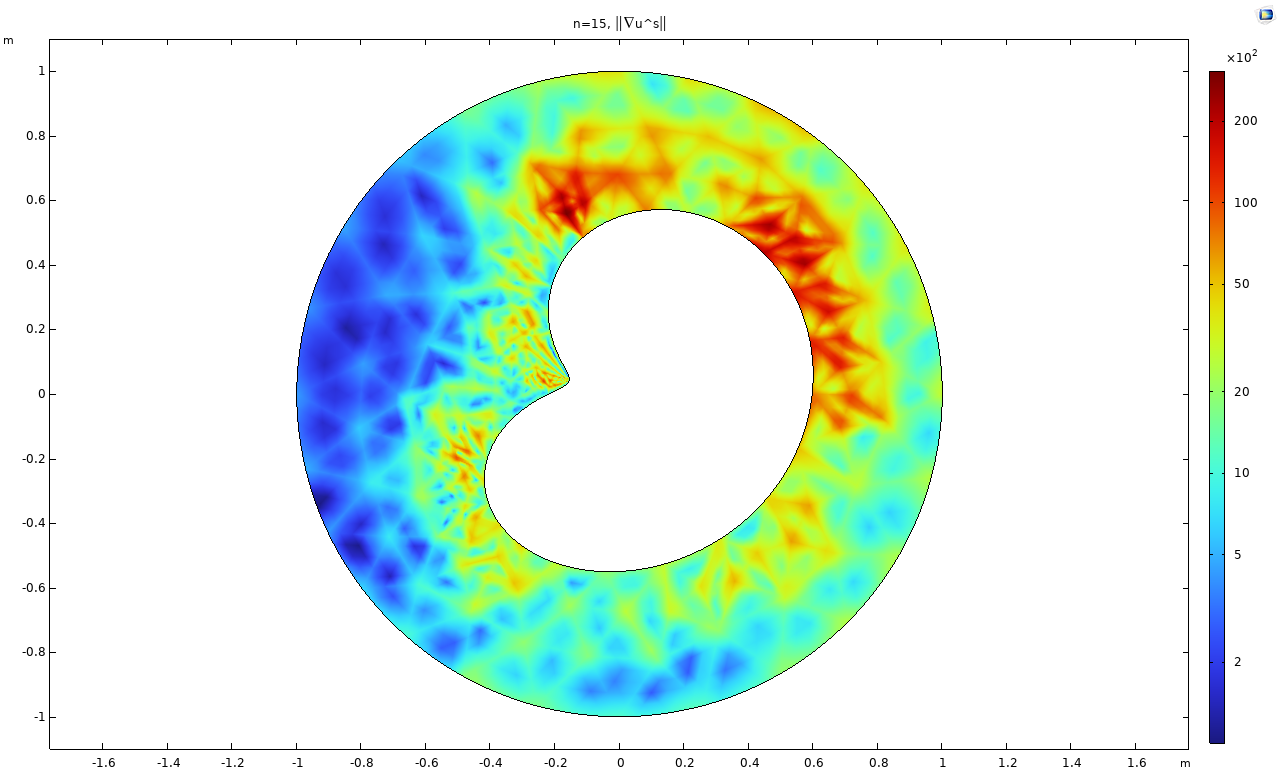}
		\caption{$n=15$,\ $\| \nabla \mathbf{u}^s\|_{L^2\left(B_2 \backslash \overline D\right)^2}$}
	\end{subfigure}%
	\hfill%
	\begin{subfigure}[t]{0.3\textwidth}
		\centering
		\includegraphics[width=\textwidth]{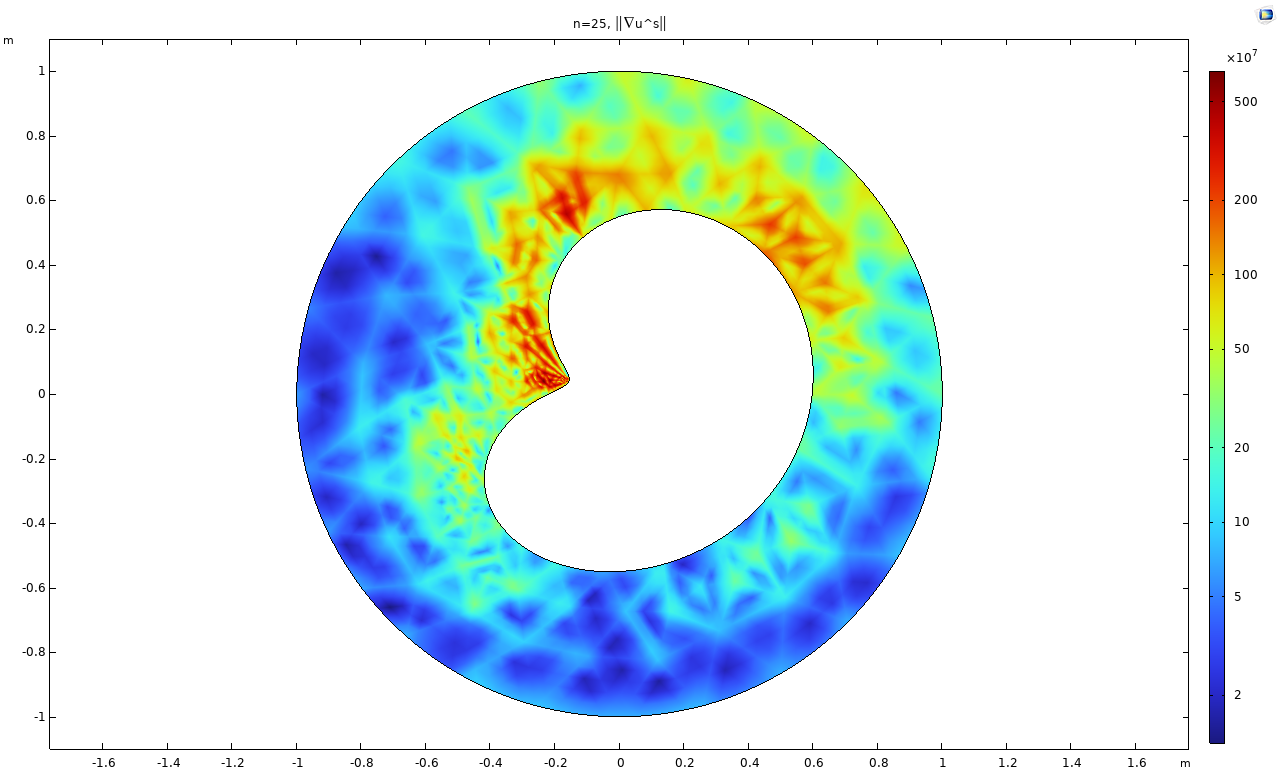}
		\caption{$n=25$,\ $\| \nabla \mathbf{u}^s\|_{L^2\left(B_2 \backslash \overline D\right)^2}$}
	\end{subfigure}
	\caption{$\| \nabla \mathbf{u}^s\|_{L^2\left(B_2 \backslash \overline D\right)^2}$ for the incident wave $\mathbf{u}^i$ with different indices $n$ ($n=5,15,25$).}
	\label{fig:17}
\end{figure}

\subsection{$D$ is a unit ball}

\begin{exm}

In this example, we analyze an air bubble $D$, modeled as a unit sphere, using the physical parameter settings from \eqref{eq:parameter}. We consider the incident wave $\mathbf{u}^i$, defined in \eqref{eq:ui in 3D}, with a frequency $\omega = 0.1$ Hz and incident wave indices $n = 5, 15, 25$. As the index $n$ of the incident wave $\mathbf{u}^i$ increases, the results align with those observed for the unit disk in Examples~\ref{exm:1}--\ref{exm:3}. Figure \ref{fig:19} presents the $L^2$-norm of the interior total field $\mathbf{u}$, confirming this consistency. Additionally, Figure \ref{fig:20} displays the stress $\mathcal{E}(\mathbf{x})$ of the exterior total field $\mathbf{u}|_{B_2 \setminus \overline{D}}$. Figures~\ref{fig:19} and \ref{fig:20} demonstrate that the exterior total field $\mathbf{u}|_{\mathbb{R}^3 \setminus \overline{D}}$ and the interior total field $\mathbf{u}|_D$ exhibit more pronounced behavior near the region of the bubble $D$ in the direction of the vertically incident wave. Furthermore, the energy distribution of $\mathbf{u}|_D$ and $\mathbf{u}|_{\mathbb{R}^3 \setminus \overline{D}}$ displays symmetry and concentrates along the boundary, consistent with the mathematical analysis in Theorems~\ref{thm:th3.1} and \ref{thm:Eu definition in thm}.

Table~\ref{tab:22} reports numerical values of the stress  $E(\mathbf{u})$, as defined in \eqref{eq:Eus th4.2} of Theorem~\ref{thm:Eu definition in thm}, and $E(\mathbf{u}^s)$ for the exterior scattered field, alongside the estimated lower bound given by \eqref{eq:Eus th4.2} as follows
\begin{align}\label{eq:beta57}
\beta(n, \zeta_2, k, \lambda, \mu, \tau) = \frac{n^2 (\zeta_2 - 1) k^2}{27 \zeta_2 (\lambda + 2\mu)^2 \tau^{2n - 2}},
\end{align}
where $\tau = 0.33627$, as specified in \eqref{eq:5.4}. When the index $ n $ of the incident wave $ \mathbf{u}^i $ increases, the values of the stress $ E(\mathbf{u}) $ and $ E(\mathbf{u}^s) $ increase significantly. For $ n = 5 $, $ E(\mathbf{u}) $ is larger than both $ E(\mathbf{u}^s) $ and the lower bound $ \beta(n, \zeta_2, k, \lambda, \mu, \tau) $. For $ n = 15 $ and $ n = 25 $, $ E(\mathbf{u}) $ and $ E(\mathbf{u}^s) $ are equal and both are larger than $ \beta(n, \zeta_2, k, \lambda, \mu, \tau) $, with $ E(\mathbf{u}) $ and $ \beta(n, \zeta_2, k, \lambda, \mu, \tau) $ being of nearly the same order of magnitude. Hence, the lower bound $ \beta(n, \zeta_2, k, \lambda, \mu, \tau) $ is sharp, as confirmed by numerical validation. Furthermore, the contribution of $ E(\mathbf{u}^i) $ to $ E(\mathbf{u}) $ is negligible, which coincides with Remark \ref{rem:41rem}. These results, presented in Table~\ref{tab:22}, align with the analysis in Theorem~\ref{thm:Eu definition in thm}.

\end{exm}

\begin{figure}
	\centering
	\begin{subfigure}[t]{0.3\textwidth}
		\centering
		\includegraphics[width=\textwidth]{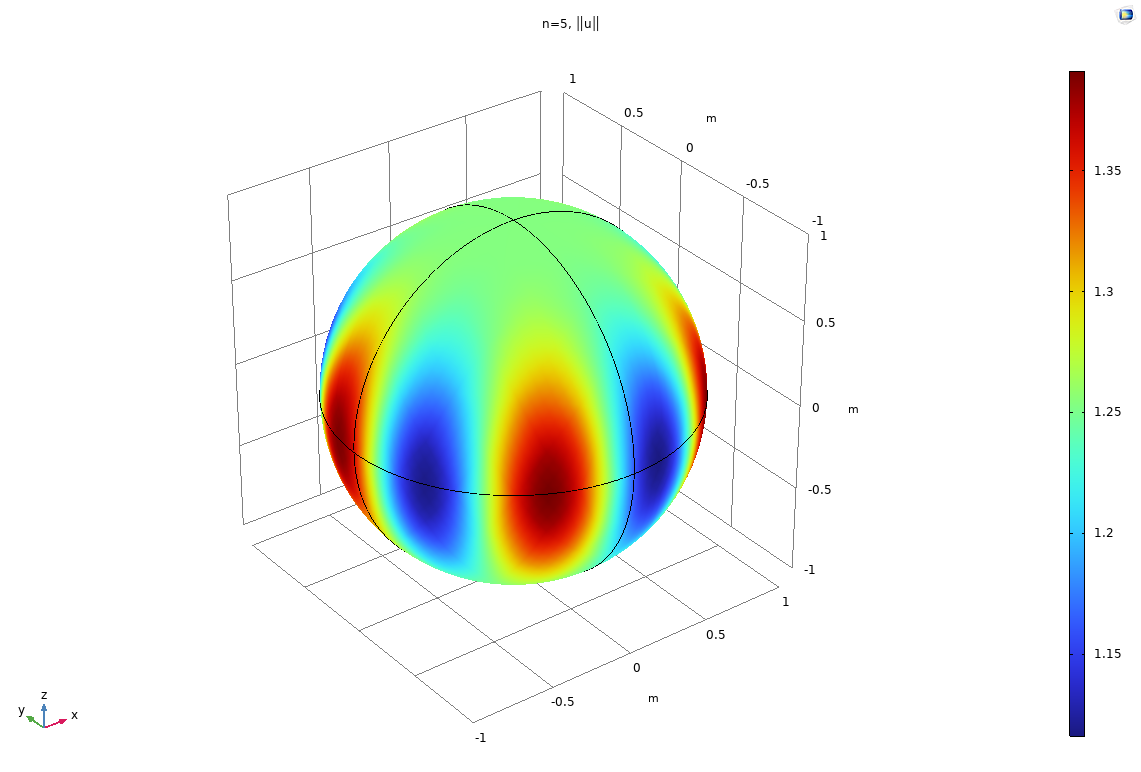} 
		\caption{$n=5$,\ $\|u\|_{L^2(D)}$}
	\end{subfigure}
	\hfill
	\begin{subfigure}[t]{0.3\textwidth}
		\centering
		\includegraphics[width=\textwidth]{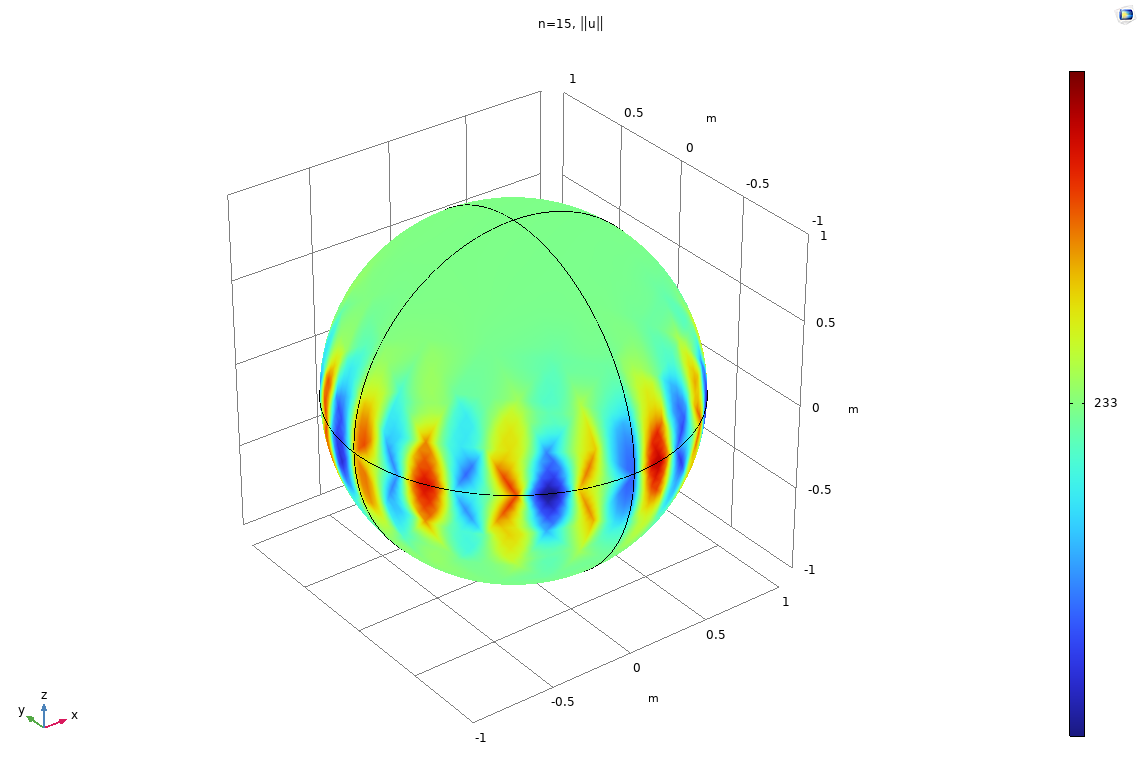}
		\caption{$n=15$,\ $\|u\|_{L^2(D)}$}
	\end{subfigure}%
	\hfill%
	\begin{subfigure}[t]{0.3\textwidth}
		\centering
		\includegraphics[width=\textwidth]{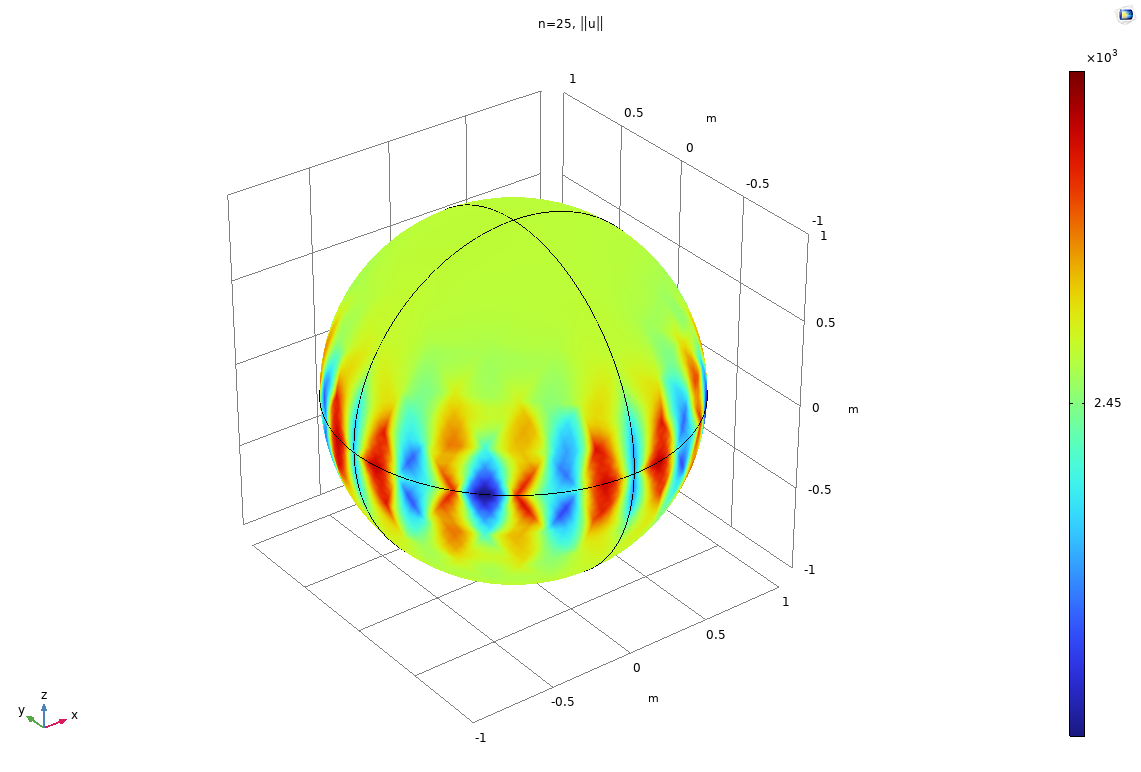}
		\caption{$n=25$,\ $\|u\|_{L^2(D)}$}
	\end{subfigure}
	\caption{$\|u\|_{L^2(D)}$ for the incident wave $\mathbf{u}^i$ with different indices $n$ ($n=5,15,25$).}
	\label{fig:19}
\end{figure}

\begin{figure}
	\centering
	\begin{subfigure}[t]{0.3\textwidth}
		\centering
		\includegraphics[width=\textwidth]{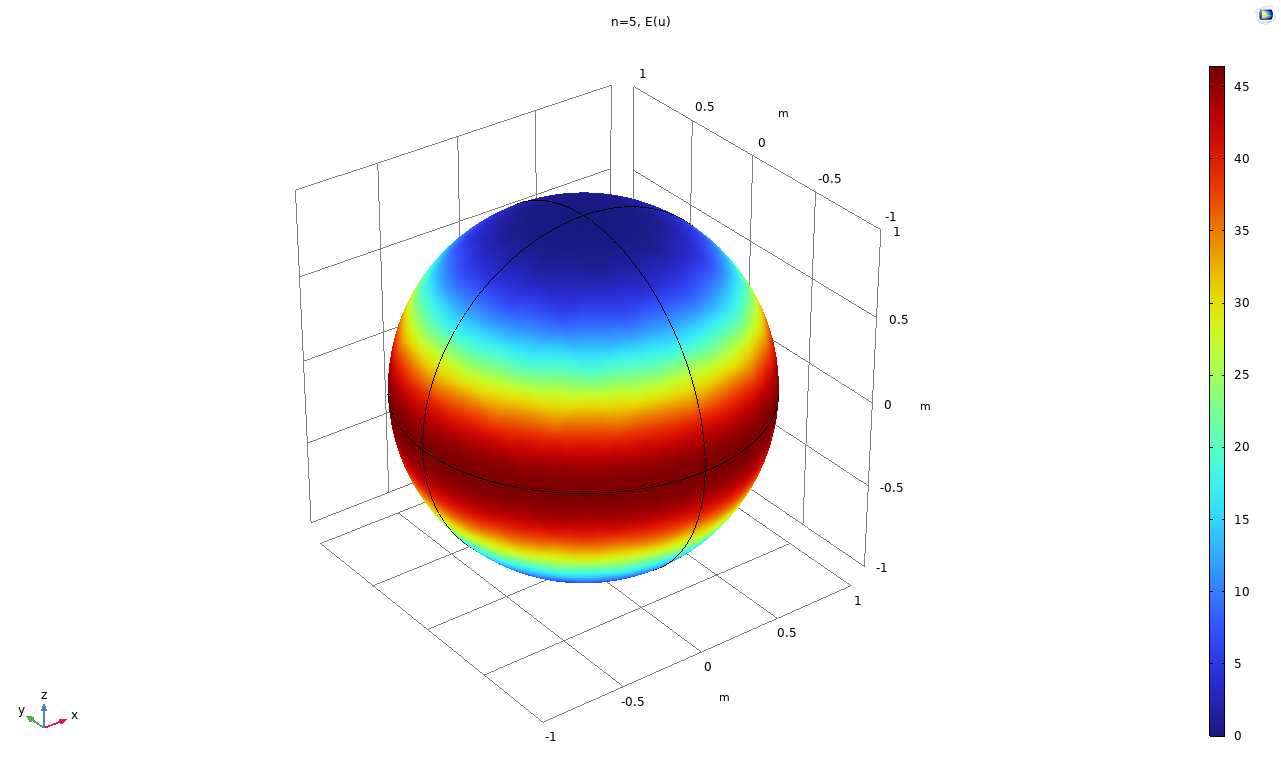} 
		\caption{$n=5$,\ ${\mathcal E}(\mathbf{u})$}
	\end{subfigure}
	\hfill
	\begin{subfigure}[t]{0.3\textwidth}
		\centering
		\includegraphics[width=\textwidth]{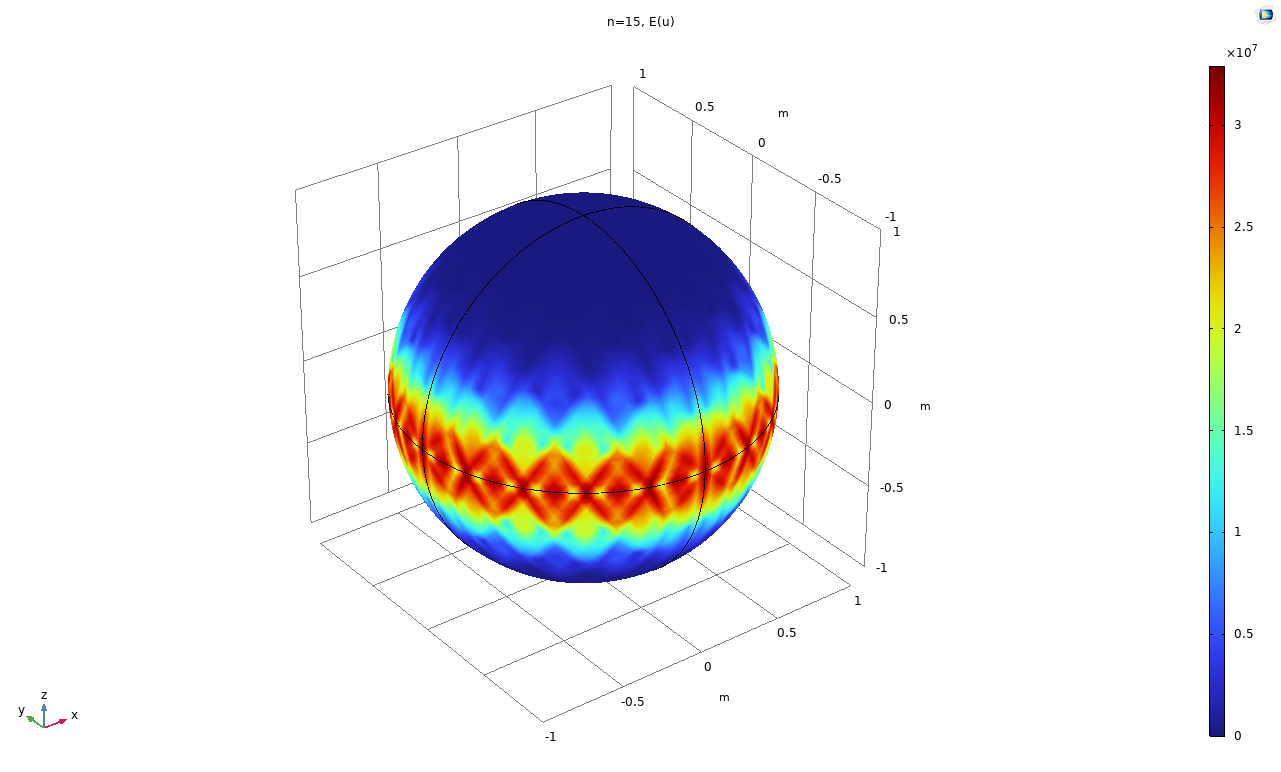}
		\caption{$n=15$,\ ${\mathcal E}(\mathbf{u})$}
	\end{subfigure}%
	\hfill%
	\begin{subfigure}[t]{0.3\textwidth}
		\centering
		\includegraphics[width=\textwidth]{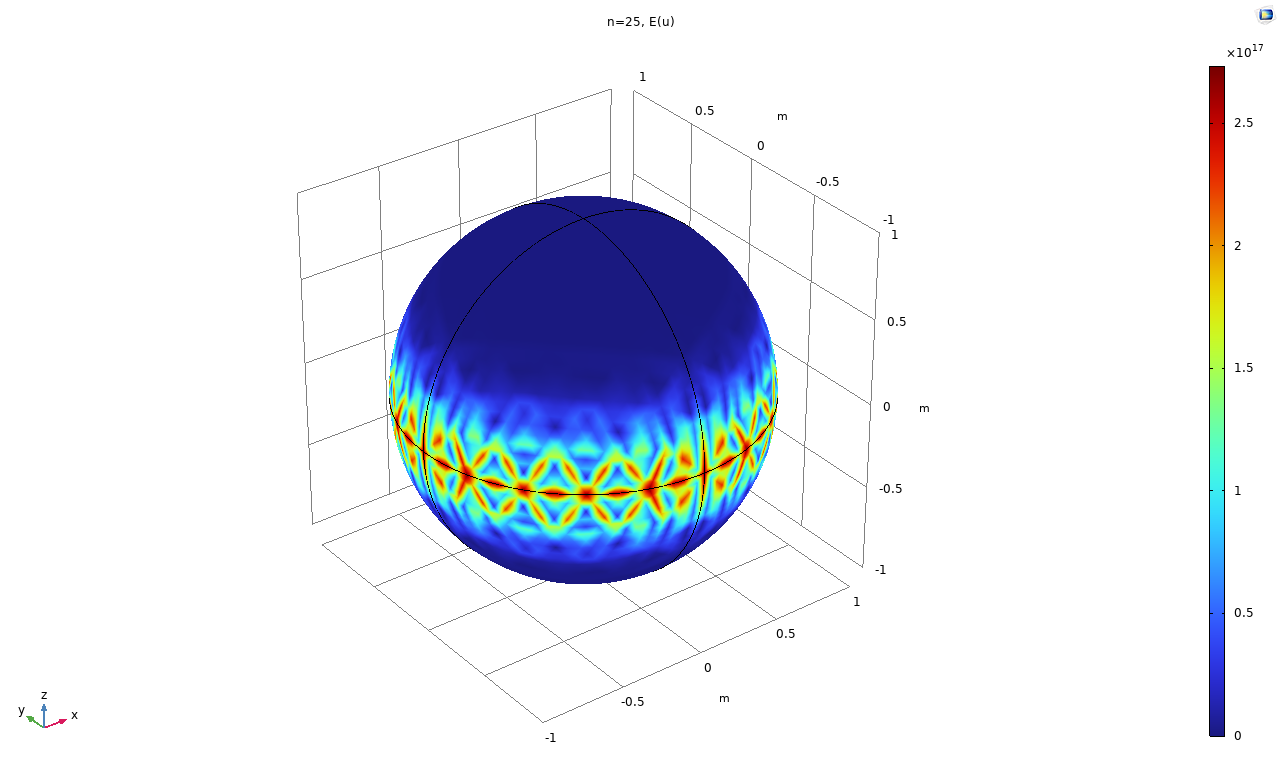}
		\caption{$n=25$,\  ${\mathcal E}(\mathbf{u})$}
	\end{subfigure}
	\caption{ The stress $\mathcal{E}(\mathbf{u})$ of the exterior total field for the incident wave $\mathbf{u}^i$ with indices $n$ ($n=5, 15, 25$)}
	\label{fig:20}
\end{figure}

\begin{table}
	\centering
	\caption{Stress $E(\mathbf{u})$ and  $E(\mathbf{u}^s)$ in $\mathbb{R}^3$ and its lower bound estimate $ \beta(n,\zeta_2,k,\lambda,\mu,\tau)$ for different $n$ ($n=5,15,25$), and given $\zeta_2=1.1$. }
      \label{tab:22}
	\begin{tabular}{@{}cccc@{}}
		\toprule
		$n$ & $E(\mathbf{u})$ & $E(\mathbf{u}^s)$   & $ \beta(n,\zeta_2,k,\lambda,\mu,\tau)$  \\
		\midrule
		5 & $4.7163144808 \times 10^1$ &  $1.7184603683 \times 10^{-3}$  &$4.3753749785 \times 10^{-4}$    \\
		15 & $7.5832965064 \times 10^7$ & $7.5832965064 \times 10^7$   &$1.1521136196 \times 10^7$  \\
		25 & $6.6295335515 \times 10^{17}$ & $6.6295335515 \times 10^{17}$    &  $9.3633300494 \times 10^{16}$  \\
		\bottomrule
	\end{tabular}
\end{table}

	 	\bigskip
	\noindent\textbf{Declaration of competing interest.}
	The authors declare that they have no known competing financial interests or personal relationships that could have appeared to influence the work reported in this paper.

\bigskip

\noindent\textbf{Acknowledgment.}
The work of H. Diao is supported by the National Natural Science Foundation of China  (No. 12371422) and the Fundamental Research Funds for the Central Universities, JLU. The work of H. Liu is supported by the Hong Kong RGC General Research Funds (projects 11311122, 11300821, and 11303125),  the NSFC/RGC Joint Research Fund (project N\_CityU101/21), the France-Hong Kong ANR/RGC Joint Research Grant, A-CityU203/19.

\end{document}